\documentclass[10pt,letterpaper,journal]{IEEEtran}

\usepackage{pdfsync}
\usepackage{microtype}
\usepackage{amsmath, amsthm, amssymb, latexsym} 
\usepackage[only,llbracket,rrbracket]{stmaryrd}
\usepackage{paralist}
\usepackage{xspace}
\usepackage{subcaption}
\usepackage[svgnames,hyperref]{xcolor}
\usepackage{pdflscape}
\usepackage[basic]{complexity}
\usepackage{slashbox}
\usepackage{dsfont}
\usepackage{thmtools,thm-restate}

\usepackage{multibib}
\newcites{appendix}{References for the appendix}

\newtheorem{proposition}{Proposition}
\newtheorem{lemma}{Lemma}

\newtheorem{theorem}{Theorem}
\newtheorem{corollary}{Corollary}
\newtheorem{assumption}{Assumption}

\theoremstyle{remark}
\newtheorem{remark}{Remark}
\newtheorem{example}{Example}
\newtheorem*{runningexample}{Running example}

\DeclareMathOperator{\lcm}{lcm}

\usepackage[textsize=small]{todonotes}

\usepackage{hyperref}
\hypersetup{
    colorlinks=true,        
    linkcolor=DarkBlue,     
    citecolor=DarkBlue,     
    filecolor=DarkMagenta,  
    urlcolor=DarkCyan       
}

\usepackage{tikz}
\usetikzlibrary{arrows,shapes,snakes,automata,backgrounds,fit,positioning,shapes.multipart}

\newcommand{\ignore}[1]{}

\newcommand{\size}[1]{\mathsf{size}({#1})}
\newcommand{\abs}[1]{\left|{#1}\right|}
\newcommand{\card}[1]{\left|{#1}\right|}
\newcommand{\restrict}[2]{#1 \downharpoonright #2}

\newcommand{\MDP}{MDP\xspace} 
\newcommand{\MDPs}{MDPs\xspace}

\newcommand{\tuple}[1]{({#1})}
\newcommand{\goesto}[1]{\stackrel {#1} \longrightarrow}
\newcommand{\set}[1]{\{#1\}}
\newcommand{\Set}[1]{\left\{#1\right\}}
\newcommand{\setof}[2]{\set{{#1}\; |\; {#2}}}
\newcommand{\SetOf}[2]{\Set{\left.{#1}\; \right.\left|\; {#2}\right.}}
\newcommand{\Plays}[1]{\Omega^\omega({#1})}
\newcommand{\Prefs}[2]{\Omega^*(#2, #1)}
\newcommand{\Support}[1]{\mathsf{Supp}({#1})}
\newcommand{\TP}{\mathsf{TP}}
\renewcommand{\MP}{\mathsf{MP}}

\newcommand{\Prob}[3]{\mathbb P_{#1}^{#2}\left[#3\right]}
\renewcommand{\E}[3]{\mathbb E_{#1}^{#2}\left[#3\right]}
\newcommand{\Events}[1]{\Omega_{#1}^\omega}
\newcommand{\Outcome}[2]{\Omega_{#1}^{#2}}
\newcommand{\projection}[2]{\mathsf{proj}^{#1}({#2})}

\newcommand{\ExpSol}[2]{\mathsf{ExpSol}_{#1}({#2})}
\newcommand{\ExpSolP}[2]{\mathsf{ExpSol}^+_{#1}({#2})}
\newcommand{\WCSol}[2]{\mathsf{WCSol}_{#1}({#2})}
\newcommand{\WCSolP}[2]{\mathsf{WCSol}^+_{#1}({#2})}

\newcommand{\ASSolP}[2]{\mathsf{ASSol^+}_{#1}({#2})}

\newcommand{\BWCSolP}[2]{\mathsf{BWCSol}^+_{#1}({#2})}

\newcommand{\BASSolP}[2]{\mathsf{BASSol}^+_{#1}({#2})}

\newcommand{\Sum}{\overrightarrow{\mathsf{Sum}}}

\newcommand{\Graph}{G}

\renewcommand{\Game}{\mathcal G}

\newcommand{\N}{\mathbb N}
\newcommand{\Z}{\mathbb Z}
\newcommand{\Q}{\mathbb Q}
\newcommand{\Qpos}{\Q^{\geq 0}}
\renewcommand{\R}{\mathbb R}
\newcommand{\Rinf}{\R_{\pm\infty}}
\newcommand{\D}{\mathcal D}

\newcommand{\A}{\mathcal A}

\newcommand{\All}[1]{\Delta({#1})}
\newcommand{\Pure}[1]{\Delta_{\textrm P}({#1})}
\newcommand{\FiniteMemory}[1]{\Delta_{\textrm F}({#1})}
\newcommand{\PureFiniteMemory}[1]{\Delta_{\textrm {PF}}({#1})}

\let\oldcoNP\coNP
\let\oldNP\NP

\newcommand{\PTIME}{\P\xspace}
\newcommand{\UPcoUP}{{\UP$\cap$\coUP}\xspace}
\newcommand{\NPcoNP}{{\oldNP$\cap$\oldcoNP}\xspace}
\newcommand{\coNPc}{\oldcoNP-complete\xspace}
\newcommand{\coNPh}{\oldcoNP-hard\xspace}
\renewcommand{\coNP}{\oldcoNP\xspace}
\renewcommand{\NP}{\oldNP\xspace}

\newcommand{\F}[3]{\mathcal F_{#1}\left(#2, #3\right)}
\newcommand{\G}[2]{\mathcal G\left(#1, #2\right)}

\newcommand{\ceiling}[1]{\left\lceil#1\right\rceil}

\renewcommand{\a}{\emph{(a)}}
\renewcommand{\b}{\emph{(b)}}
\newcommand{\ab}{\emph{(a)+(b)}}

\tikzset{
	state/.style={circle,minimum size=6mm,draw=black,fill=gray!20,font=\itshape,inner sep=2pt},
	randomstate/.style={diamond,minimum size=9mm,draw=black,fill=gray!20,font=\itshape,inner sep=2pt},
	playerstate/.style={circle,minimum size=8mm,draw=black,fill=gray!20,font=\itshape,inner sep=2pt},
	initial text=,
	>=stealth',
}

\newcommand{\lorenzo}[1]{\todo[color=blue!20]{#1}}

\title{Multidimensional beyond worst-case and almost-sure problems for mean-payoff objectives}

\author{Lorenzo~Clemente, Universit{\'e} Libre de Bruxelles, Brussels, Belgium.
\thanks{This work was supported by the ERC Starting grant 279499 (inVEST).}
\\
Jean-Fran\c{c}ois~Raskin, Universit{\'e} Libre de Bruxelles, Brussels, Belgium.
}

\begin{document}

\maketitle

\begin{abstract}
	
	The \emph{beyond worst-case threshold problem (BWC)}, recently introduced by Bruy\`ere \emph{et al.},
	asks given a quantitative game graph for the synthesis of a strategy that
	i) enforces some minimal level of performance against any adversary,
	and ii) achieves a good expectation against a stochastic model of the adversary.
	They solved the BWC problem for finite-memory strategies and unidimensional mean-payoff objectives
	and they showed membership of the problem in \NPcoNP.
	They also noted that infinite-memory strategies are more powerful than finite-memory ones,
	but the respective threshold problem was left open.

	We extend these results in several directions.
	First, we consider multidimensional mean-payoff objectives.
	Second, we study both finite-memory and infinite-memory strategies.
	We show that the multidimensional BWC problem is \coNPc in both cases.
	Third, in the special case when the worst-case objective is unidimensional (but the expectation objective is still multidimensional)
	we show that the complexity decreases to \NPcoNP.
	This solves the infinite-memory threshold problem left open by Bruy\`ere \emph{et al.},
	and this complexity cannot be improved without improving the currently known complexity of classical mean-payoff games.
	Finally, we introduce a natural relaxation of the BWC problem,
	the \emph{beyond almost-sure threshold problem (BAS)},
	which asks for the synthesis of a strategy that ensures some minimal level of performance with probability one
	and a good expectation against the stochastic model of the adversary.
	We show that the multidimensional BAS threshold problem is solvable in \PTIME. 	
	
\end{abstract}

\section{Introduction}

\label{sec:intro}

In a two-player mean-payoff game played on a weighted graph~\cite{Liggett:Lippman:1969,EhrenfeuchtMycielski:MeanPayoff:1979}, given a threshold $v \in \Q$, we must decide if there exists a strategy for Player 1 (the controller) to force plays with mean-payoff values larger than $v$, against any strategy of Player 2 (the environment).  In the {\em beyond worst-case threshold problem (BWC)}, recently introduced by Bruy\`ere \emph{et al.} in \cite{Bruyere:Filiot:Randour:Raskin:BWC:2014}, we are additionally given a stochastic model for the {\em nominal}, i.e. expected, behaviour of Player 2. Then we are asked, given two threshold values $\mu, \nu \in \Q$, to decide if there exists a strategy for Player 1 that forces $(i)$ plays with a mean-payoff value larger than $\mu$ against any strategy of Player 2, and $(ii)$ an expected mean-payoff value larger than $\nu$ when Player 2 plays according to the stochastic model of his nominal behaviour.
In the BWC problem, we thus need to solve simultaneously a two player zero-sum game for the worst-case and an optimization problem
where the adversary has been replaced by a stochastic model of his behaviour.

BWC is a natural problem: in practice, we want to build systems that ensure good performances when the environment exhibits his nominal behaviour, and at the same time, that ensure some minimal performances no matter how the environment behaves.  In~\cite{Bruyere:Filiot:Randour:Raskin:BWC:2014}, the BWC problem is solved for {\em finite-memory} strategies and {\em unidimensional} mean-payoff objectives, and shown to be in \NPcoNP.
Also, it is noted there that {\em infinite-memory strategies} are more powerful than finite-memory ones,
and that cannot even be approximated by the latter
(already in the unidimensional case; cf. \cite[Fig. 6]{BruyereFiliotRandourRaskin:BWC:arxiv} for an example).
The respective threshold problem was left unsolved.
We extend here these results in several directions.

\subsection{Contributions}

Our contributions are as follows.
First, we consider {\em $d$-dimensional} mean-payoff objectives. Multiple dimensions are useful to model systems with {\em multiple objectives} that are potentially conflicting, and to analyze the possible trade-offs.  For example, we may want to synthesize strategies that ensure a good QoS while keeping the energy consumption as low as possible.
This extends the BWC problem with one additional level of conflicting trade-offs, which makes the analysis substantially harder.
Second, we study both finite-memory and infinite-memory strategies.
We show that the multidimensional BWC problem is \coNPc in both cases,
and so not more expensive than the plain multidimensional mean-payoff games. 
This is obtained as a \coNP reduction to the solution of a linear system of inequalities of polynomial size.
Correctness follows from non-trivial approximations results for finite/infinite-memory strategies inside end-components%
\footnote{Sub-MDPs which are strongly connected and closed w.r.t. the stochastic transitions.}.
While in the unidimensional case optimal values for the expectation can always be achieved precisely (already by memoryless strategies),
in our multidimensional setting this is not true anymore.
To overcome this difficulty, we are able to show that achievable vectors can be approximated with arbitrary precision,
which is sufficient for our analysis.
Third, in the special case when the worst-case objective is unidimensional (but the expectation is still multidimensional),
we show that the complexity decreases to \NPcoNP. This solves with optimal complexity the infinite-memory threshold problem left open in~\cite{Bruyere:Filiot:Randour:Raskin:BWC:2014}.  
Finally, we introduce the {\em beyond almost-sure threshold problem (BAS)} which is a natural relaxation of the BWC problem. The BAS problem asks, given two threshold values $\vec \mu, \vec\nu \in \Q^d$, for the synthesis of a strategy for Player~1 that $(i)$ ensures a mean-payoff larger than $\vec \mu$ almost surely, i.e. {\em with probability one}, and $(ii)$ an expectation larger than $\vec \nu$ against the nominal behaviour of the environment.
This problem has been independently considered (among other generalizations thereof) in \cite{ChatterjeeKomarkovaKretinsky:Unifying:2015}.
We show that the multidimensional BAS threshold problem is solvable in \PTIME.
As in the BWC problem, we reduce to a linear system of inequalities of polynomial size,
but this time the reduction can be done in \PTIME.

\subsection{Related works}

Solutions to the expected unidimensional mean-payoff problem in Markov Decision Processes (MDP) can be found for example in~\cite{Puterman}, it can be solved in \PTIME, and pure memoryless strategies are sufficient to play optimally. The threshold problem for unidimensional mean-payoff games was first studied in~\cite{EhrenfeuchtMycielski:MeanPayoff:1979}, pure memoryless optimal strategies exist for both players, and the associated decision problem can be solved in \NPcoNP. As said above, BWC was introduced in~\cite{Bruyere:Filiot:Randour:Raskin:BWC:2014} but studied only for finite memory strategies and unidimensional payoffs, the decision problem can be solve in \NPcoNP.

Multidimensional mean-payoff games are investigated in \cite{VelnerRabinovich:Noisy:FOSSACS:2011,VelnerChatterjeeDoyenHenzingerRabinovichRaskin:Complexity:ArXiv:2012},
where it is shown that infinite-memory controllers are more powerful than finite-memory ones,
and the finite-memory and general threshold problems are both \coNPc.
The expectation problem for multidimensional mean-payoff MDPs is in \PTIME,
and finite-memory controllers always suffice \cite{BrazdilBrozekChatterjeeForejtKucera:TwoViews:LMCS:2014}.
Moreover, a recent study showed that one can add additional quantitative probability requirements for the mean-payoff to be above a certain threshold (while still optimizing the expectation),
and that the resulting decision problem is \PTIME for the so-called \emph{joint interpretation}
(where the probability threshold is the same for all dimensions),
and exponential for the \emph{conjunction interpretation} (each dimension has a different probability threshold) \cite{ChatterjeeKomarkovaKretinsky:Unifying:2015} (cf. also \cite{RandourRaskinSankur:CAV15}).
In both cases, infinite-memory strategies are required to achieve the desired performance.
Here, we study the \emph{multidimensional mean-payoff BWC threshold problem},
for both finite-memory and arbitrary controllers.
Our BWC threshold problem generalizes both the synthesis problem for multidimensional mean-payoff games and for multidimensional mean-payoff MDPs with no additional cost in worst-case computational complexity.

\ignore{
Motivation: Generalise BWC to the setting where you have an expectation threshold $\nu$ in one component
and an (unrelated) worst-case threshold in another component $\mu$.
This can be answered in our multidimensional framework by solving the BWC problem for worst-case threshold $(\nu, -W)$
and expectation threshold $(\nu, \mu)$, where $W$ is the largest absolute value of any weight in the MDP.
This cannot be solved by the original one-dimensional BWC framework.
}

\subsection{Illustrating example}

\label{sec:illustrating:example}


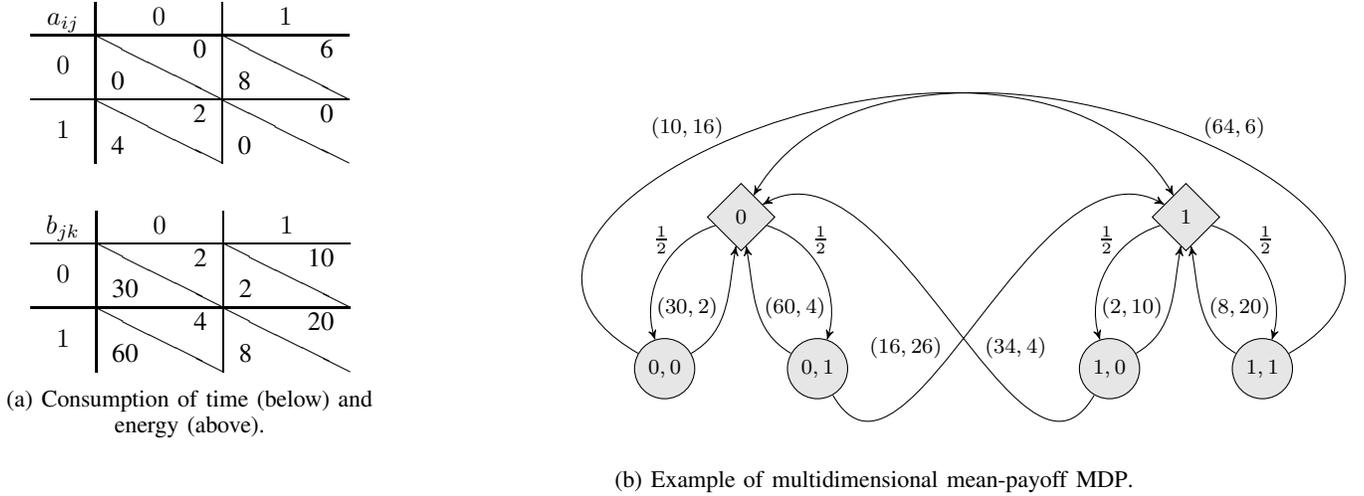
\begin{figure*}
	
	\begin{center}
	
		\begin{minipage}{.28\textwidth}
			\centering
			\begin{tabular}{ c | c | c }
				$a_{ij}$	& $0$						& $1$ 			\\\hline		
				$0$		& \backslashbox{0}{0}	& \backslashbox{8}{6}	\\\hline
				$1$		& \backslashbox{4}{2}	& \backslashbox{0}{0}
			\end{tabular}
			\\[4ex]
			\begin{tabular}{ c | c | c }
				$b_{jk}$	& $0$						& $1$ 			\\\hline		
				$0$		& \backslashbox{30}{2}	& \backslashbox{2}{10}	\\\hline
				$1$		& \backslashbox{60}{4}	& \backslashbox{8}{20}
			\end{tabular}			
			\subcaption{Consumption of time (below) and energy (above).}
			\label{fig:motivational:table}
		\end{minipage}
		\begin{minipage}{.71\textwidth}
			\footnotesize
			\centering
			\begin{tikzpicture} 

				\node[randomstate] (c) [] {$0$};
				\node[playerstate] (ct) [below left = 1.5cm and .5cm of c] {$0, 0$};
				\node[playerstate] (cu) [below right = 1.5cm and .5cm of c] {$0, 1$};
			
				\node[randomstate] (d) [right = 5cm of c] {$1$};
				\node[playerstate] (dt) [below left = 1.5cm and .5cm of d] {$1, 0$};
				\node[playerstate] (du) [below right = 1.5cm and .5cm of d] {$1, 1$};
			
				\path[->] (c) edge [bend right = 45] node [above = 1ex] {$\frac 1 2$} (ct);
				\path[->] (c) edge [bend left = 45] node [above = 1ex] {$\frac 1 2$} (cu);

				\path[->] (ct) edge [out = 30, in = -100] node [left = 0ex] {$(30, 2)$} (c);
				\path[->] (cu) edge [out = 150, in = -80] node [right = 0ex] {$(60, 4)$} (c);

				\path[->] (ct) edge [out = 150, in = 120, min distance = 30ex] node [left = 3ex] {$(10, 16)$} (d);
				\path[->] (cu) edge [out = -60, in = 150] node [below left = 1ex and 3ex] {$(16, 26)$} (d);

				\path[->] (d) edge [bend right = 45] node [above = 1ex] {$\frac 1 2$} (dt);
				\path[->] (d) edge [bend left = 45] node [above = 1ex] {$\frac 1 2$} (du);

				\path[->] (dt) edge [out = 30, in = -100] node [left = 0ex] {$(2, 10)$} (d);
				\path[->] (du) edge [out = 150, in = -80] node [right = 0ex] {$(8, 20)$} (d);

				\path[->] (dt) edge [out = -120, in = 30] node [below right = 1ex and 3ex] {$(34, 4)$} (c);
				\path[->] (du) edge [out = 30, in = 60, min distance = 30ex] node [right = 3ex] {$(64, 6)$} (c);

			\end{tikzpicture}
			\vspace{-10ex}
			\subcaption{Example of multidimensional mean-payoff \MDP.}
			\label{fig:motivational_example_MDP}
		\end{minipage}
	\end{center}
	\caption{Illustrating example.}
\end{figure*}


Consider the following task system \cite{ZwickPaterson:MeanPayoff:TCS:1996}:
There are two configurations (0 and 1),
and at each interaction between the controller and its environment, 
one new instance of two kind of tasks can be generated (0 and 1).
The two tasks are generated with equal probability $1/2$ in the nominal behavior of the environment.
Before serving pending task $k \in \set{0, 1}$,
the system may decide to go from configuration $i$ to configuration $j$ at cost $a_{ij}$, for $i, j \in \set{0, 1}$,
and then it has to serve the pending task $k$ from the new configuration $j$ at cost $b_{jk}$.
Thus, the total cost is $a_{ij} + b_{jk}$.
Costs are bidimensional: Each cost specifies an amount of time and energy;
the actual parameters are shown in Fig.~\ref{fig:motivational:table}.
For example, from configuration $0$, task $0$ takes $30$ time units to complete and it consumes $2$ energy units,
while from the other configuration the same task takes $2$ time units and $10$ energy units.
%
We are interested in synthesizing controllers that optimize the expected/worst-case mean (i.e., per task) time and energy.
There are trade-offs between the two measures:
If the controller decides to serve a task quickly then the system consumes a large amount of energy, and vice versa.
To analyze this example, we rephrase it as the multidimensional mean-payoff \MDP depicted in Fig.~\ref{fig:motivational_example_MDP}.
For example, state $0$ represents the fact that the system is in configuration $0$ waiting for a task to arrive,
while in state $(0, 0)$ a task of the first type has arrived,
and the controller needs to decide whether to serve it from the same configuration, or go to configuration $1$.
The objective of the controller is to guarantee worst-case mean time $24$ under all circumstances, in that case the probabilities in the MDP are ignored and the probabilistic choice is replaced by an adversarial choice (we have thus a two-player zero sum game). Additionally, with the same strategy for the controller, we want to minimize the expected mean energy consumption in the nominal behaviour of the controller given by the stochastic model.
If the controller decides to always serve tasks from configuration $0$,
then it ensures an expected mean energy consumption of $3$,
but under this strategy, the worst-case mean time is $60$,
which does not meet our worst-case objective of $24$.
A strategy for the controller that is good both for the worst-case and for the expectation
can be obtained as follows: For two parameters $\alpha, \beta \in \N$,
stay in configuration $0$ for $\alpha$ consecutive tasks,
then move to configuration $1$ for $\beta$ tasks, and then repeat.
This ensures worst-case time
$
	\frac {\alpha - 1} {\alpha + \beta} 60 + \frac 1 {\alpha + \beta} 64 +
	\frac {\beta - 1} {\alpha + \beta} 8 + \frac 1 {\alpha + \beta} 16
$
and expected energy
$
	\frac {\alpha - 1} {\alpha + \beta} (\frac 1 2 2 + \frac 1 2 4) + \frac 1 {\alpha + \beta} (\frac 1 2 4 + \frac 1 2 6) +
	\frac {\beta - 1} {\alpha + \beta} (\frac 1 2 10 + \frac 1 2 20) + \frac 1 {\alpha + \beta} (\frac 1 2 16 + \frac 1 2 26)	
$.
By taking $\alpha = 1$ and $\beta = 3$,
we obtain worst-case time $24$ (thus meeting the requirement) and expected energy $14$.
Note the trade-off: To ensure a stronger guarantee on the mean time,
we had to sacrifice the expected mean energy.

In this paper we address the problem of deciding the existence of controllers
ensuring a worst-case (or almost-sure) threshold, while, at the same time,
achieving a usually better expectation threshold under the nominal behavior of the environment.
We consider the class of multidimensional mean-payoff objectives.

\subsection{Structure of the paper}

In Sec.~\ref{sec:preliminaries}, we present the preliminaries that are necessary to define the BWC and BAS problems.
In Sec.~\ref{sec:BWC}, we solve the BWC problem both for finite and infinite memory strategies.
In Sec.~\ref{sec:BAS}, we solve the BAS problem and show that finite memory strategies are sufficient to achieve the BAS threshold problem.
Finally, in Sec.\ref{sec:conclusions} we conclude with some final remarks.
Full proofs can be found in the appendices \ref{app:preliminaries} and \ref{app:BWC}.

\section{Preliminaries}

\label{sec:preliminaries}

Let $\N$, $\Q$, and $\R$ be the set of natural, rational, and real numbers, respectively,
and let $\Rinf = \R \cup \set{+\infty, -\infty}$.
For two vectors $\vec\mu$ and $\vec\nu$ of the same dimension
and a comparison operator $\sim \in \set{\leq, <, >, \geq}$,
we write $\vec\mu \sim \vec\nu$ for the component-wise application of $\sim$.
In particular, $\vec\mu > \vec 0$ means that \emph{every component} of $\vec\mu$ is strictly positive.
A \emph{probability distribution} on $A$ is a function
$R : A \to \Qpos$ s.t. $\sum_{a \in A} R(a) = 1$.
The \emph{support} of $R$ is $\Support R = \setof {a \in A} {R(a) > 0}$.
Let $\D(A)$ be the set of probability distributions on $A$.

\ignore{
The \emph{size} 
of an integer number $p \in \Z$,
of a rational number $r = p / q \in \Q$ with $p$ and $q$ two co-prime integers,
of a vector $\vec r = (r_1, \dots, r_k)$ of $k$ integer or rational numbers,
of a matrix $A = (r_{ij})_{1 \leq i \leq m, 1 \leq j \leq n}$ of $m \cdot n$ integer or rational numbers,
of a system of inequalities $A \cdot \vec x \leq \vec b$,
or of a finite distribution $R : A \to \Qpos$ on a set $A$ of cardinality $n$,
is defined as follows:
\begin{align}
	\size p &= 1 + \left\lceil\lg(\abs p + 1)\right\rceil \\
	\size r &= \size p + \size q \\
	\size {\vec r} &= \size {r_1} + \cdots + \size {r_k} \\
	\size A &= \sum_{i = 1}^m \sum_{j = 1}^n \size {r_{ij}} \\
	\size {A \cdot \vec x \leq \vec b} &= \size A + \size {\vec b} \\
	\size R &= n \cdot \max_{a \in A} \size {R(a)}
\end{align}
Other reasonable notions of size can be used,
but they usually are linearly equivalent to those given here \cite{Schrijver}.
}

\subsection{Weighted graphs}

A \emph{multi-weighted graph} is a tuple $\Graph = \tuple {d, S, E, w}$,
where $d \geq 1$ is the dimension,
$S$ is a finite set of states,
$E \subseteq S  \times S$ is the set of directed edges,
and $w: E \to \Z^d$ is a function assigning to each edge a weight vector.
When $d = 1$, we refer to $\Graph$ just as a weighted graph.
With $\vec x[i]$ we denote the $i$-th component of a vector $\vec x$.
For a state $s \in S$, let $E(s) = \setof t {(s, t) \in E}$ be its set of \emph{successors}.
We assume that each state $s$ has at least one successor.
Let $W$ be the largest absolute value of a weight appearing in the graph.

A \emph{play} in $\Graph$ 
is an infinite sequence of states $\pi = s_0 s_1 \cdots$ s.t. $(s_i, s_{i+1}) \in E$ for every $i \geq 0$.
Let $\Events {s_0} (\Graph)$ be the set of plays in $\Graph$ starting at $s_0$,
and let $\Plays \Graph$ be the set of all plays of $\Graph$.
When $\Graph$ is clear from the context, we omit it.
The \emph{prefix} of length $n$ of a play $\pi = s_0 s_1 \cdots$ is the finite sequence $\pi(n) = s_0 s_1 \cdots s_{n-1}$.
For a set of states $T\subseteq S$,
Let $\Prefs T \Graph$ be the set of prefixes of plays in $\Graph$ ending in a state $s_{n-1} \in T$.

The \emph{total payoff} and \emph{mean payoff} up to length $n$
of a play $\pi = s_0 s_1 \cdots$ (or prefix of length at least $n$)
are defined as $\TP_n(\pi) = \sum_{i=0}^{n-1} w(s_i, s_{i+1})$
and $\MP_n(\pi) = \frac 1 n \TP_n(\pi)$, respectively.
The (lim-inf) total and mean payoffs on an infinite play $\pi$ are then defined as
$\TP(\pi) := \liminf_{n \to \infty} \TP_n(\pi)$ and
$\MP(\pi) := \liminf_{n \to \infty} \MP_n(\pi)$.


\subsection{Markov decision processes}

A \emph{Markov decision process}, or \MDP,
is a tuple $\Game = \tuple {\Graph, S^C, S^R, R}$,
where $\Graph = (d, S, E, w)$ is a multi-weighted graph,
$\set {S^C, S^R}$ is a partition of $S$ into states belonging to either the Controller player or to the Random player, respectively,
and $R : S^R \to \D(S)$ is a function assigning a distribution over $S$ to states belonging to Random
s.t., for every $s \in S^R$, $\Support{R(s)} = E(s)$.
We do not allow $R(s)$ to assign probability zero to any successor of $s$.%
\footnote{This restriction will simplify the presentation.
It is not present in \cite{Bruyere:Filiot:Randour:Raskin:BWC:2014}, but it can be easily lifted.}
Let $Q$ be the largest denominator used to represent probabilities in $R$.
We use $Q$ as a measure of complexity for representing $R$.
%

%
%
%

In order to discuss the complexity of strategies for Controller,
we represent them as \emph{stochastic Moore machines}.
A strategy for a MDP $\Game = \tuple {\Graph, S^C, S^R, R}$ is a tuple $f = (M, \alpha, f_u, f_o)$,
consisting of a set of memory states $M$,
the initial memory distribution $\alpha \in \D(M)$,
the stochastic memory update function $f_u : S \times M \to \D(M)$,
and the stochastic output function $f_o : S^C \times M \to \D(S)$,
where $\Support{f_o(s, m)} \subseteq E(s)$ for every $s \in S^C$ and $m \in M$.
The update function is extended to sequences $f_u^* : S^* \to \D(M)$
inductively as $f_u^*(\varepsilon) = \alpha$
and $f_u^*(\pi s)(m') = \sum_{m \in M}f_u^*(\pi)(m) f_u(s, m)(m')$.
The output function on sequences $f_o^* : S^*S^C \to \D(S)$
is defined as $f_o^*(\pi s)(s') = \sum_{m \in M}f_u^*(\pi)(m) f_o(s, m)(s')$.
A play $\pi = s_0 s_1 \cdots$ is \emph{consistent} with a Controller's strategy $f$
if, and only if, for every $i$ s.t. $s_i \in S^C$, we have $s_{i+1} \in \Support {f_o^*(s_0s_1\cdots s_i)}$.
Given a state $s_0$ and a Controller's strategy $f$,
the set of \emph{outcomes} $\Outcome {s_0} {f} (\Graph)$
is the set of plays starting at $s_0$ which are consistent with $f$.

A strategy $f$ is \emph{pure} iff $\Support {f_u(s, m)}$ and $\Support {f_o(s, m)}$ are both singletons.
A strategy $f$ is \emph{memoryless} iff $\card M = 1$,
\emph{finite-memory} iff $\card M < \infty$,
and \emph{infinite-memory} iff $\card M = \infty$.
Let $\All \Game$, $\Pure \Game$,
$\FiniteMemory \Game$, and $\PureFiniteMemory \Game$
be the sets of all, resp., pure, finite-memory, and pure finite-memory strategies.
%

\subsection{Markov chains}

A \emph{Markov chain} is an \MDP where no state belongs to Controller,
i.e., $S^C = \emptyset$,
and in this case we just write $\Game = \tuple {\Graph, R}$.
An \emph{event} is a measurable set of plays $A \subseteq \Plays \Graph$.
Given a state $s_0$ and an event $A \subseteq \Events {s_0} (\Graph)$,
let $\Prob {s_0} \Game A$ be the probability that a play starting in $s_0$ belongs to $A$,
which exists and it is unique by Carath\'eodory's extension theorem \cite{Billingsley}.
An even is \emph{almost sure} if it has probability $1$.
For a measurable payoff function $v : \Plays\Graph \to \Rinf^d$,
let $\E {s_0} \Game v$ be the expected value of $v$ of a play starting in $s_0$.

A Markov chain $\Game$ is \emph{unichain} if it contains exactly one bottom strongly connected component (BSCC).
Therefore, if $\Game$ is unichain, then all states in its unique BSCC are visited infinitely often almost surely,
and the mean payoff equals its expected value almost surely.

Given a \MDP $\Game = \tuple {\Graph, S^C, S^R, R}$
and a strategy $f$ for Controller represented as the stochastic Moore machine $(M, \alpha, f_u, f_o)$,
let the \emph{induced Markov chain} be $\Game[f] = \tuple{\Graph', R'}$,
where $\Graph' = \tuple {d, S \times M, w', E'}$
with $((s, m), (s', m')) \in E'$ iff $(s, s') \in E$,
$m' \in \Support{f_u(s, m)}$,
and $s' \in \Support{f_o(s, m)}$ whenever $s \in S^C$,
$w'((s, m), (s', m')) = w(s, s')$ for every $((s, m), (s', m')) \in E'$,
$R'(s, m)(s', m') = R(s)(s') \cdot f_u(s, m)(m')$ for every $s \in S^R$,
and $R'(s, m)(s', m') = f_o(s, m)(s') \cdot f_u(s, m)(m')$ for every $s \in S^C$.
Note that $\Game[f]$ is finite iff $f$ is finite-memory.
By a slight abuse of terminology, we say that a strategy $f$ is \emph{unichain} if $\Game[f]$ is unichain.
Plays in $\Game[f]$ can be mapped to plays in $\Game$
by a projection operator $\projection {} \cdot : \Plays{\Graph'} \to \Plays\Graph$
which discards the memory of $f$. 
Given a state $s_0$, a Controller's strategy $f$,
and an event $A \subseteq \Events {s_0}$,
let $\Prob {s_0, f} \Game A := \Prob {s_0} {\Game[f]} {\projection {-1} A}$.
For a measurable payoff function $v : \Plays\Graph \to \Rinf^d$,
let $\E {s_0, f} \Game v := \E {s_0} {\Game[f]} {v'}$,
where $v'(\pi) := v(\projection {} \pi)$.

\subsection{End-components}

A \emph{end-component} (EC) of a \MDP $\Game$
is a set of states $U \subseteq S$ s.t.
\begin{inparaenum}[a)]
	\item the induced sub-graph $\tuple {U, E \cap U \times U}$ is strongly-connected, and
	\item for any stochastic state $s \in U \cap S^R$, $E(s) \subseteq U$.
\end{inparaenum}
Thus, Controller can surely keep the game inside an EC,
and almost surely visits all states therein.
For an end-component $U$ of $\Game$,
we denote by $\restrict \Game U$ the MDP obtained by restricting $\Game$ to $U$ in the natural way.
%
ECs are central in the analysis of MDPs thanks to the following result.
\begin{proposition}[cf. \cite{DeAlfaro:PhD:1998}]
	\label{prop:EC}
	For any Controller's strategy $f \in \All \Game$,
	the set of states visited infinitely often when playing according to $f$ is almost surely an EC.
\end{proposition}

\subsection{Expected-value objective}

For a \MDP $\Game$, a starting state $s_0$, and Controller's strategy $f \in \All \Game$,
the set of \emph{expected-value achievable solutions for $f$} is
$\ExpSolP \Game {s_0, f} = \setof {\vec \nu \in \R^d} {\E {s_0, f} \Game \MP > \vec \nu}$,
i.e., it is the set of vectors $\vec \nu$ s.t. Controller can guarantee an expected mean payoff $> \vec \nu$ from state $s_0$ by playing $f$.
The set of \emph{expected-value achievable solutions} is
$\ExpSolP \Game {s_0} = \bigcup_{f \in \All \Game} \ExpSolP \Game {s_0, f}$.
%
%
%
Given a state $s_0$ and rational threshold vector $\vec\nu \in \Q^d$,
the \emph{expected-value threshold problem} asks whether $\vec\nu \in \ExpSolP \Game {s_0}$.


\begin{theorem}[\cite{BrazdilBrozekChatterjeeForejtKucera:TwoViews:LMCS:2014}]
	The expected-value threshold problem for multidimensional mean-payoff \MDPs is in \PTIME.
\end{theorem}

\noindent
While randomized finite-memory strategies are both necessary and sufficient in general for achieving a given expected mean payoff,
in ECs we can use randomized finite-memory \emph{unichain} strategies to \emph{approximate} achievable vectors.
Being unichain ensures that the mean payoff equals the expectation almost surely.
By standard convergence results in Markov chains,
this entails that by playing such a strategy for sufficiently long time
we obtain an average mean payoff close to the expectation with high probability.
%
We crucially use this property in the constructions leading to the main results of Sec.~\ref{sec:BWC} and \ref{sec:BAS};
cf. Lemmas~\ref{lem:BWC:synthesis:WEC}, \ref{lem:BWC:synthesis:EC}, and \ref{lem:BAS:synthesis:EC}.

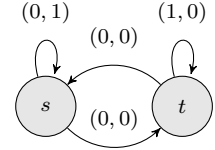
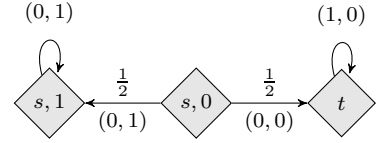
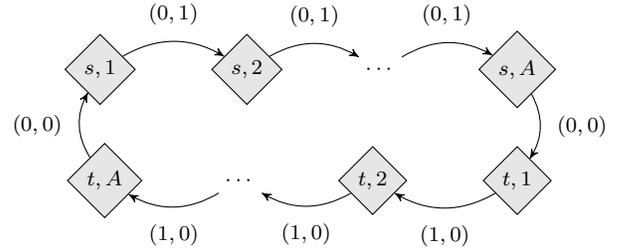
\begin{figure}
	\centering
	\footnotesize

	\begin{minipage}{.5\textwidth}
		\centering
		\begin{tikzpicture}

			\node[playerstate] (s) [] {$s$};
			\node[playerstate] (t) [right = 1 cm of s] {$t$};

			\path[->] (s) edge [loop above] node [above = 1ex] {$(0, 1)$} ();
			\path[->] (s) edge [bend right = 45] node [above = 1ex] {$(0, 0)$} (t);
		
			\path[->] (t) edge [loop above] node [above = 1ex] {$(1, 0)$} ();
			\path[->] (t) edge [bend right = 45] node [above = 1ex] {$(0, 0)$} (s);

		\end{tikzpicture}
		\subcaption{A \MDP reduced to one EC.}
		\label{fig:randomized_EC_example}
	\end{minipage}
\\[4ex]

	\begin{minipage}{.5\textwidth}
		\centering
		\begin{tikzpicture}

			\node[randomstate] (s0) [] {$s, 0$};
			\node[randomstate] (s1) [left = 1 cm of s0] {$s, 1$};
			\node[randomstate] (t) [right = 1 cm of s0] {$t$};

			\path[->] (s0) edge node [above] {$\frac 1 2$} node [below] {$(0, 1)$} (s1);
			\path[->] (s0) edge node [above] {$\frac 1 2$} node [below] {$(0, 0)$} (t);
		
			\path[->] (s1) edge [loop above] node [above = 1ex] {$(0, 1)$} ();
			\path[->] (t) edge [loop above] node [above = 1ex] {$(1, 0)$} ();

		\end{tikzpicture}
		\subcaption{Exact strategy inducing two BSCCs.}
		\label{fig:randomized_EC_2mem_example}
	\end{minipage}
\\[4ex]

	\begin{minipage}{.5\textwidth}
		\centering
		\begin{tikzpicture}

			\node[randomstate] (s0) 	[] {$s, 1$};
			\node[randomstate] (s1) 	[right = 1 cm of s0] {$s, 2$};
			\node[] (sdots)	[right = 1 cm of s1] {$\cdots$};
			\node[randomstate] (sA) 	[right = 1 cm of sdots] {$s, A$};
		
			\node[randomstate] (t0) 	[below = .5 cm of sA] {$t, 1$};
			\node[randomstate] (t1) 	[left = 1 cm of t0] {$t, 2$};
			\node[] (tdots)	[left = 1 cm of t1] {$\cdots$};
			\node[randomstate] (tA) 	[left = 1 cm of tdots] {$t, A$};

			\path[->] (s0) edge [bend left = 30] node [above = 1ex] {$(0, 1)$} (s1);
			\path[->] (s1) edge [bend left = 30] node [above = 1ex] {$(0, 1)$} (sdots);
			\path[->] (sdots) edge [bend left = 30] node [above = 1ex] {$(0, 1)$} (sA);
			\path[->] (sA) edge [bend left = 30] node [right = 1ex] {$(0, 0)$} (t0);

			\path[->] (t0) edge [bend left = 30] node [below = 1ex] {$(1, 0)$} (t1);
			\path[->] (t1) edge [bend left = 30] node [below = 1ex] {$(1, 0)$} (tdots);
			\path[->] (tdots) edge [bend left = 30] node [below = 1ex] {$(1, 0)$} (tA);
			\path[->] (tA) edge [bend left = 30] node [left = 1ex] {$(0, 0)$} (s0);

		\end{tikzpicture}
		\subcaption{Approximate finite-memory strategy inducing one BSCC.}
		\label{fig:pure_EC_example}
	\end{minipage}

	\caption{Approximating the expectation inside ECs.}
	
\end{figure}

\begin{example}

	We illustrate the idea in the single end-component \MDP in Fig.~\ref{fig:randomized_EC_example}
	(cf. \cite[Fig.~3]{VelnerChatterjeeDoyenHenzingerRabinovichRaskin:Complexity:ArXiv:2012}).
	There exists a simple randomized 2-memory strategy $f$ achieving expected mean payoff precisely $(\frac 1 2, \frac 1 2)$
	which decides, with equal probability,
	whether to stay forever in $s$ or in $t$.
	However, the induced Markov chain has two BSCCs;
	cf. Fig.~\ref{fig:randomized_EC_2mem_example}.
	%
	%
	While intuitively no pure finite-memory strategy can achieve mean payoff exactly equal $(\frac 1 2, \frac 1 2)$ in this example,
	%
	%
	finite-memory unichain strategies can approximate this value.
 	For a parameter $A \in \N$, consider the strategy $g_A$ which stays in $s$ for $A$ steps,
	and then goes to $t$, stays in $t$ for $A$ steps, and then goes back to $s$, and repeats this scheme forever.
	The induced Markov chain has only one BSCC, thus $g_A$ is unichain; cf. Fig.~\ref{fig:pure_EC_example}.
	The strategy $g_A$ achieves expected (and worst-case) mean payoff
	$\left(\frac A {2A + 2}, \frac A {2A + 2}\right)$,
	which converges from below to $(\frac 1 2, \frac 1 2)$ as $A \to \infty$.
	
\end{example}

\begin{restatable}{lemma}{ECfinitememoryexpectation}
	\label{lem:strategy:EC:finite-memory:expectation}
	Let $\Game$ be a multidimensional mean-payoff \MDP,
	let $s_0$ be a state in an EC $U$ thereof,
	and let $\vec \nu \in \ExpSolP {\restrict \Game U} {s_0}$
	be an expectation vector achievable by remaining inside $U$.
	There exists a \emph{finite-memory unichain} strategy $g \in \FiniteMemory \Game$
	achieving the same expectation $\vec \nu \in \ExpSolP {\restrict \Game U} {s_0, g}$.
\end{restatable}
\begin{remark}
	\label{rem:pure_strategy}
	In the lemma above, we can take $g$ to be even a \emph{pure} finite-memory unichain strategy.
	This can be obtained by a de-randomization technique
	at the cost of introducing extra memory of size exponential in the number of the states controlled by the player.
	However, we do not need this stronger result in the rest of the paper,
	and we content ourselves with randomized strategies for simplicity.
\end{remark}
\begin{proof}[Proof sketch]
	By the results of \cite{BrazdilBrozekChatterjeeForejtKucera:TwoViews:LMCS:2014},
	there exists a randomized finite-memory strategy $f$ achieving expected mean payoff $\vec\nu^* > \vec\nu$
	which surely stays inside $U$.
	However, $(\restrict \Game U)[f]$ is not unichain in general. 
	By Proposition~\ref{prop:EC}, the set of states visited infinitely often by a play in $(\restrict \Game U)[f]$ is an EC almost surely.
	Since there are finitely many different ECs,
	there are probabilities $\alpha_1, \dots, \alpha_n > 0$ and ECs $U_1, \dots, U_n \subseteq U$
	s.t. the set of states visited infinitely often by a play in $(\restrict \Game U)[f]$
	is $U_1$ with probability $\alpha_1$, \ldots,
	$U_n$ with probability $\alpha_n$.
	By Proposition~\ref{prop:EC}, $\alpha_1 + \dots + \alpha_n = 1$.
	%
	In the first step, we define a ``local'' randomized memoryless strategy $g_i$ which plays as $f$ once inside $U_i$.
	No approximation is introduced in this step.
	%
	In the second step, we combine the local randomized memoryless strategies $g_i$'s above.
	We build a randomized finite-memory strategy $g$
	which cycles between $U_1, \dots, U_n$ and plays according to $g_i$ inside each $U_i$ a fraction $\approx\alpha_i$ of the time.
	This is possible since $U_i, U_j$ are almost surely mutually inter-reachable
	due to the fact that we are always inside the EC $U$.
	By construction, $(\restrict \Game U)[g]$ is unichain since $g$ cycles between all the ECs $U_1, \dots, U_n$.
	Moreover, for every $\varepsilon > 0$, we can make the expected fraction of time spent changing component smaller than $\varepsilon$.
	Thus, $g$ achieves expected mean payoff at least $(1-\varepsilon) \cdot \vec\nu^* - (W, \dots, W) \cdot \varepsilon$,
	where $W$ is the largest absolute value of any weight in $\Game$.
	The latter quantity can be made $> \vec \nu$ for sufficiently small $\varepsilon > 0$.
\end{proof}

\subsection{Worst-case objective}

For a \MDP $\Game$, a starting state $s_0$, and a Controller's strategy $f \in \All \Game$,
the set of \emph{worst-case achievable solutions for $f$} is defined as
$\WCSolP \Game {s_0, f} = \setof {\vec \mu \in \R^d} {\forall \pi \in \Outcome {s_0} f \cdot \MP(\pi) > \vec \mu}$,
i.e., it is the set of vectors $\vec \mu$ s.t.
Controller can surely guarantee a mean payoff $>\vec \mu$ from state $s_0$ by playing $f$.
The set of \emph{worst-case achievable solutions} is
$\WCSolP \Game {s_0} = \bigcup_{f \in \All \Game} \WCSolP \Game {s_0, f}$.
Given a state $s_0$ and rational threshold vector $\vec\mu \in \Q^d$,
the \emph{worst-case threshold problem} asks whether $\vec\mu \in \WCSolP \Game {s_0}$.
%
%


With this worst-case interpretation, the randomized choices in the MDP are replaced by purely adversarial ones,
and the MDP can thus be viewed as a two-player zero-sum game.
While infinite-memory strategies are more powerful than finite-memory ones for the worst-case objective,
the latter suffice to approximate achievable vectors.
We make extensive use of this property in Sec.~\ref{sec:finite-memory} where we restrict our attention to finite-memory strategies.
%
%
%
%
\begin{lemma}[cf. Lemma~15 of \cite{VelnerChatterjeeDoyenHenzingerRabinovichRaskin:Complexity:ArXiv:2012}]
	\label{lem:strategy:worst-case}
	Let $\Game$ be a multidimensional mean-payoff \MDP, $s_0$ a state therein,
	and let $\vec \mu \in \WCSolP \Game {s_0}$.
	There exists a pure \emph{finite-memory} strategy
	$f \in \PureFiniteMemory \Game$ for Controller
	s.t. $\vec\mu \in \WCSolP \Game {s_0, f}$.
\end{lemma}

The finite-memory strategy threshold problem for multidimensional mean-payoff games is \coNPc \cite{VelnerRabinovich:Noisy:FOSSACS:2011,VelnerChatterjeeDoyenHenzingerRabinovichRaskin:Complexity:ArXiv:2012}.
By the lemma above, finite memory controllers suffice in our setting,
and we obtain the following complexity characterization.

\begin{theorem}[\cite{VelnerRabinovich:Noisy:FOSSACS:2011,VelnerChatterjeeDoyenHenzingerRabinovichRaskin:Complexity:ArXiv:2012}]
	\label{thm:multidimensional:worst-case:complexity}
	The worst-case threshold problem for multidimensional mean-payoff \MDPs  is \coNPc.
\end{theorem}


\noindent
In the unidimensional case, memoryless strategies suffice for both players \cite{Liggett:Lippman:1969,EhrenfeuchtMycielski:MeanPayoff:1979},
and the complexity is \NPcoNP (and even \UPcoUP \cite{ZwickPaterson:MeanPayoff:TCS:1996,Jurdzinski:1998}).
It is open since long time whether this problem is in \PTIME.
\begin{theorem}[\cite{Liggett:Lippman:1969,EhrenfeuchtMycielski:MeanPayoff:1979}]
	\label{thm:unidimensional:worst-case:complexity}
	The worst-case threshold problem for unidimensional mean-payoff \MDPs  is in \NPcoNP.
\end{theorem}

\section{Beyond worst-case synthesis}

\label{sec:BWC}

We generalize \cite{Bruyere:Filiot:Randour:Raskin:BWC:2014} to the multidimensional setting.
Given a MDP~$\Game$,
a starting state $s_0$,
and a Controller's strategy $f$, 
the set of \emph{beyond worst-case achievable solutions for $f$}, denoted $\BWCSolP \Game {s_0, f}$,
is the set of pairs of vectors $(\vec \mu; \vec \nu) \in \R^{2d}$ s.t.
$f$ surely guarantees a worst-case mean payoff $>\vec \mu$,
and achieves an expected mean payoff $>\vec \nu$
starting from $s_0$,
\begin{align*}
	\BWCSolP \Game {s_0, f} = \SetOf {(\vec \mu; \vec \nu) \in \R^{2d}}
	{\begin{array}{c}
		\vec \mu \in \WCSolP \Game {s_0, f} \\
			\textrm { and } \\
		\vec \nu \in \ExpSolP \Game {s_0, f}
	\end{array}}
\end{align*}
Let $\BWCSolP \Game {s_0} = \bigcup_{f \in \All \Game} \BWCSolP \Game {s_0, f}$
be the set of \emph{beyond worst-case achievable solutions}.
Given a starting state $s_0$ and a pair of threshold vectors $(\vec \mu; \vec \nu) \in \R^{2d}$,
the \emph{beyond worst-case threshold problem} (BWC) asks whether $(\vec \mu; \vec \nu) \in \BWCSolP \Game {s_0}$.

\begin{remark}
	We assume w.l.o.g. that $\vec \mu = \vec 0$.
	This follows by shifting each component by an appropriate amount.
	We further assume w.l.o.g. that $\vec \nu \geq \vec 0$.
	This follows from the fact that, since the mean payoff is surely $> \vec 0$ by the worst-case objective,
	then also the expectation is $> \vec 0$.
\end{remark}

\begin{remark}
	We say that $\Game$ is \emph{pruned} if $\vec 0 \in \BWCSolP \Game s$ for every state $s$ therein.
	Controller cannot satisfy the BWC objective if she ever visits a state $s$ not satisfying the worst-case objective.
	Many of our results are thus stated under the condition that $\Game$ is pruned.
	However, pruning an MDP, i.e., removing those states which are losing w.r.t. the worst-case objective,
	requires solving a mean-payoff game, and this will have a crucial impact on the complexity.
\end{remark}

\noindent
The finite-memory threshold problem for the unidimensional beyond worst-case problem has been studied in \cite{Bruyere:Filiot:Randour:Raskin:BWC:2014}.
\begin{theorem}[\cite{Bruyere:Filiot:Randour:Raskin:BWC:2014}]
	\label{thm:unidimensional:BWC:complexity}
	The finite-memory threshold problem for the unidimensional beyond worst-case problem for mean-payoff objectives is in \NPcoNP.
\end{theorem}

\subsection{Finite-memory synthesis}

\newcommand{\stratexp}[1]{g^{exp}_{#1}}
\newcommand{\stratwc}[1]{f^{wc}_{#1}}
\newcommand{\stratcmb}[1]{h^{cmb}_{#1}}

\newcommand{\mumin}{\mu^*}
\newcommand{\numax}{\nu^*}

\label{sec:finite-memory}

In this section, we address the problem of deciding whether there exists a finite-memory strategy for the BWC problem in the multidimensional setting.
By Proposition~\ref{prop:EC}, we know that the set of states visited infinitely often by any strategy (not necessarily a finite-memory one)
is almost surely an EC.
The crucial observation is that, when restricted to finite-memory, the same holds for ECs of a special kind.
An EC $U$ is \emph{winning} (WEC) iff Controller can surely guarantee the worst-case threshold $> \vec 0$
when constrained to remain in $U$, starting from any state therein.
Whether a EC is winning depends on the worst-case objective alone.

The following proposition is central in the analysis of the BWC problem for finite-memory strategies;
cf. \cite[Lemma 4]{Bruyere:Filiot:Randour:Raskin:BWC:2014} in the unidimensional case.

\begin{restatable}{proposition}{propFiniteMemWEC}
	\label{prop:finite-mem:WEC}
	Let $f$ be a \emph{finite-memory} strategy satisfying the worst-case threshold problem.
	The set of states visited infinitely often under $f$ is almost surely a \emph{winning} EC.
\end{restatable}


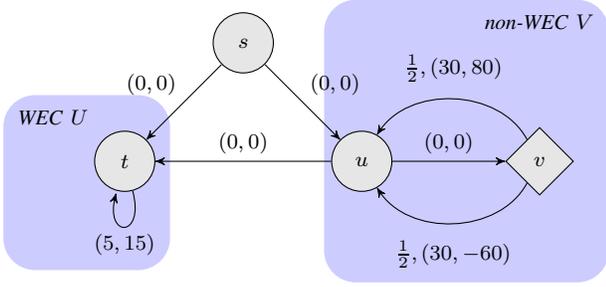
\begin{figure}
	\footnotesize
	\begin{center}
		\begin{tikzpicture} 

			\node[playerstate] (s) [] {$s$};
			\node[playerstate] (t) [below left = 1cm and 1cm of s] {$t$};
			\node[playerstate] (u) [below right = 1cm and 1cm of s] {$u$};
			\node[randomstate] (v) [right = 1.5cm of u] {$v$};
			
			\node (WEC) [above left = .1 and .1 of t] {\emph{WEC} $U$};
			\node (non-WEC) [above = 1.2 of v] {\emph{non-WEC} $V$};
			
			\path[->] (s) edge node [above left] {$(0,0)$} (t);
			\path[->] (s) edge node [above right] {$(0,0)$} (u);
			
			\path[->] (t) edge [loop below] node (tt) [below] {$(5,15)$} ();

			\path[->] (u) edge node (uv) [above] {$(0,0)$} (v);
			\path[->] (u) edge node [above] {$(0,0)$} (t);
			\path[->] (v) edge [bend right = 60] node (vu0) [above = 1ex] {$\frac 1 2, (30,80)$} (u);
			\path[->] (v) edge [bend left = 60] node (vu1) [below = 1ex] {$\frac 1 2, (30,-60)$} (u);
			
			\begin{pgfonlayer}{background}
				\node [fill=blue!20, rectangle, rounded corners = 12pt, fit=(t) (tt) (WEC)] {};
				\node [fill=blue!20, rectangle, rounded corners = 12pt, fit=(u) (v) (vu0) (vu1) (non-WEC)] {};
			\end{pgfonlayer}
			
		\end{tikzpicture}
	\end{center}
	\caption{Running example.}
	\label{fig:running_example}
\end{figure}

\begin{runningexample}
	As a simple example that will be used through the rest of the paper,
	consider the \MDP in Fig.~\ref{fig:running_example}.
	There are only two ECs $U$ and $V$,
	of which $U$ is winning, but $V$ is not.
	Indeed, from $v$ the adversary can always select the lower edge with payoff $(30, -60)$.
	In $U$ we can achieve expectation $(5, 15)$,
	and from $V$ we can achieve expectation $(15, 5)$.
	Therefore, according to the lemma above,
	any finite-memory strategy satisfying the worst-case objective will eventually go to $U$ almost surely.
\end{runningexample}

We proceed by analyzing WECs separately in Sec.~\ref{sec:finite-memory:WEC},
and then we tackle general MPDs in Sec.~\ref{sec:finite-memory:general}.
This will yield our complexity result in Sec.~\ref{sec:BWC:finite-memory:complexity}.

\subsubsection{Inside a WEC}

\label{sec:finite-memory:WEC}

We show that inside WECs finite-memory strategies always suffice for the BWC objective.
In particular, the threshold problem in WECs immediately reduces to an expectation threshold problem.
\begin{restatable}{lemma}{lemBWCsynthesisWEC}
	\label{lem:BWC:synthesis:WEC}
	Let $\Game$ be a pruned multidimensional mean-payoff \MDP,
	let $s_0$ be a state in a WEC $W$ of $\Game$,
	and let $\vec \nu \in \ExpSolP {\restrict \Game W} {s_0}$ with $\vec \nu \geq \vec 0$
	be an expectation achievable by remaining inside $W$.
	There exists a \emph{randomized finite-memory} strategy $h \in \FiniteMemory \Game$
	s.t. $(\vec 0; \vec \nu) \in \BWCSolP {\restrict \Game W} {s_0, h}$
	that also remains inside $W$.
	\ignore{
		Moreover, $h$ uses memory of pseudo-polynomial size
		$O(d \cdot W \cdot \frac 1 \delta \cdot \frac 1 \varespilon)$,
		where $\delta > 0$ is the error in the worst-case component,
		and $\varepsilon > 0$ is the error in the expectation component.
		It suffices to take $\varepsilon = \max \vec \nu / 2$.
	}
\end{restatable}
\begin{remark}
	The statement of the lemma holds even with $h$ a \emph{pure} finite-memory strategy,
	by applying Remark~\ref{rem:pure_strategy} when constructing the expectation strategy which is part of $h$.
	However, randomized strategies suffice for our purposes.
\end{remark}
\noindent
We use finite-memory strategies defined in WECs (such as $h$ above)
when constructing a global BWC strategy in the analysis of arbitrary MDPs in Sec.~\ref{sec:finite-memory:general}.
The construction of $h$ is done in a way analogous to the proof of Theorem~5 in \cite{Bruyere:Filiot:Randour:Raskin:BWC:2014};
cf. App.~\ref{app:BWC} for the details.
However, the analysis in the multidimensional case is considerably more difficult than in previous work.
It crucially relies on Lemma~\ref{lem:strategy:EC:finite-memory:expectation}
for the extraction of \emph{finite-memory unichain} strategies approximating the expectation objective inside ECs.
Note that in the unidimensional case of \cite{Bruyere:Filiot:Randour:Raskin:BWC:2014}
optimal expectation values can be reached exactly already by \emph{pure memoryless unichain strategies} (no approximation needed).
This is an key technical difference between our multidimensional setting and the unidimensional one of \cite{Bruyere:Filiot:Randour:Raskin:BWC:2014}.

\subsubsection{General case}

\label{sec:finite-memory:general}

\begin{figure*}

	\begin{align}
		\tag*{(A1)}
		1_{s_0}(s) + \sum_{(r, s) \in E} y_{rs} &= \sum_{(s, t) \in E} y_{st} + y_s
			&& \forall s \in S
		\label{eq:y:flow}
		\\
		\tag*{(A1')}
		y_{st} &= R(s)(t) \cdot ( \!\! \sum_{(r, s) \in E} \!\!\!\! y_{rs} - y_s)
			&& \forall (s, t) \in E \textrm{ with } s \in S^R
		\label{eq:y:random}
		\\
		\tag*{(A2)}
		\sum_{ \textrm{ MWEC } U }\sum_{s \in U} y_s &= 1
		\label{eq:y:sum}
		\\
		\tag*{(B)}
		\sum_{s \in U} y_s &= \sum_{(r, s) \in E\cap U \times U} x_{rs}
		 	&& \forall \textrm{ MWEC } U 
		\label{eq:x:y}
		\\
		\tag*{(C1)}
		\sum_{(r, s) \in E} x_{rs} &= \sum_{(s, t) \in E} x_{st}
			&& \forall s \in S 
		\label{eq:x:flow}
		\\
		\tag*{(C1')}
		x_{st} &= R(s)(t) \cdot \sum_{(r, s) \in E} x_{rs}
			&& \forall (s, t) \in E \textrm{ with } s \in S^R 
		\label{eq:x:random}
		\\
		\tag*{(C2)}
		\sum_{(s, t) \in E} x_{st} \cdot w(s, t)[i] & > \vec \nu[i] 
			&& \forall (1 \leq i \leq n) 
		\label{eq:MP:global}
		\\
		\tag*{(C3)}
		\sum_{(s, t) \in E\cap U \times U} \!\!\!\!\!\!\!\! x_{st} \cdot w(s, t)[i] & > 0
			&& \forall \textrm{ MWEC } U, 1 \leq i \leq n
		\label{eq:MP:local}
	\end{align}
	
	\caption{Linear system $T$ for the BWC finite-memory threshold problem.}
	\label{fig:BWC:system}
	
\end{figure*}

%
%



We reduce the finite-memory BWC problem to the solution of a system of linear inequalities.
This is similar to the solution of the multidimensional expectation problem presented in \cite{BrazdilBrozekChatterjeeForejtKucera:TwoViews:LMCS:2014}.
When only the expectation is considered, the intuition is that a ``global expectation''
is obtained by combining together ``local expectations'' achieved in ECs.
Thus, a strategy for the expectation works in two phases:
\begin{enumerate}
	\item[Phase I:] Reach ECs with appropriate probabilities.
	\item[Phase II:] Once inside an EC, switch to a local expectation strategy to achieve the right ``local expectation''.
\end{enumerate}

In the BWC problem, we need to enforce two extra conditions:
First, only ``local expectations'' from \emph{winning} ECs should be considered
(by Proposition~\ref{prop:finite-mem:WEC} finite-memory controllers cannot stay in a non-WEC forever with non-zero probability).
Second, ``local expectations'' should be $> \vec 0$ in order to satisfy the worst-case objective
(a negative ``local expectation'' would violate the worst-case objective).
%
Accordingly, a strategy for the BWC problem behaves as follows:
\begin{enumerate}
	\item[Phase I:] Reach WECs with appropriate probabilities.
	\item[Phase II:] Once inside a WEC, switch to a local BWC strategy to achieve the right ``local expectation'' $> \vec 0$.
\end{enumerate}
We write a system of linear inequalities expressing this two-phase decomposition.
W.l.o.g. we assume that state $s_0$ belongs to Controller,
and that all WECs are reachable with positive probability from $s_0$ (unreachable states can be removed).
Consider the system $T$ in Fig.~\ref{fig:BWC:system}.
%
%
%
For each state $s \in S$ we have a variable $y_s$,
and for each edge $(s, t) \in E$ we have variables $x_{st}$ and $y_{st}$.
System $T$ can be divided into three parts.
The first part consists of Equations~\ref{eq:y:flow}--\ref{eq:y:sum}.
Variable $y_s$ represents the probability that, upon visiting state $s$,
we switch to Phase II.
Variables $y_{st}$'s are used to express flow conditions.
In Eq.~\ref{eq:y:flow} we put an initial flow of $1$ in $s_0$,
and we require that the total incoming flow to a state equals the outgoing flow (including the leak $y_s$).
In Eq.~\ref{eq:y:random} ensures that the outgoing flow through an edge $y_{st}$ from a stochastic state $s$ is a fixed fraction of the incoming flow.
Finally, Eq.~\ref{eq:y:sum} states that we switch to Phase II in a WEC almost surely.
%

Before explaining the other two parts of $T$,
we need to introduce maximal WECs.
A \emph{maximal WECs} (MWEC) is a WEC which is not strictly included into another WEC.
The restriction to MWECs is crucial for complexity. 
The second part of $T$ consists of Eq.~\ref{eq:x:y} and it provides a link between Phase I and Phase II.
Variable $x_{st}$ represents the long-run frequency of edge $(s, t)$.
Eq.~\ref{eq:x:y} links the transient behaviour before switching inside a certain MWEC and the steady state behaviour once inside it.
More precisely, it guarantees that the probability to switch inside a certain MWEC
equals the total long-run frequency of all edges in the MWEC.

Finally, the remaining equations make up the third part of $T$.
Eq.~\ref{eq:x:flow} is a flow condition for the $x_{st}$'s,
stating that the incoming flow to a state equals the outgoing flow.
Eq.~\ref{eq:x:random} forces the flow to respect the probabilities of stochastic states.
Eq.~\ref{eq:MP:global} guarantees that the expected mean payoff is $> \vec\nu$,
as required.
Eq.~\ref{eq:MP:local} needs some justification.
It is specific to our setting and it does not follow from \cite{BrazdilBrozekChatterjeeForejtKucera:TwoViews:LMCS:2014}.
This equation specifies that the expected mean payoff is $> \vec 0$ \emph{inside every MWECs}.
We need to ensure that only ``local'' expected mean payoffs $> \vec 0$ should be considered in WECs,
in order to be able to apply the results from the previous Sec.~\ref{sec:finite-memory:WEC}.
Eq.~\ref{eq:MP:local} imposes a seemingly strong constraint
by requiring that \emph{all WECs} are visited infinitely often with positive probability.
Ideally, we would like to guess which are the MWECs which need to be visited infinitely often with positive probability,
but this would not yield a good complexity,
since there are exponentially many different sets of MWECs.
Instead, we require that \emph{every} MWEC is visited infinitely often with some positive probability.
Since we are only interested in approximating the expectation, 
it is always possible to put an arbitrary small total probability on MWECs that do not contribute to the ``global'' mean payoff.
This is formalized below.
\begin{proposition}
	\label{prop:BWC:strategy:all:MWEC}
	Let $\Game$ be a pruned multidimensional mean-payoff \MDP.
	If there exists a finite-memory strategy $h$ s.t. $(\vec 0; \vec \nu) \in \BWCSolP \Game {s_0, h}$,
	then there exists a finite-memory strategy $h^*$ with the same property,
	and such that, for every MWEC $U$,
	the set of states visited infinitely often by $h^*$ is a subset of $U$ with positive probability.
\end{proposition}
\begin{proof}
	Since by assumption all MWEC are reachable with positive probability from $s_0$,
	for every MWEC $U$ there exists a strategy $f_U$ reaching $U$ with positive probability from $s_0$.
	Moreover, since $U$ is a WEC,
	there exists a strategy $f_U^{wc}$ for the worst-case objective $> \vec 0$ that surely remains in $U$.
	Let $f^{wc}$ be a worst-case strategy winning everywhere (it exists since $\Game$ is pruned by assumption).
	We construct the following strategy $f_N$ parametrized by a natural number $N > 0$:
	\begin{itemize}
		
		\item Choose a MWEC $U$ uniformly at random.
		
		\item Play $f_U$ for $N$ steps.
		
			\begin{itemize}
				\item If after $N$ steps the play is in $U$, then switch to $f_U^{wc}$.
				\item Otherwise, switch to $f^{wc}$.
			\end{itemize}
			
	\end{itemize}
	By construction $f_N$ is winning for the worst-case for every $N > 0$.
	Moreover, it is easy to see that there exists an $N^*$ sufficiently large
	s.t., for every MWEC $U$, $f_{N^*}$ visits $U$ infinitely often with positive probability.
	
	Finally, the strategy $h^*$ plays with probability $p > 0$ according to $f_{N^*}$,
	and otherwise according to $h$.
	Since both $f_{N^*}$ and $h$ are winning for the worst-case, so it is $h^*$.
	The expected mean payoff of $h^*$ converges from below to the expected mean payoff of $h$
	for $p > 0$ sufficiently small.
	Therefore, there exists $p > 0$ s.t. $(\vec 0; \vec \nu) \in \BWCSolP \Game {s_0, h^*}$.
\end{proof}


\noindent
We now state the correctness of the reduction.
\begin{lemma}
	\label{lem:finite-memory:general}
	Let $\Game$ be a pruned multidimensional mean-payoff \MDP,
	let $s_0 \in S$,
	and let $\vec \nu \geq \vec 0$.
	There exists a \emph{finite-memory} strategy $h$ s.t. $(\vec 0; \vec \nu) \in \BWCSolP \Game {s_0, h}$,
	if, and only if, the system $T$ has a non-negative solution.
\end{lemma}
\noindent
The rest of this section is devoted to the proof of the lemma above.
Both directions are non-trivial.
For the right-to-left direction, we need to explain which kind of strategies can be extracted from a non-negative solution of $T$.
The following lemma shows that from a non-negative solution of $T$
we can extract a strategy for the expectation combining only ``local mean payoffs'' $> \vec 0$
and visiting infinitely often each MWEC with positive probability.

\begin{proposition}
	\label{prop:finite-memory:general:1:analysis}
	If $T$ has a non-negative solution,
	then there exists a finite-memory strategy $\hat h$ s.t. $\vec \nu \in \ExpSolP \Game {s_0, \hat h}$,
	and
	\begin{enumerate}
		\item
			\label{prop:finite-memory:general:1:analysis:yU}
			For every MWEC $U$, there is a probability $y^*_U > 0$
			s.t. the set of states visited infinitely often by $\hat h$ is a subset of $U$ with probability $y^*_U$.
		\item
			\label{prop:finite-memory:general:1:analysis:nuU}
			 Once $\hat h$ reaches the MWEC $U$, it achieves expected mean payoff $\vec \nu_U > \vec 0$.
		\item
			\label{prop:finite-memory:general:1:analysis:nu}
			$\sum_{\textrm{MWEC } U} y^*_U \cdot \vec \nu_U > \vec \nu$.
	\end{enumerate}
\end{proposition}
\begin{proof}
	
	Let $\set{y^*_s}_{s \in S}$, $\set{y^*_{st}}_{(s, t) \in E}$, and $\set{x^*_{st}}_{(s, t) \in E}$
	be a non-negative solution to $T$.
	Proposition~4.2 of \cite{BrazdilBrozekChatterjeeForejtKucera:TwoViews:LMCS:2014}
	essentially shows how to construct from the solution above a finite-memory strategy $\hat h$ 
	s.t. $\vec \nu \in \ExpSolP \Game {s_0, \hat h}$.
	
	For a MWEC $U$, let
	\begin{align}
		y^*_U = \sum_{s \in U} y^*_s 
	\end{align}
	By Eq.~\ref{eq:MP:local}, for every MWEC $U$ there exist $s, t \in U$ s.t. $x^*_{st} > 0$.
	Together with Eq.~\ref{eq:x:y}, this implies that $y^*_U > 0$,
	which proves Point~\ref{prop:finite-memory:general:1:analysis:yU}.
		
	For a MWEC $U$, let
	\begin{align}
		\vec\nu_U = \sum_{(s, t) \in E\cap U \times U} \!\!\!\!\!\!\!\! x^*_{st} \cdot w(s, t)
	\end{align}
	and notice that $\vec\nu_U$ is the expected mean payoff of $\hat h$ once inside $U$.
	By Eq.~\ref{eq:MP:local}, $\vec\nu_U > \vec 0$, which proves Point~\ref{prop:finite-memory:general:1:analysis:nuU}.

	
	Eq.~\ref{eq:x:y} implies that $\hat h$ eventually stays forever inside a WEC almost surely.
	Consequently, $\sum_{\textrm{MWEC } U} y^*_U = 1$. 
	Since states visited infinitely often with probability zero do not contribute to the expected mean payoff,
	it suffices to look at MWECs.
	By the prefix independence of the mean payoff value function,
	and since MWEC $U$ is reached with probability $y^*_U$,
	strategy $\hat h$ achieves expected mean payoff $\sum_{\textrm{MWEC} U} y^*_U \cdot \vec \nu_U$.
	By Point 1), the latter quantity is $> \vec \nu$.
\end{proof}

\ignore{

Eq.~\ref{eq:x:y} in $T$ is the same as Eq.~(3) in $L$.

Equations~\ref{eq:x:flow} and \ref{eq:x:random}	are analogous to Eq.~(4) in $L$,
adapted from the action-based setting.
Eq.~\ref{eq:x:flow} expresses a flow balance condition,
and Eq.~\ref{eq:x:random} says that the flow respects the probabilities assigned to transitions in stochastic states.

Eq.~\ref{eq:MP:global} is obtained from Eq.~(5) in $L$ by using a strict inequality sign $>$.
This enforces that the expected mean payoff is strictly larger than $\vec \nu$.

Finally, Eq.~\ref{eq:MP:local} is specific to our setting, and it has no counterpart in $L$.
It expresses the fact that, if starting from an MWEC, the expected mean payoff is $> \vec 0$.
This effectively forces a global strategy for the expectation to use only strictly positive ``local'' mean payoffs,
and this is necessary to invoke our approximation result from Lemma~\ref{lem:BWC:synthesis:WEC}.
}

\noindent
We are now ready to prove Lemma~\ref{lem:finite-memory:general}.

\newcommand{\hexp}{h^{exp}}
\newcommand{\hcmb}{h^{cmb}}
\newcommand{\hwc}{h^{wc}}
\newcommand{\hU}[1]{h^{cmb}_{#1}}


\begin{proof}
	[Proof of Lemma~\ref{lem:finite-memory:general}]
	For the left-to-right direction, assume that $h$ is a finite-memory strategy %
	guaranteeing $(\vec 0; \vec \nu) \in \BWCSolP \Game {s, h}$.
	Proposition 4.4 of \cite{BrazdilBrozekChatterjeeForejtKucera:TwoViews:LMCS:2014}
	essentially shows that any strategy satisfying the expectation objective $> \vec \nu$
	induces a solution to $T$ satisfying Equations~\ref{eq:y:flow}--\ref{eq:MP:global},
	except that Eq.~\ref{eq:x:y} should be interpreted over MECs (instead of MWECs).
	(This follows from the fact that the set of states visited infinitely often by any strategy is an EC almost surely;
	cf. Proposition~\ref{prop:EC}.)
	However, since $\vec 0 \in \WCSolP \Game {s, h}$ and $h$ is finite-memory,
	we can apply Proposition~\ref{prop:finite-mem:WEC}
	and deduce that	$h$ visits infinitely often a \emph{winning} EC almost surely.
	Thus Eq.~\ref{eq:x:y} is satisfied even over MWECs.
	
	It remains to address Eq.~\ref{eq:MP:local}.
	By Proposition~\ref{prop:BWC:strategy:all:MWEC},
	there exists a strategy $h^*$ s.t., for every MWEC $U$,
	$h^*$ eventually stays forever in $U$ with a positive probability.
	This implies that, when constructing a solution to $T$ induced by $h^*$ (as above),
	for every MWEC $U$ and $s, t \in U$, $x^*_{st} > 0$.
	Moreover, since $h^*$ is winning for the worst-case,
	it achieves an expected mean payoff $> \vec 0$ in $U$,
	and thus Eq.~\ref{eq:MP:local} is satisfied.
	%


	
	For the right-to-left direction, assume that $T$ has a non-negative solution.
	Let $\hat h$ be the strategy in $\Game$ given by Proposition~\ref{prop:finite-memory:general:1:analysis}.
	For every MWEC $U$, let $y_U^*$ and $\vec\nu_U$ be as given in the statement of the proposition.
	While $\hat h$ alone is not sufficient to show $(\vec 0; \vec \nu) \in \BWCSolP \Game s$
	since it does not satisfy the worst-case objective in general,
	we show how to construct from it another finite-memory strategy $\hcmb$	ensuring the BWC objective.
	The latter strategy is obtained by combining together the following strategies:
	\begin{itemize}
		
		\item Let $\hwc$ be a finite-memory strategy in $\Game$
		ensuring the worst-case mean payoff $\vec 0 \in \WCSolP \Game {t, \hwc}$
		from every state $t$ in $\Game$.
		This is possible since $\Game$ is pruned.
		
		\item For each MWEC $U$, let $h_U$ be a finite-memory strategy s.t.
		$(\vec 0; \vec \nu_U) \in \BWCSolP \Game {t, h_U}$ for every state $t \in U$.
		
		This strategy can be obtained as follows.
		Let $\restrict \Game U$ be the game $\Game$ restricted to the EC $U$.
		By Point~\ref{prop:finite-memory:general:1:analysis:nuU} of Proposition~\ref{prop:finite-memory:general:1:analysis},
		$\vec \nu_U \in \ExpSolP {\restrict \Game U} {t_0, h_U}$ for some state $t_0 \in U$.
		Since $U$ is an EC, $\vec \nu_U \in \ExpSolP {\restrict \Game U} {t, h_U}$ for \emph{every} state $t \in U$.
		Since $\vec \nu_U > \vec 0$, we can apply Lemma~\ref{lem:BWC:synthesis:WEC} for every $t\in U$,
		and obtain a strategy $h_t$ s.t. $(\vec 0; \vec \nu_U) \in \BWCSolP {\restrict \Game U} {t, h_t}$.
		Let $h_U$ be the finite-memory strategy in $\restrict \Game U$
		that plays according to $h_t$ when starting from state $t$.
		Clearly, $(\vec 0; \vec \nu_U) \in \BWCSolP {\restrict \Game U} {t, h_U}$.
		
	\end{itemize}
	Consider the strategy $\hcmb_N$ parameterized by a natural number $N > 0$
	which is defined as follows:
	\begin{enumerate}
		\item[1)] Play according to $\hat h$ for $N$ steps.
		\item[2)] After $N$ steps:
		\begin{enumerate}
			\item[2a)] If we are inside the MWEC $U$, then switch to $h_U$.
			\item[2b)] Otherwise, play according to $\hwc$. 
		\end{enumerate}
	\end{enumerate}
	We argue that $\hcmb_N$ satisfies the beyond worst-case objective $(\vec 0; \vec\nu) \in \BWCSolP \Game {s_0, \hcmb_N}$
	for $N$ large enough.
	For every $N$, $\hcmb_N$ clearly satisfies the worst-case objective,
	since after $N$ steps it switches to a strategy that satisfies it by construction
	(by prefix-independence of the mean payoff objective).
	We now consider the expectation objective.
	By Point~\ref{prop:finite-memory:general:1:analysis:yU} of Proposition~\ref{prop:finite-memory:general:1:analysis},
	the set of states visited infinitely often by $\hat h$ is a subset of the MWEC $U$ with probability $y^*_U$.
	By taking $N$ large enough,
	we can guarantee being inside $U$ with probability arbitrarily close to $y^*_U$.
	%
	By construction, $h_U$ can be chosen to achieve expected mean payoff arbitrarily close to $\vec \nu_U$.
	Since $\hcmb_N$ switches to $h_U$ with probability arbitrarily close to $y^*_U$,
	$\hcmb_N$ achieves expected mean payoff arbitrarily close to $\sum_{\textrm{MWEC } U} y^*_U \cdot \vec\nu_U$.
	By Point~\ref{prop:finite-memory:general:1:analysis:nu} of Proposition~\ref{prop:finite-memory:general:1:analysis},
	the latter quantity is $> \vec\nu$.
	There exists $N^*$ large enough s.t. $\hcmb_{N^*}$ achieves expected mean payoff $> \vec\nu$.
	Take $\hcmb = \hcmb_{N^*}$.
	As required, $\vec\nu \in \ExpSolP \Game {s_0, \hcmb}$.
\end{proof}

\begin{runningexample}
	Since $U = \set {t}$ is a MWEC, while $V = \set{u, v}$ is not,
	finite memory strategies must go to $U$.
	Therefore, with finite memory we can ensure BWC threshold $((0, 0);(0, 9))$,
	but not $((0, 0);(9, 9))$ for example.
\end{runningexample}

\subsubsection{Complexity}

\label{sec:BWC:finite-memory:complexity}

We obtain the following complexity characterization for the threshold problem with finite-memory controllers.

\begin{theorem}
	The finite-memory multidimensional mean-payoff BWC threshold problem is \coNPc.
\end{theorem}

\begin{proof}
	Pruning states where the worst-case objective cannot be satisfied requires solving multidimensional mean-payoff games,
	which can be done in \coNP by Theorem~\ref{thm:multidimensional:worst-case:complexity}.
	It has been already shown in \cite{Bruyere:Filiot:Randour:Raskin:BWC:2014} how the decomposition in MWEC can be performed in \P
	with an oracle for solving mean-payoff games.
	Thus, the MWEC decomposition can be performed in \coNP.
	System $T$ has size polynomial in $\Game$ (there are only polynomially many MWECs)
	and it can thus be produced in \coNP.
	By Lemma~\ref{lem:finite-memory:general},
	it suffices to solve system $T$,
	which can be done in polynomial time by linear programming.
	The lower bound follows directly from the fact that
	the multidimensional BWC threshold problem contains the worst-case as a subproblem;
	the latter is \coNPh as recalled in Theorem~\ref{thm:multidimensional:worst-case:complexity}.
\end{proof}

The complexity of the BWC problem is dominated by the worst-case subproblem.
We obtain an improved complexity by restricting the worst-case to be essentially unidimensional.
Formally, we say that a BWC threshold $(\vec\mu; \vec\nu)\in\R^{2d}$ has \emph{trivial worst-case component $i$},
with $1 \leq i \leq d$,
iff $\vec\mu[i] = -W$,
where $W$ is the maximal absolute value of any weight in $\Game$.
We say that $(\vec\mu; \vec\nu)$ \emph{is essentially worst-case unidimensional}
iff it has at most one non-trivial worst-case component.
We can ignore trivial components when solving a worst-case threshold problem.
Thus, the worst-case problem for essentially unidimensional thresholds reduces to a simple unidimensional worst-case problem.
As recalled in Theorem~\ref{thm:unidimensional:worst-case:complexity},
the latter can be solved in \NPcoNP,
thus yielding the following improved complexity for the BWC problem.
\begin{corollary}
	The finite-memory multidimensional mean-payoff BWC threshold problem
	w.r.t. \emph{essentially worst-case unidimensional thresholds} is in \NPcoNP.
\end{corollary}
\noindent
Since the \emph{unidimensional} BWC problem,
i.e., where all weights are unidimensional,
is in \NPcoNP (cf. Theorem~\ref{thm:unidimensional:BWC:complexity}),
this results shows that we can add a multidimensional expectation objective to a unidimensional worst-case obligation
without an extra price in complexity.
In particular, we can model complex situations like the task system presented in Sec.~\ref{sec:illustrating:example},
where the worst-case and expectation mean payoffs are along independent dimensions.

\subsection{Infinite-memory synthesis}

\label{sec:BWC:infinite-memory}

Already in the unidimensional case, 
infinite-memory strategies are more powerful than finite-memory ones (cf. \cite[Fig. 6]{BruyereFiliotRandourRaskin:BWC:arxiv}).
This is a consequence of the fact that finite-memory strategies for the BWC objective ultimately remain inside WECs almost surely
(cf. Proposition~\ref{prop:finite-mem:WEC}).
On the other hand, infinite-memory strategies can benefit from payoffs achievable inside arbitrary ECs.
In this section, we address the problem of deciding whether there exists a general strategy,
i.e., not necessarily finite-memory one, for the multidimensional BWC problem.
This was left as an open problem, already in the unidimensional case \cite{Bruyere:Filiot:Randour:Raskin:BWC:2014}.
As in the previous section, we first analyze ECs, and then general \MDPs.

\subsubsection{Inside an EC}

The lemma below is a direct generalization of Lemma~\ref{prop:finite-mem:WEC} to arbitrary ECs.
While for WECs we could construct finite-memory strategies,
we now construct infinite-memory strategies for arbitrary ECs.
\begin{lemma}
	\label{lem:BWC:synthesis:EC}
	Let $\Game$ be a pruned multidimensional mean-payoff \MDP,
	let $s_0$ be a state in an \emph{EC} $U$ of \MDP,
	and let $\vec\nu \geq \vec 0$ be an expectation vector $\vec \nu \in \ExpSolP {\restrict \Game U} {s_0}$
	which is achievable while remaining in $U$.
	There exists a strategy $f \in \All \Game$ (not necessarily remaining in $U$)
	s.t. $(\vec 0; \vec \nu) \in \BWCSolP \Game {s_0, f}$.
\end{lemma}
\begin{remark}
	The statement of the lemma holds even with $f$ a \emph{pure} strategy,
	by applying Remark~\ref{rem:pure_strategy} when constructing the expectation strategy $f^{exp}$ below.
	However, randomized strategies suffice for our purposes.
\end{remark}
\ignore{
\begin{remark}
	The theorem above shows that infinite-memory strategies can approximate achievable expected mean payoffs
	while at the same time guaranteeing the worst-case objective.
	Moreover, it is possible to show that, in a \emph{WEC},
	infinite-memory strategies can guarantee the worst-case objective \emph{and} precisely achieve Pareto-optimal expectation.
\end{remark}
}
The rest of this section is devoted to the proof of Lemma~\ref{lem:BWC:synthesis:EC}.
We proceed by combining in a non-trivial way a strategy for the expectation with a strategy for the worst-case.
Let $f^{wc}$ be a worst-case strategy s.t. $\vec 0 \in \WCSolP \Game {s, f^{wc}}$
for every state $s$,
which exists since the $\Game$ is pruned.
Let $f^{exp}$ be a expectation strategy s.t. $\vec\nu \in \ExpSolP {\restrict \Game U} {s_0, f^{exp}}$.
By Lemma~\ref{lem:strategy:EC:finite-memory:expectation},
we can assume that $f^{exp}$ is finite-memory and unichain.
For technical reasons, it is convenient to assume that $f^{exp}$ is finite-memory,
even though we are going to construct a infinite-memory strategy.
Moreover, since we are in an EC,
we can further assume that $f^{exp}$ achieves expectation $> \vec\nu$ from \emph{every} state of the EC.
%
%
%

The idea is to play according to two different modes.
In the first mode, we play according to $f^{exp}$,
and in the second mode according to $f^{wc}$.
We start in the first mode, and possibly go to the second mode according to certain conditions.
This happens with a certain probability, which we call \emph{switching probability}.
Once in the second mode, we remain in the second mode.
In order to achieve an expectation arbitrarily close to that achieved by $f^{exp}$,
we need to be able to make the switching probability arbitrarily small.
At the same time, in order to ensure that the worst-case objective is satisfied,
we need to guarantee that, when no switch occurs, the mean payoff is surely $> \vec 0$.
(If a switch occurs, the worst-case is satisfied by the definition of $f^{wc}$.)
These two constraints are conflicting and make the construction of a combined strategy non-trivial.

The \emph{combined} strategy $f_K$ is parameterized by a natural number $K > 0$.
In order to decide whether to switch to the second mode or not,
we keep track of the total payoff since the beginning of the play as a vector in $\Z^d$.
This value is unbounded in general, and this is explains why the strategy uses infinite memory.
Let $\vec N_i$ be
\begin{align*}
	\vec N_i = \frac {\vec \nu \cdot i \cdot K} 2.
\end{align*}
Thus, during the first mode the expected total payoff at the end of phase $i$ is $> 2 \cdot \vec N_{i+1}$.
The first mode is split into \emph{phases}, each of length $K$.
During phase $i \geq 0$, we play according to $f^{exp}$ for at most $K$ steps.
There are two conditions that can trigger a switch to the second mode:
\begin{itemize}
	
	\item[\textbf{[Switching condition 1 (SC1)]}]
	If we are in phase $i \geq 1$ and the total payoff since the beginning of the play
	is not always $> \vec N_i$ during the current phase,
	then switch to $f^{wc}$ permanently.
	
	\item[\textbf{[Switching condition 2 (SC2)]}]
	If the total payoff since the beginning of the play is not $> 2 \cdot \vec N_{i+1}$
	at the end of the current phase,
	then switch to $f^{wc}$ permanently.
	
\end{itemize}
What it remains to do is to show that we can choose $K > 0$ in order to satisfy the BWC objective.
First, we show that, for every choice of the parameter $K$,
the combined strategy $f_K$ guarantees the worst-case objective.
\begin{proposition}
	\label{prop:infinite:memory:wc}
	For every $K \in \N$ and state $s_0$ in the EC $U$,
	$\vec 0 \in \WCSolP \Game {s_0, f_K}$.
\end{proposition}
\begin{proof}
	There are two cases to consider.
	If we ever switch to the second mode,
	then the run is eventually consistent with the worst-case strategy $f^{wc}$,
	which guarantees worst-case mean payoff $>\vec 0$
	(by prefix independence).
	Otherwise, assume that we never leave the first mode.
	During phase $i \geq 1$ the total payoff is always $> \vec N_i = \frac {\vec \nu \cdot i \cdot K} 2$,
	and the total length of the play is at most $i \cdot K$.
	The average mean payoff during phase $i$ is uniformly $> \frac {\vec \nu} 2$.
	The limit inferior of the average mean payoff is also $> \frac {\vec \nu} 2 \geq \vec 0$.
\end{proof}


\noindent
We conclude by showing that $K$ can be chosen s.t. the combined strategy $f_K$
achieves expected mean payoff $> \vec \nu$.
\begin{lemma}
	\label{lemma:infinite:memory:exp}
	There exists $K \in \N$ s.t. $\vec \nu \in \ExpSolP \Game {s_0, f_K}$.
\end{lemma}

\begin{proof}
	We show that we can choose a $K > 0$ large enough s.t. the switching probability is negligible,
	and thus the impact of switching to the worst-case strategy $f^{wc}$ on the expected mean payoff is also negligible.
	For now fix an arbitrary $K > 0$,
	and consider the Markov chain $\Game[f_K]$.
	Let $p_K$ be the probability to switch to the second mode due to SC1 in any phase $i \geq 1$,
	and let $q_K$ be the probability to switch to the second mode due to SC2 in any phase $i \geq 0$.
 	Thus, with probability at most $1-(1-p_K) \cdot (1-q_K)$ we switch to the second mode.
	By prefix independence of the mean payoff objective,
	the expected mean payoff achieved by $f_K$ satisfies:
	\begin{align*}
		\E {s_0, f_K} {\Game} \MP \geq 
			\quad 	&(1-(1-p_K) \cdot (1-q_K)) \cdot \E {s_0, f^{wc}} {\Game} \MP + \\
		    +	&(1-p_K) \cdot (1-q_K) \cdot \E {s_0, f^{exp}} {\Game} \MP
	\end{align*}
	%
	Since $\E {s, f^{exp}} {\restrict \Game U} \MP > \vec\nu$ by definition,
	it suffices to show that both probabilities $p_K$ and $q_K$ can be made arbitrarily small.
	We argue about them separately.
	
	Let $p_{i, K}$ be the probability of switching to the second mode due to SC1 during phase $i \geq 1$,
	i.e, the probability that the total payoff goes below $\vec N_i$ in any component:
	\begin{align*}
		p_{i, K} = \Prob {s_0} {\Game[f_K]}
			{\exists (K \cdot i \leq h < K \cdot (i+1)) \cdot \TP_h \not> \vec N_i}
	\end{align*}
	Then, $p_K = p_{1, K} + (1-p_{1, K}) \cdot p_{2, K} + (1-p_{1, K})\cdot(1-p_{2, K})\cdot p_{3, K} + \dots$,
	and thus $p_K \leq p_{1, K} + p_{2, K} + \dots$.
	We claim the following exponential upper bound on $p_{i, K}$.
	\begin{restatable}{claim}{claimUpperBoundA}
		There are rational constants $a$ and $b$ with $b < 1$
		s.t., for every $i \geq 1$ and for sufficiently large $K$,
		%
			$p_{i, K} \leq a \cdot b^{K \cdot i}$.
		%
		Note that $a$ and $b$ do not depend neither on $K$, nor on $i$.
	\end{restatable}
	\noindent
	By the claim, $p_K \leq a \cdot (b^{K} + b^{K \cdot 2} + \dots) \leq a \cdot b^K / (1 - b^K)$,
	and thus $\lim_K p_K = 0$ since $b^K < 1$.
	
	Let $q_{i, K}$ be the probability of switching to the second mode due to SC2 at the end of phase $i \geq 0$.
	Thus, $q_{i, K}$ is the probability that, at the end of phase $i$,
	the total payoff is less than $2 \cdot \vec N_{i+1} = \vec\nu \cdot (i+1) \cdot K$ in any component:
	\begin{align*}
		q_{i, K} = \Prob {s_0} {\Game[f_K]} {\TP_{K \cdot (i+1)} \not> 2 \cdot \vec N_{i+1}}
	\end{align*}
	We have $q_K = q_{0, K} + (1-q_{0, K}) \cdot q_{1, K} + (1-q_{0, K})\cdot(1-q_{1, K})\cdot q_{2, K} + \dots$.
 	We show $\lim_K q_K = 0$ as in the last paragraph, by the following claim.
	\begin{restatable}{claim}{claimUpperBoundB}
		There exist rational constants $a$ and $b$ with $b < 1$
		s.t., for every $i \geq 0$ and sufficiently large $K$,
		%
			$q_{i, K} \leq a \cdot b^{K \cdot (i+1)}$.
		%
		Note that $a$ and $b$ do not depend neither on $K$, nor on $i$. 
	\end{restatable}
	\noindent
	Both claims are proved in the appendix.
\end{proof}

\ignore{
\begin{proof}[Proof of Lemma~\ref{lem:BWC:synthesis:EC}]
	Take as $f$ the strategy $f_K$ constructed above,
	where $K$ is as given by Lemma~\ref{lem:infinite:memory:exp}.
	Then, by definition $\nu \in \ExpSolP \Game {s_0, f}$,
	and by Proposition~\ref{prop:infinite:memory:wc},
	$0 \in \WCSolP \Game {s_0, f}$.
	Therefore, $(\vec 0; \vec \nu) \in \BWCSolP \Game {s_0, f}$.
\end{proof}
}

\subsubsection{The general case}

\begin{figure*}

	\begin{align}
		\tag*{(A1)}
		1_{s_0}(s) + \sum_{(r, s) \in E} y_{rs} &= \sum_{(s, t) \in E} y_{st} + y_s
			&& \forall s \in S
		\label{eq:y:flow'}
		\\
		\tag*{(A1')}
		y_{st} &= R(s)(t) \cdot ( \!\! \sum_{(r, s) \in E} \!\!\!\! y_{rs} - y_s)
			&& \forall (s, t) \in E \textrm{ with } s \in S^R
		\label{eq:y:random'}
		\\
		\tag*{(A2\textrm{-bis})}
		\sum_{ \textrm{ MEC } U }\sum_{s \in U} y_s &= 1
		\label{eq:y:sum'}
		\\
		\tag*{(B\textrm{-bis})}
		\sum_{s \in U} y_s &= \sum_{(r, s) \in E\cap U \times U} x_{rs}
		 	&& \forall \textrm{ MEC } U 
		\label{eq:x:y'}
		\\
		\tag*{(C1)}
		\sum_{(r, s) \in E} x_{rs} &= \sum_{(s, t) \in E} x_{st}
			&& \forall s \in S 
		\label{eq:x:flow'}
		\\
		\tag*{(C1')}
		x_{st} &= R(s)(t) \cdot \sum_{(r, s) \in E} x_{rs}
			&& \forall (s, t) \in E \textrm{ with } s \in S^R 
		\label{eq:x:random'}
		\\
		\tag*{(C2)}
		\sum_{(s, t) \in E} x_{st} \cdot w(s, t)[i] & > \vec \nu[i] 
			&& \forall (1 \leq i \leq n) 
		\label{eq:MP:global'}
		\\
		\tag*{(C3\textrm{-bis})}
		\sum_{(s, t) \in E\cap U \times U} \!\!\!\!\!\!\!\! x_{st} \cdot w(s, t)[i] & > 0
			&& \forall \textrm{ MEC } U, 1 \leq i \leq n
		\label{eq:MP:local'}
	\end{align}
	
	\caption{Linear system $T'$ for the BWC infinite-memory threshold problem.}
	\label{fig:BWC:infinite_memory:system}
	
\end{figure*}

As in the synthesis for finite-memory strategies (cf. Sec.~\ref{sec:finite-memory}),
we reduce the infinite-memory BWC problem to the solution of a system of linear inequalities.
The new system of equations $T'$ is shown in Fig.~\ref{fig:BWC:infinite_memory:system}.
It is obtained as a modification of system $T$ from the finite-memory case shown in Fig.~\ref{fig:BWC:system}:
Specifically, $T'$ is the same as $T$,
except that Equations~\ref{eq:y:sum}, \ref{eq:x:y}, and \ref{eq:MP:local} are interpreted w.r.t. MEC (instead of MWEC).
The correctness of the reduction is stated in the lemma below.
\begin{restatable}{lemma}{lemInfiniteMemoryGeneral}
	\label{lem:infinite-memory:general}
	Let $\Game$ be a pruned multidimensional mean-payoff \MDP,
	let $\vec \nu \geq \vec 0$, and let $s_0 \in S$.
	There exists a (possibly infinite-memory) strategy $h$ s.t. $(\vec 0; \vec \nu) \in \BWCSolP \Game {s_0, h}$,
	if, and only if, the system $T'$ has a non-negative solution.
\end{restatable}
\begin{proof}[Proof sketch]
	The proof is analogous to the proof of Lemma~\ref{lem:finite-memory:general}.
	The crucial difference is that, by the modifications performed to obtain $T'$ from $T$,
	we obtain strategies which almost surely stay forever inside ECs, instead of WECs.
	Since we are allowed infinite-memory, we can approximate the BWC objective inside ECs
	by replacing Lemma~\ref{prop:finite-mem:WEC} with Lemma~\ref{lem:BWC:synthesis:EC}.
\end{proof}

\begin{runningexample}
	An infinite-memory strategy can benefit both from the expectation $(5, 15)$ in $U$
	and from $(15, 5)$ in $V$ (which is not a WEC).
	By going to either EC with equal probability and playing according to a local BWC strategy,
	an infinite-memory strategy can ensure, for every $\varepsilon > 0$,
	the BWC threshold $((0, 0);(10-\varepsilon, 10-\varepsilon))$.
\end{runningexample}

\subsubsection{Complexity}

We obtain the following complexity result for the threshold problem for arbitrary controllers.
\begin{theorem}
	The multidimensional mean-payoff BWC threshold problem is \coNPc.
\end{theorem}

\begin{proof}
	Pruning the game to remove states which are losing for the worst-case objective requires solving a multidimensional mean-payoff game, which is \coNPc by Theorem~\ref{thm:multidimensional:worst-case:complexity}.
	Then, by Lemma~\ref{lem:infinite-memory:general},
	it suffices to solve system $L'$.
	Notice that system $L'$ is of polynomial size since there are only polynomially many \emph{maximal} ECs.
\end{proof}

Again, it is the worst-case problem that dominates the complexity of the BWC problem.
By restricting to essentially worst-case unidimensional thresholds we obtain a better complexity,
as in Sec.~\ref{sec:BWC:finite-memory:complexity}.
\begin{corollary}
	The multidimensional mean-payoff BWC threshold problem
	w.r.t. \emph{essentially worst-case unidimensional thresholds} is in \NPcoNP.
\end{corollary}
\noindent
This solves with optimal complexity the infinite-memory unidimensional BWC problem,
which was left open in \cite{Bruyere:Filiot:Randour:Raskin:BWC:2014}.

\section{Beyond almost-sure synthesis}

\label{sec:BAS}

We introduce a natural relaxation of the BWC problem which enjoys a better complexity.
Intuitively, we replace the worst-case objective in the BWC problem with a weaker almost sure objective.
While the BWC problem is \coNPc, we show that this relaxation can be solved in \PTIME,
even in the multidimensional setting.
A similar result has recently been obtained in \cite{ChatterjeeKomarkovaKretinsky:Unifying:2015}. 
Given an \MDP $\Game$,
a starting state $s_0$ therein,
and a Controller's strategy $f \in \All \Game$,
the set of \emph{almost sure achievable solutions for $f$},
denoted $\ASSolP \Game {s_0, f}$,
is the set of vectors $\vec\mu \in \R^d$ s.t.
Controller can almost surely guarantee mean payoff $>\vec\mu$ when playing according to $f$, i.e.,
$\ASSolP \Game {s, f} = \SetOf {\vec \mu \in \R^d}
	{\Prob {s, f} \Game {\MP > \vec \mu} = 1}$.
The set of \emph{beyond almost-sure achievable solutions for $f$},
denoted $\BASSolP \Game {s_0, f}$,
is the set of pairs of vectors $(\vec \mu; \vec \nu) \in \R^{2d}$ s.t.
Controller can almost surely guarantee mean payoff $>\vec \mu$
and achieve expected mean payoff $>\vec \nu$
when starting from $s_0$ and playing according to $f$, i.e.,
\begin{align*}
	\BASSolP \Game {s_0, f} = \SetOf {(\vec \mu; \vec \nu) \in \R^{2d}}
	{\begin{array}{c}
		\vec \mu \in \ASSolP \Game {s_0, f} \\
			\textrm { and } \\
		\vec \nu \in \ExpSolP \Game {s_0, f}
	\end{array}}
\end{align*}
The set of \emph{beyond almost-sure achievable solutions} is
$\BASSolP \Game {s_0} = \bigcup_{f \in \All \Game} \BASSolP \Game {s_0, f}$.
%
%
Given $(\vec \mu; \vec \nu) \in \R^{2d}$ and a state $s_0$,
the \emph{beyond almost-sure threshold problem} asks whether $(\vec \mu; \vec \nu) \in \BASSolP \Game {s_0}$.

\begin{remark}
	We assume w.l.o.g. that $\vec \mu = \vec 0$
	and $\vec \nu \geq \vec 0$.
	The first condition is ensured by subtracting $\vec\mu$ everywhere.
	The second condition follows from the observation that,
	if the mean payoff is $> \vec 0$ almost surely,
	then also the expectation is $> \vec 0$ surely.
\end{remark}

\noindent
We observe that, inside an EC, there is no trade-off between the almost sure and the expectation objective.
%
\begin{lemma}
	\label{lem:BAS:synthesis:EC}
	Let $\Game$ be a multidimensional mean-payoff \MDP,
	let $s_0$ be a state in an EC $U$ thereof,
	and let $\vec \nu \in \ExpSolP {\restrict \Game U} {s_0}$
	be an expectation achievable while remaining inside $U$.
	There exists a \emph{finite-memory} strategy $g \in \FiniteMemory \Game$
	s.t. $(\vec \nu; \vec \nu) \in \BASSolP {\restrict \Game U} {s_0, g}$
	which also remains inside $U$.
\end{lemma}
\begin{proof}
	By Lemma~\ref{lem:strategy:EC:finite-memory:expectation},
	there exists a finite-memory strategy $g$ s.t.
	$\vec \nu \in \ExpSolP {\restrict \Game U} {s_0, g}$
	and $\Game[g]$ is unichain.
	Consequently, the mean payoff is $> \vec\nu$ almost surely.
\end{proof}

Thus, most of the effort goes in analyzing the general case.
As in the BWC problem, we reduce the BAS problem to the solution of a system of linear inequalities.
We assume that from state $s_0$ all ECs are reachable with positive probability.
It turns out that the same system of equations $T'$ used in the infinite-memory BWC threshold problem also solves the BAS problem.
We obtain a better complexity since we do not require the MDP to be pruned (which avoids solving an expensive mean-payoff game).

\begin{theorem}
	\label{thm:robust-expectation}
	The multidimensional mean-payoff BAS threshold problem is in \PTIME.
\end{theorem}

\begin{proof}
	The proof of correctness is the same as in Lemma~\ref{lem:infinite-memory:general},
	where Lemma~\ref{lem:BAS:synthesis:EC} replaces Lemma~\ref{lem:BWC:synthesis:EC} in the analysis of ECs.
	Crucially for complexity,
	we do not need to assume that the MDP is pruned.
	Therefore, system $T'$ can be built (and solved) in \PTIME.
	Since Lemma~\ref{lem:BAS:synthesis:EC} even yields finite-memory strategies inside an EC,
	the construction of Lemma~\ref{lem:infinite-memory:general} shows that finite-memory strategies suffice for the BAS threshold problem.
	(This relies on the strict BAS semantics.
	If non-strict inequalities are used, then the problem can still be solved in \PTIME
	but the construction above yields an infinite-memory strategy,
	and infinite-memory strategies are more powerful than finite-memory ones for the non-strict BAS problem;
	cf. also \cite{ChatterjeeKomarkovaKretinsky:Unifying:2015}.)
	%
\end{proof}

\begin{runningexample}
	The BAS problem is strictly weaker than the BWC problem.
	Consider the \MDP from Fig.~\ref{fig:running_example} without the edge $(u, t)$.
	This modification makes both states $u$ and $v$ losing for the worst-case,
	thus they are pruned away when solving the BWC problem (even with infinite memory).
	On the other hand, the mean payoff is almost surely $(5, 15)$ from $V$,
	and thus it satisfies the almost sure objective $>(0, 0)$.
	Therefore, for every $\varepsilon > 0$,
	we can achieve the BAS threshold $((0, 0);(10-\varepsilon, 10-\varepsilon))$
	by going to $t$ or $u$ with equal probability.
\end{runningexample}

\section{Conclusions}

\label{sec:conclusions}

In this paper, we studied the multidimensional generalization of the beyond worst-case problem
introduced by Bruy\`ere \emph{et al.}~\cite{Bruyere:Filiot:Randour:Raskin:BWC:2014}.
We have provided tight \coNP-completeness results for the this problem under both finite-memory and general strategies.
Since multidimensional mean-payoff games are already \coNPc,
our upper bound shows that we can add a multidimensional expectation optimization objective
on top of a worst-case requirement without a corresponding increase in complexity.
Notice that, while infinite-memory strategies were known to be more powerful than finite-memory ones already in the unidimensional setting~\cite{Bruyere:Filiot:Randour:Raskin:BWC:2014},
the corresponding synthesis problem was left open.
Our results thus complete the complexity picture for this problem.
Moreover, we showed that, when the worst-case objective is unidimensional,
the complexity reduces to \NPcoNP,
and this holds even for multidimensional expectations.
This generalizes with optimal complexity the \NPcoNP upper bound for the unidimensional beyond worst-case problem~\cite{Bruyere:Filiot:Randour:Raskin:BWC:2014}.
From a practical point of view,
our reductions to linear programming can be performed in pseudo-polynomial time
by using the results of \cite{DBLP:journals/fmsd/BrimCDGR11} for unidimensional mean-payoff games,
and \cite{DBLP:conf/concur/ChatterjeeV13} for \emph{fixed} number of dimensions.
%
Furthermore, we introduced the beyond almost-sure problem as a natural relaxation of the beyond worst-case problem,
by weakening the worst-case requirement to an almost-sure one.
This natural relaxation enjoys a polynomial time solution and finite memory strategies always suffice.
Moreover, our reduction to linear programming shows that the beyond almost-sure problem is amenable to be solved efficiently in practice,
and thus it has the strongest appeal for practical applications.

\bibliographystyle{IEEEtran}
\bibliography{bibliography}

\clearpage
\onecolumn
\appendices

\section{Preliminaries}

\label{app:preliminaries}

\subsection{Hoeffding-style bounds}

In this section, we prove an Hoeffding-style bound for multidimensional Markov chains which will be used repeatedly in later proofs.
Let $\Game$ be a Markov chain which is unichain.
Recall that a Markov chain is unichain when it consists of transient states and a unique bottom strongly connected component.
Thus, any run in $\Game$ will be trapped almost surely in the bottom component,
and the mean payoff will be almost surely equal to the expected mean payoff.
Moreover, the expected mean payoff is the same from every starting state of $\Game$.
Below, we present a bound on the probability that the mean payoff deviates from the expected mean payoff for sufficiently long runs.
Let $d \geq 1$ be the dimension of the weights in the Markov chain $\Game$,
and let $\vec\nu$ be the expected mean-payoff vector from any state in $\Game$.
\begin{restatable}{lemma}{hoeffding}
	\label{lem:Hoeffding:multidimensional}
	For any $\delta > 0$,
	there exists $K_0 = O(\frac 1 {\delta^2}) \in \N$ and constants $a, b > 0$ s.t.,
	for all $K \geq K_0$ and state $s$,
	\begin{align*}
		\Prob s \Game
		{ \exists (1 \leq j \leq d) \cdot | (\MP_K - \vec\nu)[j] | \geq \delta}
		\leq 
		\G K \delta := 2^d \cdot a \cdot e^{-b \cdot K \cdot \delta^2} \ .
	\end{align*}
	%
	%
	Moreover, $a$ and $b$ are polynomial in the parameters of the Markov chain,
	$a$ is exponential in $\delta$,
	and $K_0$ is polynomial in the size of the Markov chain and in the largest absolute weight $W$
	(and thus exponential in its encoding).	
\end{restatable}

\noindent
We prove the lemma above by reducing to the unidimensional case $d = 1$.
The latter case was already dealt with in \cite[Lemma 9]{BruyereFiliotRandourRaskin:BWC:arxiv},
which in turns relies on \citeappendix[Proposition 2]{Tracol:2009}.
\begin{lemma}[cf. {\cite[Lemma 9]{BruyereFiliotRandourRaskin:BWC:arxiv}}]
	\label{lem:Hoeffding}
	For any $\delta > 0$,
	there exist $K_0 = O(\frac 1 {\delta}) \in \N$ and constants $a, b > 0$ s.t.,
	for all $K \geq K_0$ and state $s$,
	\begin{align*}
		\Prob s \Game
		{ | \MP_K - \nu | \geq \delta}
		\leq \F {a, b} K \delta := a \cdot e^{-b \cdot K \cdot \delta^2} \ .
	\end{align*}
	Moreover, $a$ and $b$ are polynomial in the parameters of the Markov chain,
	$a$ is exponential in $\delta$.
	and, $K_0$ is polynomial in the size of the Markov chain and in the largest weight $W$
	(and thus exponential in its encoding).
\end{lemma}


\begin{proof}[Proof of Lemma~\ref{lem:Hoeffding:multidimensional}]
	In the following, fix an error $\delta > 0$ and a number of steps $K > 0$.
	For a component $1 \leq j \leq d$,
	we say that a run $\pi$ is \emph{$j$-bad} iff
	the $j$-th component of the mean payoff deviates from $\vec\nu[j]$ at least by $\delta$ after $K$ steps,
	i.e., if
	\begin{align*}
		|(\MP(\pi(K)) - \vec\nu)[j] | \geq \delta
	\end{align*}
	and $\pi$ is $j$-good otherwise.
	Moreover, we say that $\pi$ is \emph{bad} if it is $j$-bad for some $1 \leq j \leq d$,
	and we say that $\pi$ is \emph{good} it is $j$-good for every $1 \leq j \leq d$.
	In other words, $\pi$ is good if for every component $j$,
	$|(\MP(\pi(K)) - \vec\nu)[j] | < \delta$.

	For a fixed dimension $1 \leq j \leq d$, we are in the unidimensional case, dealt with in Lemma~\ref{lem:Hoeffding}.
	Let $j$ be a fixed dimension.
	By Lemma~\ref{lem:Hoeffding}, there exist constants $a_j, b_j > 0$, and $K_j = O(\frac 1 {\delta}) \in \N$ s.t.,
	for every $K \geq K_j$ and state $s$,
	$\F {a_j, b_j} K \delta$ is an upper bound on the probability that $\pi$ is $j$-bad. 
	We want to choose uniform $a, b > 0$ and $K_0$ s.t.
	\begin{align*}
		\F {a_j, b_j} K \delta \leq \F {a, b} K \delta < 1
	\end{align*}
	for every $K \geq K_0$.
	To this end, let
	\begin{align*}
		a		&:= \max_j a_j, \\
		b		&:= \min_j b_j, \textrm{ and } \\
		K_0		&:= \max \set{\max_j K_j, \ceiling{\frac {\ln a}{b \cdot \delta^2}} + 1} \ .
	\end{align*}
	Note that $K_0 = O(\frac 1 {\delta^2})$.
	Then, $1 - \F {a, b} K \delta$ is a lower bound on the probability that $\pi$ is $j$-good, for any fixed $j$ and $K \geq K_0$.
	Then, $(1 - \F {a, b} K \delta)^d$ is a lower bound on the probability that $\pi$ is good.
	We derive the following simple lower bound on the latter quantity:
	\begin{align*}
		(1 - \F {a, b} K \delta)^d & = \sum_{i=0}^d {d \choose i}(-\F {a, b} K \delta)^i = \\
		& = 1 + \sum_{i=1}^d {d \choose i}(-\F {a, b} K \delta)^i \\
		& \geq 1 - \sum_{i=1}^d {d \choose i}(\F {a, b} K \delta)^i \\
		& \geq 1 - \sum_{i=1}^d {d \choose i}\F {a, b} K \delta \\
		& \geq 1 - 2^d \cdot \F {a, b} K \delta
	\end{align*}
	Finally, $1 - (1 - \F {a, b} K \delta)^d$ is an upper bound on the probability that $\pi$ is bad.
	We define
	\begin{align*}
		\G K \delta := 2^d \cdot \F {a, b} K \delta = 2^d \cdot a \cdot e^{-b \cdot K \cdot \delta^2}
	\end{align*}
	By the inequality above, $\G K \delta \geq 1 - (1 - \F {a, b} K \delta)^d$,
	and thus $\G K \delta$ is an upper bound on the probability that a period is bad
	for every $K \geq K_0$.
\end{proof}

\subsection{Mean-payoff value function}

We make a couple of simple observations on the relationship between the mean-payoff value function and the total-payoff value function.

\begin{lemma}
	\label{lem:MP:TP:basic}
	Let $\pi$ be a play in a graph $\Graph$.
	\begin{enumerate}[1)]
		
		\item If $\MP(\pi) > 0$, then $\TP(\pi) = + \infty$.
		Therefore, if $\TP(\pi) < + \infty$ (with possibly $\TP(\pi) = -\infty$),
		then $\MP(\pi) \leq 0$.
		
		\item Similarly, if $\MP(\pi) < 0$, then $\TP(\pi) = - \infty$.
		
		\item There exists a play $\pi_0$ in a finite graph $\Graph_0$
		s.t. $\MP(\pi_0) = 0$, but $\TP(\pi_0) = + \infty$.
		
		\item There exists  a play $\pi_1$ in a finite graph $\Graph_1$
		s.t. $\MP(\pi_1) = 0$, but $\TP(\pi_1) = - \infty$.
	\end{enumerate}
\end{lemma}
\begin{proof}
	We first prove Point 1).
	Assume $\MP(\pi) = a$ for some $a > 0$.
	By the definition of $\liminf$,
	for every $\varepsilon > 0$ there exists $m_0$ s.t. for every $m \geq m_0$
	we have $\MP(\pi(m)) \geq a - \varepsilon$.
	To show $\TP(\pi) = + \infty$ we show that for every bound $b > 0$,
	there exists $n_0$ s.t. for every $n \geq n_0$,
	$\TP(\pi(n)) \geq b$.
	Let $b > 0$. If we take $\varepsilon := a / 2$ in the definition above,
	we have that there exists $m_0$ s.t. for every $m \geq m_0$
	$\MP(\pi(m)) \geq a/2 > 0$.
	We take $n_0 := \max \set {m_0, 2 b / a}$, and let $n \geq n_0$.
	Since $\TP(\pi(n)) = n \cdot \MP(\pi(n))$,
	we have $\TP(\pi(n)) \geq n a / 2 \geq b$,
	where the latter inequality follows from the definition of $n_0$.
	
	The proof of Point 2) is analogous to the proof of Point 1).
	
	For Point 3), consider a play $\pi_0$ inducing the following sequence of payoffs:
	\begin{align*}
		\underbrace{1}_{2^0}
		\underbrace{1 0}_{2^1}
		\underbrace{1 0 0 0}_{2^2}
		\underbrace{1 0 0 0 0 0 0 0}_{2^3} \cdots
	\end{align*}
	i.e., the $n$-th payoff is $1$ if $n$ is a power of $2$, and $0$ otherwise.
	Then, $\TP(\pi_0(n)) = k$ where $k$ is the largest exponent s.t. $2^k \leq n$,
	i.e., $k = \lfloor \lg n \rfloor$.
	Thus, $\TP(\pi_0) = + \infty$.
	However, $\MP(\pi_0(n)) = \frac {\lfloor \lg n \rfloor} n$
	goes to $0$ as $n$ goes to $ + \infty$.
	Thus, $\MP(\pi_0) = 0$.
	Point 4) is proved analogously by taking the sequence
	\begin{align*}
		\underbrace{(-1)}_{2^0}
		\underbrace{(-1) 0}_{2^1}
		\underbrace{(-1) 0 0 0}_{2^2}
		\underbrace{(-1) 0 0 0 0 0 0 0}_{2^3} \cdots
	\end{align*}%
\ignore{	
	Finally, for Point 4), let $\pi_1$ be a play inducing the following sequence of payoffs:
	\begin{align*}
		\underbrace{-1 +1}_2
		\underbrace{-1 -1 +1 +1}_4
		\underbrace{-1 -1 -1 +1 +1 +1}_6 \cdots
	\end{align*}
	i.e., the sequence is split into blocks,
	where the $k$-th block (starting at $k = 1$)
	has length $2k$ and starts at position $\sum_{i=1}^{k-1}2i = k \cdot (k-1)$.
	Thus, position $n$ (starting at $n = 0$) belongs to block $k = \left\lfloor \frac {1 + \sqrt{1+4n}} 2 \right\rfloor$.
	Clearly, $\TP(\pi) = - \infty$ since there are longer and longer stretches of payoffs $-1$ along the sequence,
	so the total payoff gets arbitrarily low.
	Since at the beginning of each block (at positions of the form $k \cdot(k-1)$)
	the total payoff is $0$, at position $n$ in block $k$ the total payoff is at least
	$\TP(\pi_1(n)) \geq -k = - \left\lfloor \frac {1 + \sqrt{1+4n}} 2 \right\rfloor$
	(with equality if we are in the middle of the block),
	and thus the mean payoff is at least
	$\MP(\pi_1(n)) \geq -\frac k n = - \frac 1 n \cdot \left\lfloor \frac {1 + \sqrt{1+4n}} 2 \right\rfloor$
	for every $n$.
	Since the latter expression goes to $0$ as $n$ goes to $ + \infty$,
	we deduce $\MP(\pi_1) \geq 0$.
	But the mean payoff is $0$ infinitely often (at the beginning of each block),
	thus $\MP(\pi_1) = 0$.
}%
\end{proof}

As an application of Lemma~\ref{lem:Hoeffding},
we show that if the total payoff is $-\infty$ almost surely,
then the mean payoff is strictly negative almost surely.
This contrasts with Point 4) in Lemma~\ref{lem:MP:TP:basic},
which showed that there are infinite runs with total payoff equal to $-\infty$,
but which have nonetheless zero mean payoff.
(Notice that the infinite play constructed in the proof of the latter lemma with this property was non-periodic.)
We use the lemma below later in the proof of Lemma~\ref{lem:Etessami:approx}.

\begin{lemma}
	\label{lem:MP:TP:basic2}
	Let $\Game$ be a Markov chain.
	For every state $s_0$,
	\begin{align*}
		\Prob {s_0} \Game { \TP = -\infty } = \Prob {s_0} \Game { \MP < 0 } \ .
	\end{align*}
	In particular, if the mean payoff is non-negative almost surely,
	then the total payoff is $> -\infty$ almost surely.
\end{lemma}
\begin{proof}
	Fix a state $s_0$.
	Let $p = \Prob {s_0} \Game { \MP < 0 }$,
	$q = \Prob {s_0} \Game { \TP = -\infty }$,
	and, for a BSCC $B$, let $p_B$ be the probability of reaching $B$.
	In a BSCC $B$, $\MP$ and $\TP$ take value equal to their respective expectations almost surely,
	and this value is the same from every state in the BSCC.
	Let this value be $\E B \Game \MP$ and $\E B \Game \TP$, respectively.
	We thus have the following decomposition:
	\begin{align*}
		p &= \sum_{\textrm{BSCC } B \textrm{ s.t. } \E B \Game \MP < 0 } p_B, \textrm{ and } \\
		q &= \sum_{\textrm{BSCC } B \textrm{ s.t. } \E B \Game \TP = -\infty } p_B \ .
	\end{align*}
	It suffices to show that
	$\E B \Game \MP < 0$ if, and only if, $\E B \Game \TP = -\infty$
	for all BSCC's $B$.
	
	If $\E B \Game \MP < 0$, then $\MP < 0$ almost surely since $B$ is a BSCC.
	By Point 2) of Lemma~\ref{lem:MP:TP:basic},
	$\MP < 0$ implies $\TP = -\infty$ surely,
	and thus $\TP = -\infty$ holds almost surely,
	and consequently $\E B \Game \TP = -\infty$.
	
	For the other direction, let $\E B \Game \MP \geq 0$.
	If $\E B \Game \MP > 0$, then by Point 1) of Lemma~\ref{lem:MP:TP:basic} 
	and reasoning as above, we obtain $\E B \Game \TP = +\infty$.
	It remains to prove the case $\E B \Game \MP = 0$.
	This makes use of the bound provided by Lemma~\ref{lem:Hoeffding},
	and it does not hold in a non-probabilistic setting
	(cf. the counter-example in Point 4) of Lemma~\ref{lem:MP:TP:basic}).
	Assume $\E B \Game \MP = 0$.
	We prove $\E B \Game \TP > -\infty$.
	For every $s_1 \in B$ and $K$, we have
	\begin{align*}
		\Prob {s_1} \Game {\TP_K \leq -K}
			&\leq \Prob {s_1} \Game {\MP_K \leq -1} \\
			&\leq \Prob {s_1} \Game {|\MP_K| \geq 1} \\
			&\leq \F {a, b} K 1
	\end{align*}
	where the last inequality follows from Lemma~\ref{lem:Hoeffding}
	applied with $\nu = \E B \Game \MP = 0$, $\delta = 1$,
	for some constants $a, b > 0$ and for $K$ sufficiently large.
	Since $\F {a, b} K 1 \to 0$ as $K \to \infty$,
	we have that $\Prob {s_1} \Game {\TP = -\infty} = 0$,
	which implies $\E B \Game \TP > -\infty$.	
\end{proof}

\subsection{Finite-memory synthesis in an EC}

\label{sec:exp:EC}

In this section, we show that achievable values can be approximated by randomized finite-memory strategies in ECs,
with the further property that the induced (finite) Markov chain is unichain,
i.e., it contains exactly one BSCC.

\ECfinitememoryexpectation*

We prove this result as follows.
First, in Sec.~\ref{sec:SCCs:decomposition} we characterize the set of of achievable vectors
as non-negative solutions to a linear programming problem
(in the spirit of \cite{BrazdilBrozekChatterjeeForejtKucera:TwoViews:LMCS:2014}).
This yields a natural decomposition of the EC into several SCCs.
For each such SCC, we construct in Sec.~\ref{sec:exp:local} a \emph{randomized memoryless} ``local strategy''
achieving a corresponding ``local expectation''.
No approximation error is introduced in this step.
Then, in Sec.~\ref{sec:exp:global} we combine those ``local strategies''
into a randomized finite-memory ``global strategy'' approximating the expectation.
This second step uses the fact that in an EC all states are inter-reachable (under some strategy),
and thus we can cycle through all the ``local strategies'' for the appropriate fraction of time.
This step introduces an approximation error,
due to the cost of moving from a SCC to the next one.
However, by using larger amounts of finite memory,
we can make this error arbitrarily small.

\ignore{

\begin{figure}
	\centering
	\footnotesize

	\begin{minipage}{.5\textwidth}
		\centering
		\begin{tikzpicture}

			\node[playerstate] (s) [] {$s$};
			\node[playerstate] (t) [right = 1 cm of s] {$t$};

			\path[->] (s) edge [loop above] node [above = 1ex] {$(0, 1)$} ();
			\path[->] (s) edge [bend right = 45] node [above = 1ex] {$(0, 0)$} (t);
		
			\path[->] (t) edge [loop above] node [above = 1ex] {$(1, 0)$} ();
			\path[->] (t) edge [bend right = 45] node [above = 1ex] {$(0, 0)$} (s);

		\end{tikzpicture}
		\subcaption{A \MDP reduced to one EC.}
		\label{fig:randomized_EC_example}
	\end{minipage}
\\[4ex]

	\begin{minipage}{.5\textwidth}
		\centering
		\begin{tikzpicture}

			\node[randomstate] (s0) [] {$s, 0$};
			\node[randomstate] (s1) [left = 1 cm of s0] {$s, 1$};
			\node[randomstate] (t) [right = 1 cm of s0] {$t$};

			\path[->] (s0) edge node [above] {$\frac 1 2$} node [below] {$(0, 1)$} (s1);
			\path[->] (s0) edge node [above] {$\frac 1 2$} node [below] {$(0, 0)$} (t);
		
			\path[->] (s1) edge [loop above] node [above = 1ex] {$(0, 1)$} ();
			\path[->] (t) edge [loop above] node [above = 1ex] {$(1, 0)$} ();

		\end{tikzpicture}
		\subcaption{Exact strategy inducing two BSCCs.}
		\label{fig:randomized_EC_2mem_example}
	\end{minipage}
\\[4ex]

	\begin{minipage}{.5\textwidth}
		\centering
		\begin{tikzpicture}

			\node[randomstate] (s0) 	[] {$s, 1$};
			\node[randomstate] (s1) 	[right = 1 cm of s0] {$s, 2$};
			\node[] (sdots)	[right = 1 cm of s1] {$\cdots$};
			\node[randomstate] (sA) 	[right = 1 cm of sdots] {$s, A$};
		
			\node[randomstate] (t0) 	[below = .5 cm of sA] {$t, 1$};
			\node[randomstate] (t1) 	[left = 1 cm of t0] {$t, 2$};
			\node[] (tdots)	[left = 1 cm of t1] {$\cdots$};
			\node[randomstate] (tA) 	[left = 1 cm of tdots] {$t, A$};

			\path[->] (s0) edge [bend left = 30] node [above = 1ex] {$(0, 1)$} (s1);
			\path[->] (s1) edge [bend left = 30] node [above = 1ex] {$(0, 1)$} (sdots);
			\path[->] (sdots) edge [bend left = 30] node [above = 1ex] {$(0, 1)$} (sA);
			\path[->] (sA) edge [bend left = 30] node [right = 1ex] {$(0, 0)$} (t0);

			\path[->] (t0) edge [bend left = 30] node [below = 1ex] {$(1, 0)$} (t1);
			\path[->] (t1) edge [bend left = 30] node [below = 1ex] {$(1, 0)$} (tdots);
			\path[->] (tdots) edge [bend left = 30] node [below = 1ex] {$(1, 0)$} (tA);
			\path[->] (tA) edge [bend left = 30] node [left = 1ex] {$(0, 0)$} (s0);

		\end{tikzpicture}
		\subcaption{Approximate finite-memory strategy inducing one BSCC.}
		\label{fig:pure_EC_example}
	\end{minipage}

	\caption{Approximating the expectation inside ECs.}
	
\end{figure}

We have seen in Lemma~\ref{lem:strategy:expectation} that randomized 2-memory strategies
are sufficient for expected-value achievable solutions,
and, by Remark~\ref{rem:expected-value:randomization}, randomization is necessary in general \MDPs .
However, 2-memory randomized strategies have the unpleasant property that they split the \MDP into several BSCCs in general.
In turn, this creates technical difficulties later,
since playing such a strategy for long enough time does not guarantee that the mean payoff is close to the expectation.

For example, consider the EC \MDP in Fig.~\ref{fig:randomized_EC_example}
(which is the same as the one in Fig.~3 of \cite{VelnerChatterjeeDoyenHenzingerRabinovichRaskin:Complexity:ArXiv:2012}).
There exists a simple randomized 2-memory strategy achieving expected mean payoff $(\frac 1 2, \frac 1 2)$
which stays forever in $s$ with probability $\frac 1 2$,
and in $t$ with probability $\frac 1 2$.
The induced Markov chain is shown in Fig.~\ref{fig:randomized_EC_2mem_example} (when starting from $s$),
and clearly the mean payoff of an infinite play is either $(0,1)$ or $(1, 0)$ with equal probability,
but never the expected value $(\frac 1 2, \frac 1 2)$.
In Sec.~\ref{sec:exp:EC:finite_memory} we design \emph{pure} finite-memory strategies
approximating the expected mean payoff and inducing unichain Markov chains (i.e., with just one BSCC),
and thus, by running such a strategy for long enough time, we can get close to the expected mean payoff with arbitrarily high probability.

Moreover, this example also shows that randomization is necessary for finite-memory strategies even \emph{inside} ECs
(whereas Remark~\ref{rem:expected-value:randomization} showed the same property in a game with two different ECs).
Indeed, no pure \emph{finite-memory} strategy can achieve mean payoff $(\frac 1 2, \frac 1 2)$,
since any finite-memory strategy switching between the two states infinitely often must take the edges from $s$ to $t$ and back
with a non-zero long-run frequency,
and thus the vector $(0, 0)$ will have a non-negligible impact in the long run.
In Sec.~\ref{sec:exp:EC:infinite_memory} we show that this can be avoided at the expense of using infinite memory,
which is used to make negligible the long-run frequency of those ``transient'' transitions.

The idea is to use finite-memory to approximate the long-run frequencies of edges resulting from playing 2-memory randomized strategies
(which exists by Lemma~\ref{lem:strategy:expectation}).
As an illustration, consider again the example in Fig.~\ref{fig:randomized_EC_example}.
To achieve expected mean payoff $(\frac 1 2, \frac 1 2)$ it must be the case that,
in the long run, edges $s \goesto {(0, 1)} s$ and $t \goesto {(1, 0)} t$ are taken with equal relative frequency $\frac 1 2$,
which can be approximated with arbitrary additive error by pure finite-memory strategies.
For example, for a parameter $A \in \N$, consider the pure finite-memory strategy $g_A$ which stays in $s$ for $A$ steps,
and then goes to $t$, stays in $t$ for $A$ steps, and then goes back to $s$, and repeats this scheme forever.
The Markov chain induced by $g_A$ is shown in Fig.~\ref{fig:pure_EC_example}.
The strategy $g_A$ uses $A$ memory states, and achieves (expected and worst-case) mean payoff
\begin{align*}
	&\frac A {2A + 2} \cdot (0, 1) + \frac 1 {2A + 2} \cdot (0, 0) + \frac A {2A + 2} \cdot (1, 0) + \frac 1 {2A + 2} \cdot (0, 0) = \\
	&\qquad = \left(\frac A {2A + 2}, \frac A {2A + 2}\right)
\end{align*}
which converges from below to $(\frac 1 2, \frac 1 2)$ as $A \to \infty$.

}

\ignore{

\begin{theorem}[\citeappendix{VelnerRabinovich:Noisy:FOSSACS:2011,VelnerChatterjeeDoyenHenzingerRabinovichRaskin:Complexity:ArXiv:2012}]
	The expected-value threshold problem for multidimensional mean-payoff \stochastic-games is \coNPc.
\end{theorem}
\begin{proof}
	While infinite memory may be needed for Controller to achieve the threshold
	\citeappendix{ChatterjeeDoyenHenzingerRaskin:Generalized:FSTTCS:2010},
	memoryless strategies suffice for the Player~1 to spoil the objective
	\citeappendix{Kopczynski:HalfPositional:2006,GimbertKelmendi:HalfPositional:ArXiv:2014}.%
	Thus, after guessing a strategy for Player~1,
	it suffices to check that in the resulting MDP Controller does not win,
	which can be done in \PTIME \cite{BrazdilBrozekChatterjeeForejtKucera:TwoViews:LMCS:2014}.
	
	We give a detailed proof below.
\end{proof}

}

\ignore{

randomized finite-memory strategies suffice for achievable solutions (cf. \cite{BrazdilBrozekChatterjeeForejtKucera:TwoViews:LMCS:2014})
\begin{lemma}
	\label{lem:strategy:expectation}
	Let $\Game$ be a multidimensional mean-payoff \MDP,
	let $s$ be a state therein,
	and let $\vec \nu \in \ExpSol \Game s$ be an achievable solution.
	Then, there exists a \emph{finite-memory randomized} strategy
	$f \in \FiniteMemory \Game$ for Controller
	s.t. $\vec \nu \in \ExpSol \Game {s, f}$.	
\end{lemma}
\noindent
In Sec.~\ref{sec:exp:EC} we complement this result by showing that, inside ECs,
randomization can be traded for infinite memory (cf. Lemma~\ref{lem:strategy:EC:infinite-memory:expectation})
and finite-memory suffices to \emph{approximate} achievable solutions (cf. Lemma~\ref{lem:strategy:EC:finite-memory:expectation});
moreover, in the special case of unichain \MDPs  (i.e., when memoryless strategies always induce unichain Markov chains),
we show that pure finite-memory strategies suffice for achievable solutions (cf. Lemma~\ref{lem:strategy:EC:irreducible:expectation}).
\begin{proof}
	In \cite{BrazdilBrozekChatterjeeForejtKucera:TwoViews:LMCS:2014} it is shown that 2-memory stochastic-update strategies
	are sufficient to ensure achievable vectors (cf. Theorem 4.1 therein),
	and, by trading randomness (in the update of the memory) for extra memory states,
	it is also shown that finite-memory deterministic-update strategies also suffice for achievable vectors
	(cf. Proposition A.2 therein).
	Our notion of randomized strategy is exactly what is called deterministic-update strategy in \cite{BrazdilBrozekChatterjeeForejtKucera:TwoViews:LMCS:2014}.
\end{proof}

}

\ignore{

\begin{remark}
	\label{rem:expected-value:randomization}
	randomization is necessary in general for achieving expected-value solutions,
	independently of the available memory.
	Consider a simple \MDP with three states $S = \set{s_0, s_1, s_2}$
	and weighted edges $s_0 \goesto {(0, 0)} s_1$, $s_0 \goesto {(0, 0)} s_2$,
	$s_1 \goesto {(0, 1)} s_1$, and $s_2 \goesto {(1, 0)} s_2$.
	There exists a randomized strategy achieving mean payoff $(1, 1)$ in expectation from $s_0$
	which goes to equal probability to either $s_1$ or $s_2$,
	but there is no pure strategy achieving $(1, 1)$.
	Notice that there are two maximal ECs, $U_0 = \set {s_1}$ and $U_1 = \set {s_2}$.
\end{remark}

}

\subsubsection{Decomposition in SCCs}

\label{sec:SCCs:decomposition}

\renewcommand{\A}{$\textrm{A}_{\vec \nu}$}
\renewcommand{\a}{\emph{(a)}}
\renewcommand{\b}{\emph{(b)}}

\begin{figure}
	\begin{align}
		\tag*{(\textrm{EC}-1)}
			\sum_{s \in S} x_s &= 1 &
		\\
		\tag*{(\textrm{EC}-\textrm{IN})}
			 x_s &= \sum_{(r, s) \in E} x_{rs}	&& \forall s \in S
		\\
		\tag*{(\textrm{EC}-\textrm{OUT})}
			x_s &= \sum_{(s, t) \in E} x_{st}	&& \forall s \in S
		\\
		\tag*{(\textrm{EC}-\textrm{RAND})}
		\label{eq:A:Rand}
			x_{st} &= R(s)(t) \cdot x_s			&& \forall (s, t) \in E \textrm{ with } s \in S^R
		\\
		\tag*{(\textrm{EC}-\textrm{MP})}
		\label{eq:A:threshold}
			\sum_{(s, t) \in E} x_{st} \cdot w(s, t)[i] &\geq \vec \nu[i]	&& \forall (1 \leq i \leq d)
	\end{align}
	
	\caption{System of linear inequalities for the expectation problem inside an EC.}
	\label{fig:LP:exp}
\end{figure}

In the following, let $\Game = \tuple {\Graph, S^0, S^R, R}$
with $\Graph = \tuple {d, S, E, w}$
be a fixed \MDP.
W.l.o.g. we assume that $\Game$ is reduced to a single EC $S$.
Let $\vec \nu\in\Q^d$ be an expected-value achievable vector.
Consider the linear program \A~of Fig.~\ref{fig:LP:exp}.
(Cf. \cite{BrazdilBrozekChatterjeeForejtKucera:TwoViews:LMCS:2014}
for a similar linear program in the more general case where the \MDP is not just an EC.).
We use the linear program \A~to obtain the long-run ``frequencies'' of edges guaranteeing mean payoff $\vec \nu$.
For each state $s \in S$, we have a variable $x_s$
representing the long-run probability to be in $s$,
and, for each edge $(s, t) \in E$, we have a variable $x_{st}$
for the long-run probability of taking edge $(s, t)$.
\ignore{

	\begin{remark}
		We can check whether one can reach an expectation strictly above $\vec\nu$ in every component
		by introducing a new variable $y \geq 0$ and by modifying the last inequality as follows:
		\begin{align*}
			\sum_{(s, t) \in E} x_{st} \cdot w(s, t) \geq \vec \nu + (y, \dots, y) &\textrm{ [strict mean-payoff condition] }
		\end{align*}
		Let $y^*$ be the maximal $y$ s.t. the system above has a feasible solution.
		Then, one can achieve a mean payoff $> \vec\nu$ if, and only if, $y^* > 0$.
	\end{remark}

}
In the following, let
\begin{align*}
	\ExpSol \Game {s_0, f} &= \setof {\vec \nu \in \R^d} {\E {s_0, f} \Game \MP \geq \vec \nu} \ , \\
	\ExpSol \Game {s_0} &= \bigcup_{f \in \All \Game} \ExpSol \Game {s_0, f} \ .
\end{align*}
The lemma below shows that \A~has a non-negative solution if $\vec \nu$ is achievable in $\Game$.
Correctness follows directly from the analysis of \cite{BrazdilBrozekChatterjeeForejtKucera:TwoViews:LMCS:2014}.
The complexity of the solution follows from \citeappendix[Theorem 10.1]{Schrijver}.
In the statement below, recall that $W$ is the maximum absolute value of any weight in $\Game$,
and $Q$ is be the largest denominator of any probability appearing therein.
\begin{lemma}
	\label{lem:strategy:EC:LP}
	Let $\Game$ be a multidimensional mean payoff \MDP reduced to a single EC $S$,
	let $s_0 \in S$ be a state therein,
	and let $\vec \nu \in \ExpSol \Game {s_0}$ be an achievable value for the expectation.
	Then, the system \A~has a non-negative solution of size%
	\footnote{The size of a non-negative rational number $x = p/q$ with $p, q \in \N$ relatively prime, and $q > 0$,
	is the size of their bit representation, $1 + \lceil\log_2 (p + 1)\rceil + \lceil\log_2 (q + 1)\rceil$;
	cf.~\citeappendix[Section~3.2]{Schrijver}.}
	bounded by a polynomial in $W$ and $Q$.	
\end{lemma}

Let $\set{\tilde x_{st}}_{(s, t) \in E}, \set{\tilde x_s}_{s \in S}$ be a non-negative solution to \A~
of complexity polynomial in $W$ and $Q$.
This allows us to perform the following decomposition of $\Game$ into strongly connected components.
Let $S_{>0}$ be the set of states visited with (strictly) positive long-run average probability,
and let $E_{>0}$ be the set of edges visited with (strictly) positive long-run average probability:
\begin{align*}
	S_{>0} &= \setof {s \in S}{\tilde x_s > 0} \\
	E_{>0} &= \setof {(s, t) \in E}{\tilde x_{st} > 0}
\end{align*}
By the flow conditions of \A~, $E_{>0} = E \cap (S_{>0} \times S_{>0})$,
i.e., positive edges are exactly those connecting positive states.
States in $S_{>0}$ can be partitioned into maximal strongly connected components
$\set{S_1, \dots, S_k}$ (w.r.t. $E_{>0}$),
such that there is no positive edge between different components.
Since states in $S_i$ have at least one successor (as $S_i \subseteq S_{>0}$)
and all successors are in fact inside $S_i$,
we have that $S_i$ is an EC in $\Game$.
Let $E_i = E_{>0} \cap S_i \times S_i$ be the restriction of $E_{>0}$ to $S_i$,
let $\Graph_i = \tuple {d, S_i, E_i, w}$ be the corresponding graph,
and let $\Game_i = \tuple {\Graph_i, S^0_i, S^R_i, R}$ be the resulting \MDP,
where $S^0_i = S^0 \cap S_i$ and $S^R_i = S^R \cap S_i$.

For each $i \in \set{1, \dots, k}$, let $x_i > 0$ be the total long-run average probability of being in $S_i$,
and let $\vec\nu_i$ be the expected mean payoff vector achieved when starting from (anywhere) in component $S_i$:
\begin{align}
	\label{eq:x_i}
	x_i &= \sum_{s \in S_i} \tilde x_s > 0 \\
	\label{eq:nu_i}
	\vec\nu_i &= \frac 1 {x_i} \cdot \sum_{(s, t) \in E_i} \tilde x_{st} \cdot w(s, t)
\end{align}
Since component $S_i$ cannot be left,
it is reached with probability $x_i$,
and thus the expected mean payoff is
$\sum_{i=1}^k x_i \cdot \vec\nu_i \geq \vec\nu$.

\begin{remark}
	This analysis immediately yields a 2-memory \emph{randomized} strategy achieving expected mean payoff $\vec \nu$.
	Such a strategy goes to SCC $S_i$ with probability $x_i$,
	and then plays edge $(s, t)$ with probability $\tilde x_{st}/\tilde x_s$.
	Two memory states are required to discriminate the two phases.
	However, such a strategy is not unichain in general, which is what we aim at in this section.
\end{remark}

We design a randomized finite memory unichain strategy that plays edges $(s, t) \in E_{>0}$
with approximate long-run average frequency $\tilde x_{st}$,
in order to have mean payoff close to $\vec\nu$.
We do this in two steps.
First, in Sec.~\ref{sec:exp:local} we design, for each SCC $S_i$,
a randomized memoryless ``local strategy'' $g_i$
which plays edge $(s, t) \in E_i$
with long-run average frequency $\tilde x_{st} / x_i$ when started inside $S_i$.
Then, in Sec.~\ref{sec:exp:global} we combine those $g_i$'s into a global strategy
that spends in each $S_i$'s an approximate long-run fraction of time $x_i$.
By using larger amounts of memory, the error in this approximation can be made arbitrarily small.



\subsubsection{Inside a SCC $S_i$ (Strategy $g_i$)}

\label{sec:exp:local}

For each SCC $S_i$, let $g_i$ be the randomized memoryless strategy
that plays edge $(s, t) \in E_i$ with $s \in S_i^0$ with probability $\tilde x_{st} / \tilde x_s$.
Thus, in $\Game_i[g_i]$ edge $(s, t) \in E_i$ is visited for a long-run proportion of time $\tilde x_{st} / x_i$.
\begin{lemma}
	\label{lem:de-randomization:strategy}
	$\Game_i[g_i]$ is recurrent, and for every state $s_0 \in S_i$,
	\begin{align*}
		\E {s_0, g_i} {\Game_i} \MP = \vec\nu_i
	\end{align*}
\end{lemma}
In the following lemma, we show the mean payoff obtained by playing according to $g_i$ for sufficiently long time
is close to $\vec \nu_i$ with high probability.
The constant $L_0$ in the statement of the lemma does not depend on $S_i$,
thus the guarantee holds in every component $S_i$.
The lemma follows from a Hoeffding-style analysis.
\begin{lemma}
	\label{lem:MP:gi}
 	For every $\delta > 0$,
	there exists $L_0 = O(\frac 1 {\delta^2})\in \N$ and constants $a, b > 0$ s.t.,
	for every component $S_i$,
	for every $L \geq L_0$
	and for every state $s_0 \in S_i$,
	\begin{align*}
		\Prob {s_0, g_i} {\Game_i}
		{ \exists (1 \leq j \leq d) \cdot | (\MP_L - \vec\nu_i)[j] | \geq \delta}
		\leq \G L \delta := 2^d \cdot a \cdot e^{-b \cdot L \cdot \delta^2} \ .
	\end{align*}
\end{lemma}
\ignore{
	\begin{remark}
		$L_0$ \lorenzo{find a precise expression for $L_0$?}
		is exponential in the number of states of the \MDP $\Game$,
		polynomial in the maximum absolute weight $W$ and in the maximum denominator of probabilities $Q$;
		$a$ and $b$ are polynomial the number of states of $\Game$,
		and $a$ is exponential in $\delta$.
	\end{remark}
}
\begin{proof}
	The lemma follows from an application of Lemma~\ref{lem:Hoeffding:multidimensional} to each component $S_i$ separately,
	and then by aggregating the constants.
	More precisely, for each irreducible (and thus unichain) Markov chain $\Game_i[g_i]$, $1 \leq i \leq k$,
	Lemma~\ref{lem:Hoeffding:multidimensional} provides $L_i$ (called $K_0$ in the lemma) and constants $a_i, b_i > 0$ s.t.
	for every $L \geq L_i$ and state $s_0 \in S_i$,
	\begin{align*}
		\Prob {s_0, g_i} {\Game_i}
		{ \exists (1 \leq j \leq d) \cdot | (\MP_L - \vec\nu_i)[j] | \geq \delta}
		\leq \G L \delta := 2^d \cdot a_i \cdot e^{-b_i \cdot L \cdot \delta^2} \ .
	\end{align*}
	Just take $L_0 := \max \set{L_1, \dots, L_k}$,
	$a := \max \set{a_1, \dots, a_k}$,
	and $b := \min \set{b_1, \dots, b_k}$
	to satisfy the claim.
	\ignore{
	Each $L_i$'s is polynomial in the parameters of the Markov chain $\Game_i[g_i]$ and in the maximum absolute weight $W$
	(cf. Lemma~\ref{lem:Hoeffding:multidimensional}),
	and thus the same holds for $L_0$.
	%
	Moreover, $a$ and $b$ are polynomial in the number of states of $\Game$, $W$ and $Q$,
	and $a$ is exponential in $\delta$.
	}
\end{proof}

\subsubsection{Across SCCs (The global strategy)}

\label{sec:exp:global}

We now combine the local strategies $g_i$'s in order to achieve approximate expected mean payoff $\vec\mu$ with finite memory.
For each SCC $S_i$, let $h_i$ be a memoryless strategy ensuring that $S_i$ is reached almost surely
from any state in $\Game$.
The strategy $g_A$ is parametrized by a natural number $A > 0$.
Assume that $x_i$ is of the form $x_i = a_i / b_i$, with $a_i, b_i \in \N$ relatively prime, $b_i > 0$,
let $b = \lcm \set{b_1, \dots, b_k}$,
and let $c_i = b \cdot x_i$.
Note that $c_i$ is a natural number.
Intuitively, $g_A$ works in $k$ different stages.
In stage $i \in \set{1, \dots, k}$, $g_A$ does the following:
\begin{enumerate}
	
	\item[\a] Play $h_i$ to reach $S_i$ almost surely.
	
	\item[\b] Once in $S_i$, switch to strategy $g_i$ for $A \cdot c_i$ steps.
	Then, switch to stage $(i\mod k) + 1$ and go to \a.
	
\end{enumerate}
A full repetition of stages $\set{1, \dots, k}$ is called a \emph{phase}.
Intuitively, $g_A$ spends a proportion of time $x_i$ in $S_i$ in the limit,
and, while the game stays in $S_i$, $g_A$ plays according to $g_i$.
Recall that $g_i$ is memoryless.

\begin{remark}
	\label{rem:memory:fA}
	Strategy $g_A$ can be implemented with memory bounded by $k \cdot A \cdot b$. 
	%
	%
	%
%
%
%
%
		%
%
	%
	Notice that in both \a~and \b, $g_A$ plays according to a memoryless strategy,
	and no memory is needed to distinguish \a~from \b~since it suffices to look at the current state.
	Since the size of the binary representation of $b$ is polynomial in $W$ and $Q$ (cf. Lemma~\ref{lem:strategy:EC:LP}),
	strategy $g_A$ uses memory exponential in $W$ and $Q$, and linear in $n$ and $A$,
	where $n$ is the number of states of $\Game$.
\end{remark}

\begin{figure}
	\centering
	\footnotesize
	\begin{tikzpicture}

		\node (G1) [] {$G_1$};
		\node (g1) [below = .5cm of G1, align = center] {play $g_1$ \\ for $A \cdot c_1$ steps};
		
		\node (G2) [right = 3cm of G1] {$G_2$};
		\node (g2) [below = .5cm of G2, align = center] {play $g_2$ \\ for $A \cdot c_2$ steps};
		
		\node (G3) [right = 3cm of G2] {$G_3$};
		\node (g3) [below = .5cm of G3, align = center] {play $g_3$ \\ for $A \cdot c_3$ steps};
		
		\path[->] (g1) edge [bend left = 30] node [above = 1ex] {play $h_2$} (g2);
		\path[->] (g2) edge [bend left = 30] node [above = 1ex] {play $h_3$} (g3);
		\path[->] (g3) edge [bend left = 15] node [below = 1ex] {play $h_1$} (g1);
		
		\begin{pgfonlayer}{background}
			\node [fill=blue!20, rectangle, rounded corners = 12pt, fit=(G1) (g1)] {};
			\node [fill=blue!20, rectangle, rounded corners = 12pt, fit=(G2) (g2)] {};
			\node [fill=blue!20, rectangle, rounded corners = 12pt, fit=(G3) (g3)] {};
		\end{pgfonlayer}
		
	\end{tikzpicture}
	\caption{The global strategy $g_A$}
	\label{fig:global_strategy}	
\end{figure}
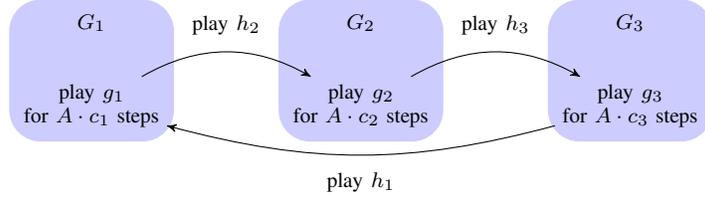

We show that for any additive error $\varepsilon > 0$,
we can play each stage sufficiently long (by increasing the parameter $A$)
s.t. the probability of deviating from the expected mean payoff $\vec\nu$ by more than $\varepsilon$
in any component is small.
\begin{lemma}
	\label{lem:strategy:EC:A}
	For any achievable vector $\vec\nu \in \ExpSol \Game {s_0}$ and $\varepsilon > 0$,
	there exists an $A_\varepsilon \in \N$ ($= O(\frac 1 \varepsilon)$)
	s.t., for every $A \geq A_\varepsilon$,
	$\left(\vec \nu - \vec \varepsilon\right) \in \ExpSol \Game {s_0, g_A}$.
\end{lemma}
\begin{proof}	

	We begin by analysing the expected mean payoff of $g_A$ over a single phase.
	Let $\vec e$ be the expected mean payoff of a single phase.
	We prove that for every $\varepsilon > 0$ there exists $A$ large enough s.t.
	strategy $g_A$ achieves at least expected mean payoff $\vec \nu - \vec\varepsilon$ over a single phase.
	
	Let $l$ be an upper bound on the expected length of periods of type \a.
	That a finite such $l$ exists can be seen as follows.
	Formally, let $H_{s, i}$ be the random variable that returns the \emph{first hitting time} of the set $S_i$ when starting from $s \in S$,
	that is the number of steps to reach $S_i$.
	%
	Let $l_i = \max_{s \in S} \E {s, h_i} \Game {H_{s, i}}$ be the worst expected first hitting time of $S_i$ from any state in $S$
	when playing according to the memoryless strategy $h_i$ (which reaches $S_i$ almost surely),
	and take $l = \sum_{i=1}^k l_i$.
	By the definition of $h_i$ and standard results about hitting times,
	the value $l$ is finite.
	
	The length of a period of type \b~at stage $i$ is $A \cdot c_i$,
	thus the total length of periods of type \b~over a single phase is $\sum_{i=1}^k A \cdot c_i = A \cdot b$.
	Since $l$ (computed above) is the expected length of periods of type \a~over a single phase,
	the expected length of a single phase is at most $A \cdot b + l$.
	
	Let $\vec e_{\b, i}$ be the expected mean payoff of a period of type \b~at stage $i$.
	Apply Lemma~\ref{lem:MP:gi} with $\delta := \varepsilon / 8 > 0$,
	and let $L_0$ as given by the lemma.
	Thus, for every component $S_i$,
	if we play $g_i$ for time $L \geq L_0$,
	we have the following lower bound on $\vec e_{\b, i}$:
	\begin{align*}
		\vec e_{\b, i} \geq (1 - \G L {\varepsilon/8}) \cdot (\vec\nu_i - \vec\varepsilon/8) +
			\G L {\varepsilon/8} \cdot (- \vec W)
	\end{align*}
	where $W$ is the largest absolute value of weights in $\Game$
	and $\vec W = (W, \dots, W) \in \R^d$.
	Moreover, since $\G L {\varepsilon/8} \to 0$ as $L \to \infty$,
	there exists an $L_* \geq L_0$ s.t., for every $L \geq L_*$,
	$(1 - \G L {\varepsilon/8}) \cdot (\vec\nu_i - \vec\varepsilon/8) + \G L {\varepsilon/8} \cdot (- \vec W)
	\geq \vec\nu_i - \vec\varepsilon/4$, and thus
	\begin{align}
		\label{ineq:e:b:i}
		\vec e_{\b, i} \geq \vec\nu_i - \vec\varepsilon/4
	\end{align}
	for every $L \geq L_*$.
	We derive a precise bound for $L_*$.
	\begin{align*}
		&
		(1 - \G L {\varepsilon/8}) \cdot (\vec\nu_i - \vec\varepsilon/8) + \G L {\varepsilon/8} \cdot (- \vec W)
			\geq \vec\nu_i - \vec\varepsilon/4
		\\ \textrm{if } \quad&
		\vec\nu_i - \vec\varepsilon/8 - \G L {\varepsilon/8} \cdot (\vec\nu_i - \vec\varepsilon/8 + \vec W)
			\geq \vec\nu_i - \vec\varepsilon/4
		\\ \textrm{if } \quad&
		\vec\varepsilon/8 - \G L {\varepsilon/8} \cdot (\vec\nu_i - \vec\varepsilon/8 + \vec W)
			\geq \vec 0
		\\ \textrm{if } \quad&
		\varepsilon/8 - \G L {\varepsilon/8} \cdot (\nu_i^\textrm{max} - \varepsilon/8 + W)
			\geq 0
		\\ \textrm{if } \quad&
		\G L {\varepsilon/8}
			\leq \frac \varepsilon {8\nu_i^\textrm{max} - \varepsilon + 8W}
		\\ \textrm{if } \quad&
		a \cdot 2^d \cdot e^{-b \cdot L \cdot \varepsilon^2/64}
			\leq \frac \varepsilon {8\nu_i^\textrm{max} - \varepsilon + 8W}
		\\ \textrm{if } \quad&
		e^{-b \cdot L \cdot \varepsilon^2/64}
			\leq \frac 1 {a \cdot 2^d} \cdot \frac \varepsilon {8\nu_i^\textrm{max} - \varepsilon + 8W}
		\\ \textrm{if } \quad&
		-b \cdot L \cdot \varepsilon^2/64
			\leq \ln \left( \frac 1 {a \cdot 2^d} \cdot \frac \varepsilon {8\nu_i^\textrm{max} - \varepsilon + 8W} \right)
		\\ \textrm{if } \quad&
		L
			\geq \frac {64} {b \cdot \varepsilon^2} \cdot \ln \frac {a \cdot 2^d \cdot (8\nu_i^\textrm{max} - \varepsilon + 8W)} \varepsilon = \\
			& = \frac {64} {b \cdot \varepsilon^2} \left(\ln a + d\cdot\ln 2 + \ln (8\nu_i^\textrm{max} - \varepsilon + 8W) - \ln \varepsilon \right)
	\end{align*}
	where $\nu_i^\textrm{max} = \max \set{\vec\nu_i[1], \dots, \vec\nu_i[d]}$ is the largest component of $\vec\nu_i$
	and $\G L {\varepsilon/8} = a \cdot 2^d \cdot e^{-b \cdot L \cdot \varepsilon^2/64}$.
	Thus, take $L_* := \frac {64} {b \cdot \varepsilon^2} \left(\ln a + d\cdot\ln 2 + \ln (8\nu_i^\textrm{max} - \varepsilon + 8W) - \ln \varepsilon \right)$.
	Notice that $L_* = O(\frac 1 \varepsilon)$ (since $a$ is exponential in $\varepsilon$).
	
	Therefore, we stay in stage $i$ at least $L_*$ number of steps,
	which implies that we should have $A \cdot c_i \geq L_*$
	for every $1 \leq i \leq k$.
	%
	%
	\begin{assumption}
		$A \geq A_0 := \max \set{L_*, L_0} = O(\frac 1 \varepsilon)$.
	\end{assumption}
	By the definition of $L_0$, $A_0$ is exponential in $n$ (the number of states of $\Game$),
	and polynomial in $W$ and $Q$.

	For a period of type \b~in stage $i$, the expected total payoff is $A \cdot c_i \cdot \vec e_{\b, i}$.
	The expected total payoff of all periods of type \a~during a single phase is at least $-\vec W \cdot l$.
	Thus,
	\begin{align*}
		\vec e &\geq
			\frac {A \cdot c_1 \cdot \vec e_{\b, 1} + \cdots + A \cdot c_k \cdot \vec e_{\b, k} -\vec W \cdot l}
				{A \cdot b + l} = \\
		&= \frac { x_1 \cdot \vec e_{\b, 1} + \cdots + x_k \cdot \vec e_{\b, k} } {1 + \frac l {A \cdot b}} -
			\frac {\vec W \cdot l} {A \cdot b + l}
	\end{align*}
 	First, we choose $A$ large enough s.t. (A) $\frac {\vec W \cdot l} {A \cdot b + l} \leq \frac {\vec\varepsilon} 2$:
	\begin{align*}
		\frac {\vec W \cdot l} {A \cdot b + l} \leq \frac {\vec\varepsilon} 2 \quad \textrm{ iff } \quad & 
			\frac {W \cdot l} {A \cdot b + l} \leq \frac {\varepsilon} 2 \\
		\textrm{ iff } \quad & 2 \cdot W \cdot l \leq \varepsilon \cdot (A \cdot b + l) = (\varepsilon \cdot b) \cdot A + \varepsilon \cdot l \\
		\textrm{ iff } \quad & A \geq \frac {l \cdot (2 \cdot W - \varepsilon)} {\varepsilon \cdot b}
	\end{align*}
	\begin{assumption}
		$A \geq A_1 := \frac {l \cdot (2 \cdot W - \varepsilon)} {\varepsilon \cdot b} = O(\frac 1 \varepsilon)$.
	\end{assumption}
	
	When $A \geq A_1$ we also have (B)
	$\vec e_{\b, i} / (1 + \frac l {A \cdot b}) \geq \vec \nu_i - \frac {\vec\varepsilon} 2$,
	for every $1 \leq i \leq k$.
	Indeed,
	\begin{align*}
		\vec e_{\b, i} / \left(1 + \frac l {A \cdot b}\right) \geq \vec \nu_i - \frac {\vec\varepsilon} 2
		\quad \textrm{ if } \quad & \left(\vec\nu_i - \frac {\vec\varepsilon} 4\right) / \left(1 + \frac l {A \cdot b}\right) \geq \vec \nu_i - \frac {\vec\varepsilon} 2 \\
		\quad \textrm{ iff } \quad & \left(\vec\nu_i - \frac {\vec\varepsilon} 4\right) \cdot A \cdot b \geq \left(\vec \nu_i - \frac {\vec\varepsilon} 2\right) \cdot (A \cdot b + l) \\
		\quad \textrm{ iff } \quad & \frac {\vec\varepsilon} 2 \cdot A \cdot b \geq l \cdot \left(\vec \nu_i - \frac {\vec\varepsilon} 2\right) \\
		\quad \textrm{ iff } \quad & \forall j \quad \frac \varepsilon 2 \cdot A \cdot b \geq l \cdot \left(\vec \nu_i[j] - \frac \varepsilon 2\right) \\
		\quad \textrm{ iff } \quad & \forall j \quad A \geq \frac {l \cdot \left(2\vec \nu_i[j] - \varepsilon \right)} {\varepsilon \cdot b} \\
		\quad \textrm{ if } \quad & A \geq A_1 = \frac {l \cdot (2 \cdot W - \varepsilon)} {\varepsilon \cdot b}
	\end{align*}
	where the first step follows from the inequality in Eq.~\ref{ineq:e:b:i},
	and the last step follows from
	$\frac {l \cdot (2 \cdot W - \varepsilon)} {\varepsilon \cdot b}
	\geq \frac {l \cdot \left(2\vec \nu_i[j] - \varepsilon \right)} {\varepsilon \cdot b}$,
	since, for an achievable vector $\vec\nu_i$,
	it must clearly be the case that $\vec\nu_i[j] \leq W$,
	the largest absolute value of any weight in the game.
	%
	%
	
	Thanks to the two assumptions above,
	we can derive the following bound on $\vec e$:
	\begin{align*}
		\vec e &\geq \frac { x_1 \cdot \vec e_{\b, 1} + \cdots + x_k \cdot \vec e_{\b, k} } {1 + \frac l {A \cdot b}} -
			\frac {\vec W \cdot l} {A \cdot b + l} \\
			&\geq \frac { x_1 \cdot \vec e_{\b, 1} + \cdots + x_k \cdot \vec e_{\b, k} } {1 + \frac l {A \cdot b}} - \frac {\vec\varepsilon} 2 \\
			&\geq x_1 \cdot \left(\vec \nu_1 - \frac{\vec\varepsilon} 2 \right) + \cdots + x_k \cdot \left(\vec \nu_k - \frac{\vec\varepsilon} 2 \right) - \frac {\vec\varepsilon} 2 \\
			&= x_1 \cdot \vec \nu_1 + \cdots + x_k \cdot \vec \nu_k - \vec\varepsilon \\
			&= \vec \nu - \vec\varepsilon
	\end{align*}
	where the second inequality follows from (A),
	and the third inequality follows from (B).
	We also use the property that $\sum_{i=1}^k x_i = 1$
	and $x_1 \cdot \vec \nu_1 + \cdots + x_k \cdot \vec \nu_k = \vec \nu$.

	Now that we have a lower bound on the expected mean payoff over a single phase,
	we can show that the expected mean payoff of $g_A$ over longer and longer prefixes of an infinite run (thus spanning many phases)
	converges to $\vec e$, from which the lemma follows.
	Notice that the expected mean payoff over $m$ phases is simply $m \cdot \vec e$.
	Let $\vec e_n$ be the expected mean payoff of $g_A$ in the first $n$ steps.
	Since in $n$ steps there are an expected number of $\left\lfloor \frac n {A\cdot b + l} \right\rfloor$ phases
	of expected length $A\cdot b + l$ each,
	we obtain
	\begin{align*}
		\vec e_n &\geq \frac {\left\lfloor \frac n {A\cdot b + l} \right\rfloor \cdot (A\cdot b + l) \cdot \vec e + \left(n-\left\lfloor \frac n {A\cdot b + l} \right\rfloor \cdot (A\cdot b + l)\right) \cdot (- \vec W)} n \\
		&\geq \frac {\left\lfloor \frac n {A\cdot b + l} \right\rfloor \cdot (A\cdot b + l) \cdot \vec e - (A\cdot b + l) \cdot \vec W} n \\
		&\geq \frac {(n - (A\cdot b + l)) \cdot \vec e - (A\cdot b + l) \cdot \vec W} n
	\end{align*}
	and thus $\liminf_n \vec e_n = \vec e$.
	To conclude, take $A_\varepsilon := \max \set{A_0, A_1} = O(\frac 1 \varepsilon)$,
	and the claim is satisfied for any $A \geq A_\varepsilon$.	
\end{proof}

\begin{lemma}
	\label{lem:gA:unichain}
	For every $A \geq n$, $\Game[g_A]$ is unichain.
\end{lemma}
\begin{proof}
	Each $\Game_i[g_i]$ is recurrent by Lemma~\ref{lem:de-randomization:strategy}.
	By the definition of $g_A$, $\Game[g_A]$ is obtained by staying in $\Game_i[g_i]$ for a certain number of steps
	and then going to $\Game_j[g_j]$ with $j = (i\mod k) + 1$ almost surely.
	If we stay in $\Game_i[g_i]$ for at least $n$ steps, then with positive probability we can visit every state therein.
	Thus, $\Game[g_A]$ is unichain.
\end{proof}

In the following, for every $\varepsilon > 0$,
we denote by $g_\varepsilon$ the strategy $g_{\max \set{n, A_\varepsilon}}$,
where $A_\varepsilon$ is the bound provided by Lemma~\ref{lem:strategy:EC:A}.

\begin{proof}[Proof of Lemma~\ref{lem:strategy:EC:finite-memory:expectation}]
	W.l.o.g. we assume that $\Game$ reduces to a single end-component $S$.
	Let $\vec\nu \in \ExpSolP \Game {s_0}$.
	There exists $\vec\nu^* \in \ExpSol \Game {s_0}$ s.t. $\vec\nu < \vec\nu^*$.
	Let $\varepsilon = \frac 1 2 \cdot \min_{i=1}^d (\vec\nu^*[i] - \vec\nu[i]) > 0$
	be the half of the minimal difference between $\vec\nu^*$ and $\vec\nu$ in any component.
	Take $g := g_\varepsilon$.
	Then, $(\vec \nu^* - \vec\varepsilon) \in \ExpSol \Game {s_0, g}$ by Lemma~\ref{lem:strategy:EC:A}.
	By the choice of $\varepsilon$, $\vec \nu < \vec \nu^* - \vec\varepsilon$,
	and thus $\vec \nu \in \ExpSolP \Game {s_0, g}$.
	Finally, $\Game[g]$ is unichain by Lemma~\ref{lem:gA:unichain}.
	%
\end{proof}

The next lemma will be used later to derive a bound on the number of steps $K$ that strategy $g_\varepsilon$ should be played for
in order to have a mean payoff close to $\vec \nu$ with high probability.
\begin{lemma}
	\label{lem:strategy:EC:bound:A:K}
	For any $\varepsilon > 0$,
	there exist $K_0 \in \N$ and constants $a, b > 0$ s.t.,
	for any $K \geq K_0$ and state $s_0$,
	\begin{align*}
		\Prob {s_0, g_{\varepsilon/2}} {\Game}
			{ \exists (1 \leq j \leq d) \cdot | (\MP_K - \vec\nu)[j] | \geq \varepsilon}
			\leq \G K {\varepsilon/2} := 2^d \cdot a \cdot e^{-b \cdot K \cdot \varepsilon^2/4} \ .
	\end{align*}
	%
\end{lemma}
%

\begin{proof}
	Fix an error $\varepsilon > 0$,
	and let $\vec \nu^* = \E {s_0, g_{\varepsilon/2}} \Game {\MP}$ the expected mean payoff when playing according to $g_{\varepsilon/2}$.
	By the choice of $g_{\varepsilon/2}$, $\vec\nu - \varepsilon/2 \leq \vec \nu^* \leq \vec\nu$.
	%
%
	%
	%
	%
	Thus, $| (\MP_K - \vec\nu)[j] | \geq \varepsilon$
	implies $| (\MP_K - \vec\nu^*)[j] | \geq \varepsilon/2$,
	and we have
	\begin{align}
		\label{eq:epsilon_inequality}
		\Prob {s_0, g_{\varepsilon/2}} {\Game}
			{ \exists (1 \leq j \leq d) \cdot | (\MP_K - \vec\nu)[j] | \geq \varepsilon }
		\leq
		\Prob {s_0, g_{\varepsilon/2}} {\Game}
			{ \exists (1 \leq j \leq d) \cdot | (\MP_K - \vec\nu^*)[j] | \geq \varepsilon/2 } \ .
	\end{align}
	Since $\Game[g_{\varepsilon/2}]$ is unichain by Lemma~\ref{lem:gA:unichain},
	we can apply Lemma~\ref{lem:Hoeffding:multidimensional},
	and, thus, there exist $K_0$ and constants $a, b > 0$ s.t.,
	for every state $s$ and $K \geq K_0$,
	\begin{align*}
		\Prob {s, g_{\varepsilon/2}} \Game
		{ \exists (1 \leq j \leq d) \cdot | \MP_K[j] - \vec\nu^*[j] | \geq \varepsilon/2} 
		\leq \G K {\varepsilon/2} \ ,
	\end{align*}
	from which the claim follows by Eq.~\ref{eq:epsilon_inequality}.
%
\end{proof}

\section{Beyond worst-case synthesis}

\label{app:BWC}

\ignore{

	We can assume w.l.o.g. that $\vec \mu = \vec 0$.
	Indeed, for any component $1 \leq i \leq d$,
	let $\vec \mu[i] = \frac {a_i} {b_i}$, with $a_i, b_i \in \Z$.
	We define a new \MDP $\Game' = \tuple {\Graph', S^0, S^R, R}$,
	where the new multi-weighted graph $\Graph' = \tuple {d, S, E, w'}$ is the same as $\Graph$,
	except that we apply the following affine transformation to the weights:
	For every component $1 \leq i \leq d$,
	$w'(e)[i] := w(e)[i] \cdot b_i - a_i$.
	Then, for every Controller's strategy $f$,
	\begin{align*}
		(\vec \mu, \vec \nu) \in \BWCSolP \Game {s, f} \quad \textrm{ iff } \quad
		(\vec 0, \vec \nu') \in \BWCSolP {\Game'} {s, f}
	\end{align*}
	where $\nu'[i] := \nu[i] \cdot b_i - a_i$ for every component $1 \leq 1 \leq d$.
	This follows immediately from the linearity of the expectation.

}

\subsection{Finite-memory synthesis}

\label{app:BWC:finite-memory}

\ignore{

	\begin{remark}
		\label{rem:order}
		We can assume w.l.o.g. that $\vec \mu < \vec \nu$ (in every component).
		Indeed, if $(\vec \mu; \vec \nu) \in \BWCSolP \Game s$,
		then by definition there exist $\vec\mu' > \vec \mu$ and $\vec\nu' > \vec \nu$
		s.t. $(\vec \mu'; \vec \nu') \in \BWCSolP \Game s$.
		Let $\vec \nu''$ be the component-wise maximum of $\vec \mu'$ and $\vec \nu'$.
		Clearly $\vec \mu < \vec \nu''$, and
		\begin{align*}
			(\vec \mu; \vec \nu) \in \BWCSolP \Game s
				\quad \textrm{ if, and only if, } \quad
			(\vec \mu; \vec \nu'') \in \BWCSolP \Game s
		\end{align*}
		since the expected value of the mean payoff is $> \vec\mu$ by definition.
	\end{remark}
}

In this section, we give full proofs for some statements from Sec.~\ref{sec:finite-memory}.
Fix a game $\Game$ and worst-case threshold $\vec \mu := \vec 0$.

\propFiniteMemWEC*

\begin{proof}
	By Proposition~\ref{prop:EC}, the set of states visited infinitely often by $f$ is an EC $U$ almost surely.
	By contradiction, assume that $U$ is not winning with some positive probability.
	Since $f$ is finite-memory, $\Game[f]$ is finite.
	Since $U$ is visited with positive probability,
	there exists a reachable bottom strongly connected component $B$ in $\Game[f]$ which projected to $\Game$ is a subset of $U$.
	Since $U$ is not winning, there exists a play in $B$ with mean payoff $\not> \vec 0$.
	By prefix independence of the mean-payoff value function,
	there exists a play in $\Game[f]$ with mean payoff $\not> \vec 0$,
	contradicting that $f$ is surely winning.
\end{proof}

\lemBWCsynthesisWEC*

\begin{remark}
	The statement of the lemma holds even with $h$ a \emph{pure} finite-memory strategy,
	by applying Remark~\ref{rem:pure_strategy} in the construction of strategy $\stratexp{\frac \varepsilon 2}$ below.
\end{remark}

\noindent
The rest of this section is devoted to the proof of Lemma~\ref{lem:BWC:synthesis:WEC}.
To simplify the notation, we assume w.l.o.g. that the MDP $\Game$ reduces to a single WEC $W = S$.
Therefore, $\Game = {\restrict \Game W}$.
In the following, let
\begin{align*}
	\WCSol \Game {s_0, f} &= \setof {\vec \mu \in \R^d} {\forall \pi \in \Outcome {s_0} f \cdot \MP(\pi) \geq \vec \mu} \ , \\
	\WCSol \Game {s_0} &= \bigcup_{f \in \All \Game} \WCSol \Game {s_0, f} \ .
\end{align*}
Recall that $\Outcome {s_0} f$ is the set of plays originating from state $s_0$ which are consistent with strategy $f$.
Since we are in a WEC, $\vec 0 \in \WCSolP \Game {s_0}$,
and thus $\WCSol \Game {s_0}$ contains a vector $\vec \mu > \vec 0$.
Moreover, since $\vec\nu \in \ExpSolP \Game {s_0}$ with $\vec \nu \geq 0$,
we can assume that $\vec \nu \in \ExpSol \Game {s_0}$ with $\vec \nu > 0$.
W.l.o.g. we further assume $\vec \mu < \vec \nu$.
We show that, for every $\delta > 0$ and $\varepsilon > 0$,
there exists a randomized finite-memory strategy $\stratcmb{\delta, \varepsilon}$
satisfying the BWC threshold $\geq(\vec \mu - \vec \delta; \vec \nu - \vec \varepsilon)$.
%
The construction of $\stratcmb{\delta, \varepsilon}$ relies on the existence of the following two strategies:
\begin{itemize}
\item
	Since $\vec \mu$ is guaranteed in the worst case,
	by Lemma~\ref{lem:strategy:worst-case},
	for every $\delta > 0$ there exists
	a pure finite-memory strategy $\stratwc\delta$
	s.t. $(\vec \mu - \vec\delta) \in \WCSol \Game {s, \stratwc\delta}$
	for every state $s$ in the WEC.
\item
	Similarly, since $\vec \nu$ is achievable in expectation from state $s_0$,
	by Lemma~\ref{lem:strategy:EC:finite-memory:expectation},
	for every $\varepsilon > 0$ there exists a randomized finite-memory strategy $\stratexp\varepsilon$
	s.t. $(\vec\nu - \vec\varepsilon) \in \ExpSol \Game {s, \stratexp\varepsilon}$
	for every state $s$ in the WEC.  
\end{itemize}
The strategy $\stratcmb{\delta,\varepsilon}$ is parameterised by two natural numbers $K, L \in \N$,
and it alternates between $\stratexp{\frac \varepsilon 2}$ and $\stratwc{\frac \delta 2}$
over periods of length $K$ and $L$, respectively:
\begin{enumerate}
	\item[\a] Play $\stratexp{\frac \varepsilon 2}$ for $K > 0$ steps
	and record in $\Sum \in \Z^d$ the current sum of the weights since the beginning of the period.
	\item[\b] If $\Sum \geq (\vec \mu - \vec \delta) \cdot K$,
	then go to \a.
	Otherwise, play $\stratwc{\frac \delta 2}$ for $L > 0$ steps, and then go to \a.
\end{enumerate}
\noindent
Every time a new period starts, the memory of the relevant strategy is reset.
It remains to determine the values of the parameters $K$ and $L$ to reach the desired accuracy.
The parameter $K$, which depends on $\varepsilon$,
controls the probability that a period is of type \a{} or \ab,
and thus the quality of the approximation of the expectation objective:
the larger the $K$, the higher the probability that a period is of type \a,
and thus the closer the expectation to $\vec\nu$.
The parameter $L$ (dependent on $K$ and $\delta$) controls the length of the recovery period \ab,
and thus the larger the $L$, the closer the worst case to $\vec\mu$.

We show that we can always choose $L$ s.t. the worst-case objective is satisfied.
\begin{lemma}
	\label{lem:combined:worst-case}
	For every $\delta, \varepsilon > 0$ and $K \in \N$,
	there exists $L \in \N$ s.t. $L = O(\frac 1 \delta)$ and
	$(\vec\mu-\vec\delta) \in \WCSol \Game {s, \stratcmb{\delta,\varepsilon}}$
	for every state $s$ in the EC.
\end{lemma}
\begin{proof}
	Let $m$ be the product of the size of the memory of $\stratwc{\frac\delta 2}$ and the number of states in $\Game$,
	and let $\mumin > 0$ be the smallest component of $\vec \mu$.
	W.l.o.g. we assume $\delta < \mumin$,
	since $\WCSol \Game {s, \stratcmb{\delta,\varepsilon}}$ is downward-closed.
	%
	Below, we derive an expression for $L$ for the worst-case objective $\geq \vec \mu-\vec\delta$ to be satisfied.
	Let $\pi$ be any $\stratcmb{\delta,\varepsilon}$-consistent play.
	We decompose $\pi = \rho_0\rho_1 \cdots$ according to periods \a{} or \ab.
	If $\rho_i$ has type \a,
	then $\TP(\rho_i) \geq (\vec \mu - \vec \delta) \cdot K$ directly from the definition of periods of type \a.
	If $\rho_i$ has type \ab,
	then 
	at the end of the \a{} part (of length $K$)
	the sum of weights is at least $-K\cdot W$ in every component.
	(Recall that $W$ is the largest absolute value of any weight in $\Game$.)
	Moreover, during the following \b{} part (of length $L$),
	we have that every time the same memory state of $\stratcmb{\delta,\varepsilon}$ and state of $\Game$ repeats,
	the mean payoff is at least $\vec\mu-\frac {\vec\delta} 2$,
	thus yielding a sum of weights which is at least $- m\cdot W + (L - m) \cdot (\mumin - \frac \delta 2)$ in every component.
	Thus, for every component $0 \leq j < k$, we have
	\begin{align*}
		\TP(\rho_i)[j] \geq -K\cdot W - m\cdot W + (L - m) \cdot \left(\mumin - \frac \delta 2\right) \ .
	\end{align*}
	%
	In order to have $\MP(\pi) \geq \vec\mu - \vec\delta$,
	it suffices to have $\MP(\rho_i) \geq \vec\mu - \vec\delta$
	for each period $i$ since periods are of uniformly bounded length,
	and thus $\TP(\rho_i) \geq (\vec \mu - \vec \delta) \cdot (K + L)$,
	since this length is at most $K+L$.
	It is easy to see that the following choice for $L$ satisfies the constraint above:
	\begin{align}
		\label{eq:L:def}
		L := \left\lceil \frac {2\cdot K \cdot (W + \mumin - \delta) + m \cdot (2 \cdot W + 2 \cdot \mumin - \delta) } \delta \right\rceil
	\end{align}
\end{proof}

\noindent
In the following, we consider $L$ as fixed by Eq.~\ref{eq:L:def}.
We show that one can always choose $K$ s.t. the expectation objective is satisfied.
We crucially use the fact that $L$ is linear in $K$,
and that the probability of periods of type $\ab$ can be made negligible for large $K$.
\begin{restatable}{lemma}{combinedexpectation}
	\label{lem:combined:expectation}
	For sufficiently small $\delta, \varepsilon > 0$,
	there exists $K \in \N$ s.t. $K = O(\frac 1 \varepsilon)$,
	and $(\vec\nu - \vec\varepsilon) \in \ExpSol \Game {s, \stratcmb{\delta,\varepsilon}}$
	for every state $s$ in the EC.
\end{restatable}
\begin{proof}

	Let $\vec E(K, s)$ be the expected mean payoff vector in the \MDP $\Game$
	when Controller is playing according to $\stratcmb{\delta,\varepsilon}$, starting from state $s$.
	We prove that there exists a $K = O(\frac 1 \varepsilon)$ s.t.
	\begin{align*}
		\vec E(K, s) \geq \vec \nu - \vec \varepsilon \ .
	\end{align*}
	Let $\vec E_{(a)}(K, s)$ and $\vec E_{(a)+(b)}(K, s)$ be the expected mean payoff of periods of type \a~and \ab, respectively.
	%
	%
	%
	%
	Let $p(K)$ be the probability of having a period of type \ab.
	The expected length of a period is
	$(1-p(K)) \cdot K + p(K) \cdot (K+L)$.
	Similarly, the expected total payoff of a period is
	$(1-p(K)) \cdot \vec E_{(a)}(K, s) \cdot K + p(K) \cdot \vec E_{(a)+(b)}(K, s) \cdot (K + L)$.
	We thus obtain the following expression on the right for the expected mean payoff over one period,
	where the equality to the expected mean payoff over the entire play is easy to show:
	\begin{align*}
		\vec E(K, s) = \frac
			{(1-p(K)) \cdot \vec E_{(a)}(K, s) \cdot K + p(K) \cdot \vec E_{(a)+(b)}(K, s) \cdot (K + L)}
			{(1-p(K)) \cdot K + p(K) \cdot (K+L)}
	\end{align*}
	By dividing by $(1-p(K)) \cdot K$, we obtain the following inequality:
	\begin{align*}
		\vec E(K, s) = \frac
			{\vec E_{(a)}(K, s) + \frac{p(K)}{1-p(K)} \cdot \vec E_{(a)+(b)}(K, s) \cdot \frac{K + L}{K}}
			{1 + \frac{p(K)}{1 - p(K)} \cdot \frac{K+L}{K}}
		\geq \vec\nu - \vec\varepsilon
	\end{align*}
	Since $\vec E_{(a)+(b)}(K, s) \geq \vec \mu - \vec \delta$ by the choice of $L$
	(cf. the proof of Lemma~\ref{lem:combined:worst-case}),
	and $\vec \mu - \vec \delta > \vec 0$ for sufficiently small $\delta$,
	it suffices to find $K$ s.t.
	\begin{align}
		\label{eq:meanpayoff:fcmb}
		\frac	{\vec E_{(a)}(K, s)}
			 	{1 + \frac{p(K)}{1 - p(K)} \cdot \frac{K+L}{K}}
		\geq \vec\nu - \vec\varepsilon
	\end{align}
	From a qualitative point of view,
	when taking the limit for $K \to \infty$ in Eq.~\ref{eq:meanpayoff:fcmb},
	\begin{itemize}
		\item $\frac{K + L}{K}$ tends to a constant, since $L$ is linear in $K$, and
		\item $p(K)$ tends to $0$.
	\end{itemize}
	Thus, $\lim_{K \to \infty} \vec E(K, s) = \lim_{K \to \infty} \vec E_{(a)}(K, s) \geq \vec \nu - \frac {\vec\varepsilon} 2$,
	where the inequality follows from the fact in periods of type \a{}
	we are playing according to $\stratexp{\frac \varepsilon 2}$,
	which guarantees expectation $\geq \vec\nu - \frac {\vec\varepsilon} 2$ in the long run.
	Consequently, for every $\varepsilon > 0$, there exists a $K$ s.t. $\vec E(K, s) \geq \vec \nu - \vec \varepsilon$, as required.
	

\ignore{

	\begin{figure}
	\centering
	\footnotesize
	\begin{tikzpicture}

		\node[playerstate] (s) [] {$s$};
		\node[randomstate] (t) [right = 1cm of s] {$t$};
		
		\path[->] (s) edge node [above] {$0$} (t);
		\path[->] (s) edge [loop above] node {$1$} ();

		\path[->] (t) edge [bend right = 45] node [above] {$\frac 1 3, -30$} (s);
		\path[->] (t) edge [bend left = 45] node [below] {$\frac 2 3, 30$} (s);
		
	\end{tikzpicture}
	\caption{Finite-memory strategies need to approximate both the worst-case and the expectation components of the BWC objective.}
	\label{fig:weakly_winning}	
\end{figure}

	\begin{remark}
		\label{rem:weakly_winning}
		Notice that the finite-memory strategy $\stratcmb{\delta, \varepsilon}$ that we have constructed
		achieves only BWC the threshold $(\vec\mu-\vec\delta; \vec\nu-\vec\epsilon)$,
		where both the worst-case and the expectation objectives are under-approximated.
		In general, this is the best that a finite-memory strategy can do in a WEC.
		Consider the example in Fig.~\ref{fig:weakly_winning},
		where the optimal worst-case is $1$ and expectation is $5$.
		Thus, with finite memory, one can achieve the BWC objective $(1-\varepsilon; 5-\varepsilon)$ for every $\varepsilon > 0$,
		and the trivial $(1; 1)$.
	
		However, neither $(1; 1+\varepsilon)$ nor $(-15+\varepsilon; 5)$ are achievable with finite memory, for any $\varepsilon > 0$.
		Indeed, any finite memory strategy plays edge $(s, t)$ with a fixed long-run frequency $0 \leq \alpha \leq 1$,
		and thus the BWC mean payoff is 
		\begin{align*}
			(\alpha\cdot\frac {0 - 30} 2 + (1-\alpha) \cdot 1; \alpha \cdot 5 + (1-\alpha) \cdot 1) = (-16\alpha + 1; 4\alpha + 1)
		\end{align*}
		To see why the threshold $(1; 1+\varepsilon)$ is not achievable,
		notice that $-16\cdot\alpha + 1 \geq 1$ only if $\alpha = 0$, for which the expectation is just $1$.
		This shows that in order to gain anything for the expectation, you need to sacrifice the worst-case (if only finite memory is allowed).
		To see why threshold $(-15+\varepsilon; 5)$ is not achievable,
		notice that $4\cdot\alpha + 1 \geq 5$ only if $\alpha = 1$, for which the worst case is $-15$.
		This shows that with finite memory optimal expectation is not reachable in general, even by losing as much as $(-16+\varepsilon)$ on the worst-case.
			%
			%

			%
	\end{remark}

}


	The rest of the proof shows that $K$ can be taken to be $O(\frac 1 \varepsilon)$.
	We proceed by an Hoeffding-style analysis.
	%
	%
	The following assumption can be satisfied for sufficiently small $\delta, \varepsilon > 0$,
	since $\vec \mu < \vec \nu$ was assumed throughout this section. 
	\begin{assumption}
		$\vec\nu - \vec\varepsilon \geq \vec\mu - \vec\delta \geq \vec 0$.
	\end{assumption}

	Lemma~\ref{lem:strategy:EC:bound:A:K} provides a constant $K_0$ and an upper bound $\G K {\frac \varepsilon 4}$
	on the probability that the mean payoff deviates from $\vec\nu$ by more than $\frac \varepsilon 2$ in any component
	when playing according to $\stratexp{\frac \varepsilon 2}$ for every $K \geq K_0$.
	Thus, $1 - \G K {\frac \varepsilon 4}$
	is a lower bound on the probability that the mean payoff in any component deviates from $\vec\nu$ by less than $\frac \varepsilon 2$.
	\begin{assumption}
		Let $K \geq K_0$.
	\end{assumption}
	We get the following lower bound on the expected mean payoff $\vec E_{(a)}(K, s)$ of a period of type \a:
	\begin{align*}
		\vec E_{(a)}(K, s) \geq (1 - \G K {\frac \varepsilon 4}) \cdot \left(\vec\nu-\frac {\vec\varepsilon} 2\right) +
		\G K {\frac \varepsilon 4} \cdot \vec\nu_0 
	\end{align*}
	where $\vec\nu_0 \geq \vec 0$ is a lower bound on the mean payoff of any period of type \a,
	Therefore, we obtain the following simpler bound:
	\begin{align}
		\label{eq:bound:Ea}
		\vec E_{(a)}(K, s) \geq (1 - \G K {\frac \varepsilon 4}) \cdot \left(\vec\nu-\frac {\vec\varepsilon} 2\right)
	\end{align}

	Moreover, we can use Lemma~\ref{lem:strategy:EC:bound:A:K} also to provide a bound on $p(K)$,
	the probability of a period of type \ab.
	If a period is of type \ab,
	then there exists a component $1 \leq j \leq d$ s.t.
	$\TP_K[j] \leq 0$,
	and thus $\MP_K[j] \leq 0$.
	Since the strategy $\stratexp{\frac \varepsilon 2}$ achieves expectation $\geq \vec\nu - \frac {\vec\varepsilon} 2 \geq \vec 0$,
	we have $|(\MP_K - \vec\nu)[j]| \geq \frac \varepsilon 2$ for some $1 \leq j \leq d$.
	The latter event happens with probability at most $\G K {\frac {\vec\varepsilon} 4}$ by Lemma~\ref{lem:strategy:EC:bound:A:K}.
	We thus obtain the following bound:
	\begin{align}
		\label{eq:bound:pK}
		p(K) \leq \G K {\frac \varepsilon 4}
	\end{align}

	By using Eq.~\ref{eq:bound:Ea} and \ref{eq:bound:pK} in Eq.~\ref{eq:meanpayoff:fcmb},
	it suffices to find a $K$ s.t.
	\begin{align*}
		\frac
			{(1 - \G K {\frac {\vec\varepsilon} 4}) \cdot (\vec\nu - \frac {\vec\varepsilon} 2)}
			{1 + \frac{\G K {\frac {\vec\varepsilon} 4}}{1 - \G K {\frac {\vec\varepsilon} 4}} \cdot \frac{K+L}{K}}
		\geq \vec\nu - \vec\varepsilon
	\end{align*}

	Let $\gamma$ be an upper bound on
	$\frac{\G K {\frac \varepsilon 4}}{1 - \G K {\frac \varepsilon 4}} \cdot \frac{K+L}{K}$.
	Then, it suffices to show
	\begin{align}
		\nonumber
		&\forall (1 \leq j \leq d) \qquad
		\frac
			{(1 - \G K {\frac \varepsilon 4}) \cdot (\vec\nu[j] - \frac \varepsilon 2)}
			{1 + \gamma}
		\geq \vec\nu[j] - \varepsilon
		\\
		\nonumber
		\textrm{if} \quad &\forall (1 \leq j \leq d) \qquad
		1 - \G K {\frac \varepsilon 4}
		\geq \frac {(\vec\nu[j] - \varepsilon) (1 + \gamma)} {\vec\nu[j] - \frac \varepsilon 2}
		\\
		\nonumber
		\textrm{if} \quad &\forall (1 \leq j \leq d) \qquad
		1 - \G K {\frac \varepsilon 4}
		\geq \frac {\vec\nu[j] - \varepsilon + \gamma \cdot (\vec\nu[j] - \varepsilon)} {\vec\nu[j] - \frac \varepsilon 2}
		\\
		\label{eq:second:bound}
		\textrm{if} \quad &\forall (1 \leq j \leq d) \qquad
		\G K {\frac \varepsilon 4}
		\leq \frac {\frac \varepsilon 2 - \gamma\cdot(\vec\nu[j] - \varepsilon)} {\vec\nu[j] - \frac \varepsilon 2}
	\end{align}
	Let $\numax = \max \set{\vec\nu[1], \dots, \vec\nu[d]}$ be the maximal component of $\vec\nu$.
	We show that with the following value of $\gamma$ we can satisfy all objectives.
	\begin{assumption}
		$\gamma := \frac 1 4 \cdot \frac \varepsilon {\numax - \varepsilon}$.
	\end{assumption}
	We now plug in the value for $\gamma$ into Eq.~\ref{eq:second:bound}, to obtain
	$\G K {\varepsilon / 4} \leq \frac {\frac \varepsilon 2 - \frac 1 4 \cdot \frac \varepsilon {\numax - \varepsilon} \cdot(\numax - \varepsilon)} {\numax - \frac \varepsilon 2}$,
	yielding the following bound on $\G K {\varepsilon / 4}$:
	\begin{align}
		\label{eq:master:bound:one}
		\G K {\varepsilon / 4} \leq \frac 1 2 \cdot \frac \varepsilon {2\numax - \varepsilon}
	\end{align}
	Before solving the inequality above for $K$,
	we go back one step discuss the constraint derived from the definition of $\gamma$.
	Recall that $\gamma$ should satisfy
	$\frac{\G K {\varepsilon / 4}}{1 - \G K {\varepsilon / 4}} \cdot \frac{K+L}{K} \leq \gamma$.
	By replacing $\gamma$ with its definition $\gamma := \frac 1 4 \cdot \frac \varepsilon {\numax - \varepsilon}$
	in the latter inequality,
	we obtain the following inequality:
	\begin{align}
		\frac{\G K {\frac \varepsilon 4}}{1 - \G K {\frac \varepsilon 4}} \cdot \frac{K+L}{K} \leq \frac 1 4 \cdot \frac \varepsilon {\numax - \varepsilon}
	\end{align}
	By the definition,
	$L = \left\lceil \frac {2\cdot K \cdot (W + \mumin - \delta) + m \cdot (2 \cdot W + 2 \cdot \mumin - \delta) } \delta \right\rceil$,
	and thus
	$L \leq \left\lceil \frac {2\cdot K \cdot (W + \mumin) + m \cdot (2 \cdot W + 2 \cdot \mumin)} \delta \right\rceil \leq \frac {2 \cdot (W + \mumin) \cdot (K + m)} \delta + 1$.
	Thus, $\frac {K + L} K \leq \frac {K + \frac {2 \cdot (W + \mumin) \cdot (K + m)} \delta + 1} K =
	1 + \frac {2(W+\mumin)} \delta + \frac 1 K \cdot \frac {2m(W + \mumin) + \delta} \delta$.
	In particular, for $K \geq \frac {2m(W + \mumin) + \delta} \delta$, we have
	$\frac {K + L} K \leq 2 + \frac {2(W+\mumin)} \delta$.
	\begin{assumption}
		$K \geq K_1 := \max \set {K_0, \frac {2m(W + \mumin) + \delta} \delta}$.
	\end{assumption}
	Thus, for $K \geq K_1$, it suffices to satisfy the following inequality:
	\begin{align}
		\frac{\G K {\frac \varepsilon 4}}{1 - \G K {\frac \varepsilon 4}} \cdot \left( 2 + \frac {2(W+\mumin)} \delta \right) \leq \frac 1 4 \cdot \frac \varepsilon {\numax - \varepsilon}
	\end{align}
	We perform the following algebraic manipulations:
	\begin{align}
		\notag
		&&\frac{\G K {\frac \varepsilon 4}}{1 - \G K {\frac \varepsilon 4}} \cdot \frac 2 \delta \cdot (W + \mumin + \delta)
			& \leq \frac 1 4 \cdot \frac \varepsilon {\numax - \varepsilon}
		\\
		\notag
		\textrm{ if } \qquad &&\G K {\frac \varepsilon 4}
			& \leq \frac 1 8 \cdot \frac \delta {W + \mumin + \delta} \cdot \frac \varepsilon {\numax - \varepsilon} \cdot (1 - \G K {\frac \varepsilon 4})
		\\
		\notag		
		\textrm{ if } \qquad &&\G K {\frac \varepsilon 4} \left(1 + \frac 1 8 \cdot \frac \delta {W + \mumin + \delta} \cdot \frac \varepsilon {\numax - \varepsilon} \right)
				& \leq \frac 1 8 \cdot \frac \delta {W + \mumin + \delta} \cdot \frac \varepsilon {\numax - \varepsilon}
		\\
		\notag
		\textrm{ if } \qquad &&\G K {\frac \varepsilon 4} \left(8 (W + \mumin + \delta)(\numax - \varepsilon) + \delta\varepsilon \right)
				& \leq \delta\varepsilon
		\\
		\notag
		\textrm{ if } \qquad &&\G K {\frac \varepsilon 4}
				& \leq \frac {\delta \varepsilon} {8 (W + \mumin + \delta)(\numax - \varepsilon) + \delta\varepsilon}
		\\
		\notag
		\textrm{ if } \qquad &&\G K {\frac \varepsilon 4}
				& \leq \frac {\delta \varepsilon} { (8\mumin + 8W + 8\delta) \cdot \numax - (8\mumin + 8W + 7\delta) \cdot \varepsilon}
		\\
		\label{eq:master:bound:two}
		\textrm{ if } \qquad &&\G K {\frac \varepsilon 4}
				& \leq \frac {\varepsilon} { (\frac {8\mumin} \delta + \frac {8W} \delta + 8) \cdot \numax - (\frac {8\mumin} \delta + \frac {8W} \delta + 7) \cdot \varepsilon}
	\end{align}
	We now compare the bounds on $\G K {\frac \varepsilon 4}$ given by Eq.~\ref{eq:master:bound:one} and \ref{eq:master:bound:two},
	and we show that, for sufficiently small $\varepsilon$,
	the latter bound implies the former.
	Therefore, we seek for $\varepsilon$ s.t.
	\begin{align*}
		&&\frac {\varepsilon} { (\frac {8\mumin} \delta + \frac {8W} \delta + 8) \cdot \numax - (\frac {8\mumin} \delta + \frac {8W} \delta + 7) \cdot \varepsilon}
		& \leq \frac 1 2 \cdot \frac \varepsilon {2\numax - \varepsilon}
		\\
		\textrm{ if } \qquad && \left(\frac {8\mumin} \delta + \frac {8W} \delta + 8\right) \cdot \numax - \left(\frac {8\mumin} \delta + \frac {8W} \delta + 7\right) \cdot \varepsilon
		& \geq 4 \cdot \numax - 2 \cdot \varepsilon
		\\
		\textrm{ if } \qquad && \left(\frac {8\mumin} \delta + \frac {8W} \delta + 4\right) \cdot \numax
		& \geq \left(\frac {8\mumin} \delta + \frac {8W} \delta + 5\right)\cdot\varepsilon
		\\
		\textrm{ if } \qquad && \varepsilon \leq \frac {8(\mumin + W) + 4\delta} {8(\mumin + W) + 5\delta} \cdot \numax
	\end{align*}
	%
	%
	\begin{assumption}
		$\varepsilon \leq \varepsilon_0 := \frac {8(\mumin + W) + 4\delta} {8(\mumin + W) + 5\delta} \cdot \numax$.
	\end{assumption}
	Therefore, for sufficiently small $\varepsilon \leq \varepsilon_0 < \numax$ it suffices to solve Eq.~\ref{eq:master:bound:two}.
	Recall that $\G K {\frac \varepsilon 4} = a \cdot 2^d \cdot e^{-b \cdot K \cdot \varepsilon^2/16}$.
	Thus, we seek for $K$ s.t.
	\begin{align}
		\notag
		&& a \cdot 2^d \cdot e^{-b \cdot K \cdot \varepsilon^2/16}
		&\leq \frac {\varepsilon} { (\frac {8\mumin} \delta + \frac {8W} \delta + 8) \cdot \numax - (\frac {8\mumin} \delta + \frac {8W} \delta + 7) \cdot \varepsilon}
		\\
		\notag
		\textrm{ if } \qquad && e^{-b \cdot K \cdot \varepsilon^2/16}
		&\leq \frac 1 {a \cdot 2^d} \cdot \frac {\varepsilon} { (\frac {8\mumin} \delta + \frac {8W} \delta + 8) \cdot \numax - (\frac {8\mumin} \delta + \frac {8W} \delta + 7) \cdot \varepsilon}
		\\
		\notag
		\textrm{ if } \qquad && -b \cdot K \cdot \varepsilon^2/16
		&\leq \ln \left(\frac 1 {a \cdot 2^d} \cdot \frac {\varepsilon} { (\frac {8\mumin} \delta + \frac {8W} \delta + 8) \cdot \numax - (\frac {8\mumin} \delta + \frac {8W} \delta + 7) \cdot \varepsilon}\right)
		\\
		\notag
		\textrm{ if } \qquad && K
		&\geq \frac {16} {b \cdot \varepsilon^2} \cdot \ln \left(a \cdot 2^d \cdot \frac {(\frac {8\mumin} \delta + \frac {8W} \delta + 8) \cdot \numax - (\frac {8\mumin} \delta + \frac {8W} \delta + 7) \cdot \varepsilon} \varepsilon \right)
		\\
		\notag
		\textrm{ if } \qquad && K
		&\geq \frac {16} {b \cdot \varepsilon^2} \cdot \ln \left(a \cdot 2^d \cdot \frac {(\frac {8\mumin} \delta + \frac {8W} \delta + 8) \cdot \numax} \varepsilon \right)
		\\
		\label{eq:bound:K}
		\textrm{ if } \qquad && K
		&\geq \frac {16} {b \cdot \varepsilon^2} \cdot \left( \ln a + d \ln 2 + \ln \left(\frac {8\mumin} \delta + \frac {8W} \delta + 8\right) + \ln \numax - \ln \varepsilon \right) =: K_2
	\end{align}
	The constant $K_2$ is of the form $\frac \alpha {\varepsilon^2}$,
	with $\alpha$ linear in $\varepsilon$ and $d$ (recall that $a$ is exponential in $\varepsilon$),
	and polynomial in the characteristics of the Markov chain.
	Thus, $K = O(\frac 1 \varepsilon)$.
\end{proof}

\subsection{Infinite-memory synthesis}

\subsubsection{Inside an EC}

We complete the proof of Lemma~\ref{lemma:infinite:memory:exp} by showing the two claims.
The first claim relies on Lemma~3.9 in \citeappendix{BrazdilBrozekEtessamiKucera:IC2013},
while the second claim relies on Lemma~\ref{lem:strategy:EC:bound:A:K}.

\claimUpperBoundA*

\begin{proof}
	Recall that $p_{i, K}$ is, by definition,
	the probability that the total payoff goes below $\vec N_i$ in any component during phase $i$.
	Since at the beginning of phase $i$ the total payoff is $\TP_{K \cdot i} > 2 \cdot \vec N_i$ by definition,
	$p_{i, K}$ is upper bounded by the probability that the total payoff decreases by at least $\vec N_i$ in any component:
	%
	%
	\begin{align*}
		p_{i, K} \leq \Prob {s_0} {\Game[f_K]}
			{\exists (K \cdot i \leq h < K \cdot (i+1)) \cdot \TP_h - \TP_{K \cdot i} \not> -\vec N_i}
	\end{align*}
	For any state $s$ in the EC,
	and payoff threshold $N \in \N$, 
	let $p_{s, N, K}$ be the probability that the total payoff goes below $-N$ in any component
	when starting from $\vec 0$ within at most $K$ steps:
	\begin{align*}
		p_{s, N, K}		&= \Prob {s} {\Game[f_K]} {\exists (0 \leq h < K) \cdot \exists (1 \leq j \leq d) \cdot \TP_h[j] \leq -N}
	\end{align*}
	Thus, $p_{i, K} \leq \min_{s \in S} p_{s, \min \vec N_i, K}$.
	We prove the claim thanks to the following lemma.
	(It follows from Lemma~3.9 in \citeappendix{BrazdilBrozekEtessamiKucera:IC2013}. For completeness, we provide a proof below.)
	\begin{restatable}{lemma}{lemEtessamiApprox}
		\label{lem:Etessami:approx}
		There exists a rational constant $c < 1$ and an integer $N^* \geq 0$ s.t.,
		for every $N \geq N^*$ and starting state $s$ in $\Game[f_K]$ s.t. $\Prob s {\Game[f_K]} {\MP \geq \vec 0} = 1$,
		\begin{align*}
			p_{s, N, K} \leq 2^d \cdot \frac {c^N} {1-c}
		\end{align*}
	\end{restatable}
	\noindent
	By Lemma~\ref{lem:strategy:EC:finite-memory:expectation}
	we can assume that when playing according to $f^{exp}$
	we obtain a Markov chain $\Game[f^{exp}]$ which is \emph{unichain}.
	Since $f^{exp}$ has strictly positive expected mean payoff and $\Game[f^{exp}]$ is unichain,
	the mean payoff is almost surely strictly positive,
	and thus we can apply Lemma~\ref{lem:Etessami:approx} to the Markov chain $\Game[f^{exp}]$
	and obtain constants $N^*$ and $c$
	s.t. for every $N \geq N^*$, $K \geq 0$, and state $s$,
	$p_{s, N, K} \leq 2^d \frac {c^N} {1-c}$.
	Let $\nu^* = \min_{1 \leq j \leq d} \vec\nu[j]$,
	and let $K^*$ be large enough s.t., for any $K \geq K^*$,
	$N_i^{min} \geq N^*$.
	(For example, take $K^* = \frac {2 N^*} {\nu^*}$.)
	Assume we are in any phase $i \geq 1$.
	Since $N_i^{min} \geq N^*$ by the choice of $K^*$, by Lemma~\ref{lem:Etessami:approx} we obtain,
	for every $K \geq K^*$,
	$p_{s, N_i^{min}, K} \leq 2^d \cdot \frac {c^{N_i^{min}}} {1-c} = 2^d \cdot \frac {c^{\frac {\nu^* \cdot i \cdot K} 2}} {1-c}$.
	Since $p_{i, K} \leq \min_{s \in S} p_{s, \min \vec N_i, K}$,
	we have $p_{i, K} \leq 2^d \frac {c^{\frac {\nu^* \cdot i \cdot K} 2}} {1-c}$.
	The claim follows by taking $a = 2^d/(1-c)$ and $b = c^{\frac {\nu^*} 2}$, which is $ < 1$ since $c$ is.
\end{proof}
\begin{proof}[Proof of Lemma~\ref{lem:Etessami:approx}]

	We first prove the lemma in the unidimensional case.
	Let $T$ be the set of states from where the total payoff is almost surely $-\infty$,
	and let $T^*$ be the set of states which can reach $T$ with positive probability.
	Lemma 3.9 in \citeappendix{BrazdilBrozekEtessamiKucera:IC2013} when applied to the Markov chain $\Game[f_K]$
	(in general, it applies to Markov decision processes),
	guarantees that there exists a rational constant $c < 1$ and an integer $N^* \geq 0$ s.t. for every $N \geq N^*$
	\begin{align*}
		\Prob s {\Game[f_K]} {\exists i \in \N \cdot \TP_i \leq -N \textrm { and $T^*$ is not visited } } \leq \frac {c^N} {1-c}
	\end{align*}
	Since we assume $\Prob s {\Game[f_K]} {\MP \geq 0} = 1$,
	by Lemma~\ref{lem:MP:TP:basic2}
	the probability to have a total payoff equal to $-\infty$ is $0$,
	and thus $T$ is empty.
	Therefore, the probability on the left above is just
	$\Prob s {\Game[f_K]} {\exists i \in \N \cdot \TP_i \leq -N}$,
	which is the definition of $p_{s, N, K}$, and thus
	\begin{align*}
		p_{s, N, K}	\leq \frac {c^N} {1-c}
	\end{align*}
	
	Now let $d > 1$,
	and, for a component $1 \leq j \leq d$,
	let $p_{s, N, K}^j$ be the probability that we have a drop by $-N$ in the total payoff in component $j$,
	i.e.,
	\begin{align*}
		p_{s, N, K}^j = \Prob {s} {\Game[f_K]} {\exists i \in \N \cdot \TP_i[j] \leq -N}
	\end{align*}
	By looking at the Markov chain $\Game[f_K]$ projected to component $j$,
	we can apply the result from the first part and obtain that,
	for every fixed $j$,
	there exist $c_j < 1$ and $N_j^* \geq 0$
	s.t., for every $N \geq N_j^*$,
	\begin{align*}
		p_{s, N, K}^j \leq \frac {c_j^N} {1-c_j}
	\end{align*}
	By taking $N^* = \max_{1 \leq j \leq d} N_j^*$
	and letting $c$ be the $c_j$ maximizing $\frac {c_j} {1-c_j}$, we have, 
	for every $N \geq N^*$,
	\begin{align*}
		p_{s, N, K}^j \leq \frac {c^N} {1-c}
	\end{align*}
	By definition, $1-p_{s, N, K}^j$ is the probability that the total payoff never goes below $-N$ in component $j$,
	and thus $\Pi_{j=1}^d(1 - p_{s, N, K}^j)$ is a lower bound
	on the probability that the total payoff never goes below $-N$ in \emph{any component}.
	Therefore $p_{s, N, K} \leq 1 - \Pi_{j=1}^d(1 - p_{s, N, K}^j) \leq 1 - (1 - \frac {c^N} {1-c})^d$.
	By a simple calculation (cf. the proof of Lemma~\ref{lem:Hoeffding:multidimensional}),
	we derive $(1 - \frac {c^N} {1-c})^d \geq 1 - 2^d \cdot \frac {c^N} {1-c}$,
	and thus
	\begin{align*}
		p_{s, N, K} \leq 1 - \left(1 - 2^d \cdot \frac {c^N} {1-c}\right) = 2^d \cdot \frac {c^N} {1-c}
	\end{align*}
	as required.
\end{proof}

\claimUpperBoundB*

\begin{proof}
	Recall the definition of $q_{i, K} = \Prob {s_0} {\Game[f_K]} {\TP_{K \cdot (i+1)} \not> 2 \cdot \vec N_{i+1}}$.
	By the definition of $\vec N_{i+1} = \frac {\vec\nu (i+1) K} 2$,
	\begin{align*}
		q_{i, K}=\Prob {s_0} {\Game[f_K]} {\MP_{K \cdot (i+1)} \not> \vec \nu}
	\end{align*}
	Let $\vec\nu^*$ be the expected mean payoff of $f^{exp}$, 
	and let $\delta = \min_{1 \leq j \leq d} (\vec\nu^*[j] - \vec\nu[j]) > 0$.
	Then, $q_{i, K}$ is upper bounded by the probability that the mean payoff deviates from its expected value $\vec\nu^*$ by more than $\delta$ in any component.
	By Lemma~\ref{lem:strategy:EC:bound:A:K}, this is upper bounded by $a' 2^d e^{-b'(i+1)K\delta^2}$
	for sufficiently large $(i+1)K$, and thus for sufficiently large $K$.
	Take $a = a' \cdot 2^d$ and $b = e^{-b\delta^2} < 1$.
	%
\end{proof}

\subsection{Inside a general MDPs}

In this section, we prove the correctness of the reduction of the BWC problem to solving system $T'$.

\lemInfiniteMemoryGeneral*

The following proposition is the analogous of Proposition~\ref{prop:BWC:strategy:all:MWEC}
where we consider MEC (instead of MWEC) and arbitrary strategies (instead of finite-memory ones).
As before, it rests on the assumption that all states of $\Game$ are reachable from $s_0$.
\begin{proposition}
	\label{prop:BWC:strategy:all:MEC}
	Let $\Game$ be a pruned multidimensional mean-payoff \MDP.
	If there exists a strategy $h$ s.t. $(\vec 0; \vec \nu) \in \BWCSolP \Game {s_0, h}$,
	then there exists a strategy $h^*$ with the same property,
	and such that, for every MEC $U$,
	the set of states visited infinitely often by $h^*$ is a subset of $U$ with positive probability.
\end{proposition}

The following proposition (and its proof) is analogous to Proposition~\ref{prop:finite-memory:general:1:analysis},
where MWECs are replaced by MECs.
\begin{proposition}
	\label{prop:infinite-memory:general:1:analysis}
	If $T'$ has a non-negative solution,
	then there exists a (finite-memory) strategy $\hat h$ s.t.
	\begin{enumerate}
		\item $\vec \nu \in \ExpSolP \Game {s_0, \hat h}$.
		\item For every MEC $U$, there is a probability $y^*_U > 0$
		s.t. the set of states visited infinitely often by $\hat h$ is inside $U$ with probability $y^*_U$.
		\item The set of states visited infinitely often by $\hat h$ is almost surely an EC.
		Consequently, $\sum_{\textrm{MEC } U} y^*_U = 1$.
		\item Once inside a MEC $U$, $\hat h$ achieves expected mean payoff $\vec \nu_U > \vec 0$.
		\item $\sum_{\textrm{MEC } U} y^*_U \cdot \vec \nu_U > \vec \nu$.
	\end{enumerate}
\end{proposition}

\begin{proof}[Proof of Lemma~\ref{lem:infinite-memory:general}]
	
	The proof is similar to finite-memory case in Lemma~\ref{lem:finite-memory:general}.
	We sketch here the crucial differences.
	
	For one direction, let $h$ be a strategy s.t. $(\vec 0; \vec \nu) \in \BWCSolP \Game {s_0, h}$.
	By Proposition~\ref{prop:BWC:strategy:all:MEC}, there exists a strategy $h^*$
	that additionally visits each MEC infinitely often with positive probability.
	By the construction in Proposition 4.4 of \cite{BrazdilBrozekChatterjeeForejtKucera:TwoViews:LMCS:2014}
	applied to strategy $h^*$,
	we obtain a solution to system $T'$.
	Equations~\ref{eq:y:flow}, \ref{eq:y:random}, \ref{eq:y:sum}, \ref{eq:x:flow}, \ref{eq:x:random}, and \ref{eq:MP:global}
	are shown to be satisfied in the proof of Proposition 4.4.
	Eq.~\ref{eq:x:y'} is satisfied since, by Proposition~\ref{prop:EC},
	$h^*$ is eventually trapped in an EC almost surely (not necessarily a WEC).
	Finally, Eq.~\ref{eq:MP:local'} is satisfied: 
	By construction, $h^*$ visits each MEC infinitely often with positive probability.
	For every MWEC $U$ there exist $s, t \in U$ s.t. $x^*_{st} > 0$.
	Since $h^*$ is winning for the worst-case,
	it achieves an expected mean payoff $> \vec 0$ in $U$,
	and thus Eq.~\ref{eq:MP:local'} is satisfied.
	
	For the other direction, we use Proposition~\ref{prop:infinite-memory:general:1:analysis} to obtain strategy $\hat h$,
	and then we proceed as in the second part of the proof of Lemma~\ref{lem:finite-memory:general}
	by replacing MWEC with MEC, and by using Lemma~\ref{lem:BWC:synthesis:EC} instead of Lemma~\ref{lem:BWC:synthesis:WEC}
	in the construction of strategies $h_U$'s.
\end{proof}

\bibliographystyleappendix{IEEEtran}
\bibliographyappendix{bibliography}


\end{document}